\documentclass[11pt,a4paper]{article}

\usepackage[pdftex,plainpages=false,colorlinks=false,hyperindex,bookmarksopen]{hyperref}

\usepackage[english]{babel}
\usepackage{comment}
\usepackage{amssymb,amsmath,amsfonts,amsthm}

\usepackage{tikz}
\usepackage{graphicx}
\usepackage{listings}
\usepackage{subfigure}
\usepackage{pifont}
\usepackage{tikz-cd}
\usepackage{rotating}
\usepackage{tabularx,colortbl}
\usepackage[a4paper]{geometry}
\usepackage{datetime}
\usepackage{todonotes}
\usepackage{xspace}
\usepackage[defblank]{paralist}
\usepackage{enumerate}
\usepackage[scaled=.85]{newtxtt} 
\usepackage[linesnumberedhidden,ruled,vlined]{algorithm2e}
\usepackage{booktabs}
\usepackage{multirow}
\usepackage{wrapfig}
\usepackage{arydshln}

\usepackage[T1]{fontenc}
\usepackage[latin1]{inputenc}
\usepackage{a4wide}

\usepackage{fancybox}
\usepackage{fancyvrb}

\usepackage{lipsum}
\makeatletter
\if@titlepage
{\if@twocolumn\else\endquotation\fi}
\fi
\makeatother

\usepackage{fancyhdr}
\usetikzlibrary{arrows,shapes,snakes,automata,backgrounds,petri,positioning,shadows,matrix,decorations.pathmorphing,fit,calc,backgrounds}

\newcommand{\ua}{\ensuremath{\underline a}}

\newcommand{\ue}{\ensuremath{\underline e}}
\newcommand{\ui}{\ensuremath{\underline i}}
\newcommand{\uu}{\ensuremath{\underline u}}
\newcommand{\uv}{\ensuremath{\underline v}}
\newcommand{\ux}{\ensuremath{\underline x}}
\newcommand{\uy}{\ensuremath{\underline y}}
\newcommand{\uz}{\ensuremath{\underline z}}
\newcommand{\uF}{\ensuremath{\underline F}}
\newcommand{\uG}{\ensuremath{\underline G}}
\newcommand{\uS}{\ensuremath{\underline S}}

\newcommand{\PSPACE}{\textsc{Pspace}\xspace}

\newcommand{\cA}{\ensuremath \mathcal A}
\newcommand{\cC}{\ensuremath \mathcal C}
\newcommand{\cB}{\ensuremath \mathcal B}
\newcommand{\cM}{\ensuremath \mathcal M}
\newcommand{\cN}{\ensuremath \mathcal N}
\newcommand{\cS}{\ensuremath \mathcal S}

\newcommand{\lra}{\longrightarrow}

\newcommand{\safe}{\texttt{SAFE}\xspace}
\newcommand{\unsafe}{\texttt{UNSAFE}\xspace}
\newcommand{\mcmt}{\textsc{mcmt}\xspace}

\renewcommand{\int}{\ensuremath {\mathcal I}}

\newcommand{\domain}[1]{\mathit{dom}(#1)}

\newcommand{\sas}{SAS\xspace}
\newcommand{\SAS}{Simple Artifact System\xspace}

\newcommand{\ras}{RAS\xspace}
\newcommand{\RAS}{Relational Artifact System\xspace}

\newcommand{\type}[1]{\ensuremath{\mathsf{#1}}\xspace}
\newcommand{\funct}[1]{\ensuremath{\mathit{#1}}\xspace}
\newcommand{\relname}[1]{\ensuremath{\mathit{#1}}\xspace}

\newcommand{\constant}[1]{\texttt{#1}}

\newcommand{\sorts}[1]{#1_{\mathit{srt}}}
\newcommand{\functs}[1]{#1_{\mathit{fun}}}
\newcommand{\vals}[1]{#1_{\mathit{val}}}
\newcommand{\ids}[1]{#1_{\mathit{ids}}}
\newcommand{\ext}[1]{#1_{\mathit{ext}}}

\newcommand{\nullv}{\texttt{undef}}

\newcommand{\lentry}[2]{\ensuremath{#1\,{:}\,#2}}

\definecolor{deepblue}{HTML}{0C3B80}
\definecolor{deepgreen}{HTML}{2EA601}
\definecolor{lightOrange}{HTML}{FFA03C}
\definecolor{darkOrange}{HTML}{F1800A}
\definecolor{lightBlue}{HTML}{0174CD}
\definecolor{greenF}{HTML}{2CBB5C}
\definecolor{cyan}{HTML}{86A6D5}

\tikzstyle{sortnode} = [
  rectangle,
  rounded corners=5pt,
  minimum width=18mm,
  minimum height=6mm,
  very thick,
]

\tikzstyle{functnode} = [
  rectangle,
  minimum width=1cm,
  minimum height=5mm,
]

\tikzstyle{idnode} = [
  sortnode,
  draw=deepblue,
  fill=cyan!50,
  text=deepblue,  
]

\tikzstyle{valnode} = [
  sortnode,
  densely dashed,
  draw=violet,
  fill=magenta!50,
  text=violet,
]

\tikzstyle{f} = [
  thick,
]

\tikzstyle{fd} = [
  f,
  -angle 90,
]

\tikzstyle{relation}=[rectangle split, rectangle split parts=#1, rectangle split part align=base, draw, anchor=center, align=center, text height=3mm, font=\bfseries, ultra thick, text centered]

\newcommand{\smtlambda}[2]{
  \lambda #1.\left(#2\right)
}
\newcommand{\smtif}[4]{
  \begin{array}[#1]{@{}l@{}}
  \mathsf{if~}#2\mathsf{~then~}#3
  \\\mathsf{else~}#4
  \end{array}
}

\newcommand{\smtifinline}[3]{
  \mathsf{if~}#1\mathsf{~then~}#2\mathsf{~else~}#3
}

\newcommand{\typedvar}[2]{#1{:}#2}

\newcommand{\artvar}[1]{\mathit{#1}}

\newcommand{\hrsig}{\Sigma_{\mathit{hr}}}
\newcommand{\hrras}{\S_{\mathit{hr}}}
\newcommand{\userid}{\type{UserId}}
\newcommand{\jobcatid}{\type{JobCatId}}
\newcommand{\empid}{\type{EmpId}}
\newcommand{\compinid}{\type{CompInId}}
\newcommand{\stringval}{\type{String}}

\newcommand{\username}{\funct{userName}}
\newcommand{\empname}{\funct{empName}}
\newcommand{\compemp}{\funct{who}}
\newcommand{\compjob}{\funct{what}}
\newcommand{\jobdescr}{\funct{jobCatDescr}}

\newcommand{\pstatev}{\artvar{pState}}
\newcommand{\uidv}{\artvar{uId}}

\newcommand{\jid}{\artvar{jId}}
\newcommand{\dval}{\artvar{pubDate}}
\newcommand{\ostate}{\artvar{pubState}}
\newcommand{\recstate}{\artvar{aState}}
\newcommand{\uid}{\artvar{uId}}
\newcommand{\eid}{\artvar{eId}}
\newcommand{\cid}{\artvar{cId}}

\newcommand{\enab}{\constant{enabled}}
\newcommand{\publishing}{\constant{publishing}}
\newcommand{\published}{\constant{published}}
\newcommand{\joopen}{\constant{open}}
\newcommand{\joclosed}{\constant{closed}}
\newcommand{\joselected}{\constant{joSelected}}
\newcommand{\appreceived}{\constant{received}}
\newcommand{\win}{\constant{winner}}
\newcommand{\los}{\constant{loser}}
\newcommand{\final}{\constant{final}}
\newcommand{\notified}{\constant{notified}}

\newcommand{\joidx}{\type{joIndex}}
\newcommand{\jocat}{\funct{joCat}}
\newcommand{\jodate}{\funct{joPDate}}
\newcommand{\jostate}{\funct{joState}}

\newcommand{\appidx}{\type{appIndex}}
\newcommand{\appjob}{\funct{appJobCat}}
\newcommand{\appuser}{\funct{applicant}}
\newcommand{\appresp}{\funct{appResp}}
\newcommand{\appscore}{\funct{appScore}}
\newcommand{\appres}{\funct{appResult}}
\newcommand{\scoreval}{\type{Score}}

\tikzstyle{sortnode} = [
  rectangle,
  rounded corners=5pt,
  minimum width=18mm,
  minimum height=6mm,
  very thick,
]

\tikzstyle{functnode} = [
  rectangle,
  minimum width=1cm,
  minimum height=5mm,
]

\tikzstyle{idnode} = [
  sortnode,
  draw=deepblue,
  fill=cyan!50,
  text=deepblue,  
]

\tikzstyle{artnode} = [
  sortnode,
  draw=deepblue,
  fill=yellow!50,
  text=deepblue,  
]

\tikzstyle{valnode} = [
  sortnode,
  densely dashed,
  draw=violet,
  fill=magenta!50,
  text=violet,
]

\tikzstyle{f} = [
  thick,
]

\tikzstyle{fd} = [
  f,
  -angle 90,
]

\tikzstyle{relation}=[rectangle split, rectangle split parts=#1, rectangle split part align=base, draw, anchor=center, align=center, text height=3mm, font=\bfseries, text centered]

\newcommand{\I}{\ensuremath{\mathcal{I}}}

\newcommand{\M}{\ensuremath{\mathcal{M}}}

\renewcommand{\S}{\ensuremath{\mathcal{S}}}

\newcommand{\set}[1]{\{#1\}}                      

\newcommand{\tup}[1]{\langle #1\rangle}            

\newtheorem{lemma}{Lemma}[section]
\newtheorem{theorem}{Theorem}[section]
\newtheorem{corollary}[theorem]{Corollary}
\newtheorem{proposition}{Proposition}[section]

\newtheorem{definition}[theorem]{Definition}
\newtheorem{assumption}[theorem]{Assumption}

\theoremstyle{remark}
\newtheorem{remark}{Remark}[section]
\newtheorem{example}{Example}[section]

\begin{document}

\title{Verification of Data-Aware Processes via Array-Based Systems (Extended Version)}

\author{Diego Calvanese$^1$, Silvio Ghilardi$^2$, Alessandro Gianola$^1$,\\
 Marco Montali$^1$ and Andrey Rivkin$^1$\\
 \\
 $^1$ Free University of Bozen-Bolzano\\
 \texttt{\textit{surname}@inf.unibz.it}\\
 \\
 $^2$ Università degli Studi di Milano\\
 \texttt{silvio.ghilardi@unimi.it}}

\date{Wednesday 5th December, 2018}

\maketitle

\begin{abstract}
  We study verification over a general model of artifact-centric systems, to
  assess (parameterized) safety properties irrespectively of the initial
  database instance.  We view such artifact systems as array-based systems,
  which allows us to check safety by adapting backward reachability,
  establishing for the first time a correspondence with model checking based on
  Satisfiability-Modulo-Theories (SMT). To do so, we make use of the
  model-theoretic machinery of model completion, which surprisingly turns out
  to be an effective tool for verification of relational systems, and
  represents the main original contribution of this paper.  In this way, we
  pursue a twofold purpose.  On the one hand, we reconstruct (restricted to
  safety) the essence of some important decidability results obtained in the
  literature for artifact-centric systems, and we devise a genuinely novel
  class of decidable cases.  On the other, we are able to exploit SMT
  technology in implementations, building on the well-known MCMT model checker
  for array-based systems, and extending it to make all our foundational
  results fully operational.
\end{abstract}

\section{Introduction}

During the last two decades, a huge body of research has been dedicated to the
challenging problem of reconciling data and process
management within contemporary organizations~\cite{Rich10,Dum11,Reic12}. This
requires to move from a purely control-flow understanding of business processes
to a more holistic approach that also considers how data are manipulated and
evolved by the process.
Striving for this integration, new models were
devised, with two
prominent representatives: object-centric processes \cite{KuWR11}, and business
artifacts \cite{Hull08,DaHV11}.

In parallel, a flourishing series of results has been dedicated to the formalization of such integrated models, and on the boundaries of decidability and complexity for their static analysis and  verification~\cite{CaDM13}. Such results are quite fragmented, since they consider a variety of different assumptions on the model and on the static analysis tasks \cite{Vian09,CaDM13}.
Two main trends can be identified within this line. A recent series of results
focuses on very general data-aware processes that evolve a full-fledged,
relational database (DB) with arbitrary first-order constraints
\cite{BeLP12,BCDDM13,AAAF16,CDMP17}. Actions amount to full bulk updates that
may simultaneously operate on multiple tuples at once, possibly injecting fresh
values taken from an infinite data domain.  Verification is studied by fixing
the initial instance of the DB, and by considering all possible
evolutions induced by the process over the initial data.

A second trend of research
is instead focused on the formalization and verification of artifact-centric
processes. These systems are traditionally formalized using three components
\cite{DHPV09,DaDV12}:
\begin{inparaenum}[\itshape (i)]
\item a read-only DB that stores fixed, background
  information,
\item a working memory that stores the evolving state of artifacts, and
\item actions that update the working memory.
\end{inparaenum}
Different variants of this model, obtained via a careful tuning of the relative
expressive power of its three components, have been studied towards
decidability of verification problems parameterized over the read-only DB (see,
e.g., \cite{DHPV09,DaDV12,boj,DeLV16}). These are verification problems where a
property is checked for every possible configuration of the read-only DB.

The overarching goal of this work is to connect, for the first time, such
formal models
and their corresponding verification problems on the one hand, with the models
and techniques of \emph{model checking via Satisfiability-Modulo-Theories
 (SMT)} on the other hand.  This is concretized through four technical
contributions.

Our \emph{first contribution} is the definition of a general framework of
so-called \emph{\RAS{s}} (\ras{s}), in which artifacts are formalized in the
spirit of \emph{array-based systems}, one of the most sophisticated setting
within the SMT tradition. In this setting, \sas{s} are a particular class 
of \ras{s}, where only artifact variables are allowed.
``Array-based systems'' is an umbrella term generically referring to
infinite-state transition systems implicitly specified using a declarative,
logic-based formalism.  The formalism captures transitions manipulating arrays
via logical formulae, and its precise definition depends on the specific
application of interest.  The first declarative formalism for array-based
systems was introduced in~\cite{ijcar08,lmcs} to handle the verification of
distributed systems, and afterwards was successfully employed also to verify a
wide range of infinite-state systems~\cite{FI,fmsd}.
Distributed systems are parameterized in their essence: the number $N$ of
interacting processes within a distributed system is unbounded, and the
challenge is that of supplying certifications that are valid for all possible
values of the parameter $N$. The overall state of the system is typically
described by means of arrays indexed by process identifiers, and used to store
the content of process variables like locations and clocks. These arrays are
genuine \emph{second order function} variables: they map indexes to elements,
in a way that changes as the system evolves. \emph{Quantifiers} are then used
to represent sets of system states.
%
%
\ras{s} employ arrays to capture a very rich working memory that simultaneously
accounts for artifact variables storing single data elements, and full-fledged
artifact relations storing unboundedly many tuples. Each artifact relation is
captured using a collection of arrays, so that a tuple in the relation can be
retrieved by inspecting the content of the arrays with a given index. The
elements stored therein may be fresh values injected into the \ras, or data
elements extracted from the read-only DB, whose relations are subject to key
and foreign key constraints. This constitutes a big leap from the usual applications of array-based systems, because the nature of such constraints is quite different and requires
completely new  techniques for handling them (for instance, for quantifier elimination, see below). To attack this complexity,
by relying on array-based systems, \ras{s} encode the read-only DB using a
functional, algebraic view, where relations and constraints are captured using
multiple sorts and unary functions.  The resulting model captures the essential
aspects of the model in \cite{verifas}, which in turn is tightly related
(though incomparable) to the sophisticated formal model for artifact-centric
systems of \cite{DeLV16}.

Our \emph{second contribution} is the development of \emph{algorithmic
 techniques} for the verification of \emph{(parameterized) safety} properties
over \ras{s}, which amounts to determine whether there exists an instance of
the read-only DB that allows the \ras to evolve from its initial configuration
to an \emph{undesired} one that falsifies a given state property. To attack
this problem, we build on \emph{backward reachability}~\cite{ijcar08,lmcs}, one
of the most well-established techniques for safety verification in array-based
systems.  This is a correct, possibly non-terminating technique that
\emph{regresses} the system from the undesired configuration to those
configurations that reach the undesired one.  This is done by iteratively
computing symbolic pre-images, until they either intersect the initial
configuration of the system (witnessing unsafety), or they form a fixpoint that
does not contain the initial state (witnessing safety).

Adapting backward reachability to the case of \ras{s}, by retaining soundness
and completeness, requires genuinely novel research so as to eliminate new
(existentially quantified) ``data'' variables introduced during regression.
Traditionally, this is done
by quantifier instantiation or elimination.  However, while quantifier
instantiation
can be transposed to \ras{s}, quantifier elimination cannot, since the data
elements contained in the arrays point to the content of a full-fledged DB with
constraints.  To reconstruct quantifier elimination in this setting, which is
the main technical contribution of this work, we employ the classic
model-theoretic machinery of \emph{model completions}~\cite{mma}: via model
completions, we prove that the runs of a \ras can be faithfully lifted to
richer contexts
where quantifier elimination is indeed available, despite the fact that it was
not available in the original
structures.  This allows us to recast safety problems over \ras{s} into
equivalent safety problems in this richer setting.

Our \emph{third contribution} is the identification of three notable classes of
\ras{s} for which backward reachability terminates, in turn witnessing
decidability of safety. The first class restricts the working memory to
variables only, i.e., focuses on \sas. The second class focuses on \ras operating under the restrictions imposed in~\cite{verifas}:
it requires acyclicity of foreign keys and ensures a sort of locality principle
where different artifact tuples are not compared.  Consequently, it
reconstructs the decidability result exploited in~\cite{verifas} if one
restricts the verification logic used there to safety properties only.  In
addition, our second class supports full-fledged bulk updates, which greatly increase
the expressive power of dynamic systems \cite{SchmitzS13} and, in our setting,
witness the incomparability of our results and those in \cite{verifas}.
The third class is genuinely novel, and while it further restricts foreign keys
to form a tree-shaped structure, it does not impose any restriction on the
shape of updates, and consequently supports not only bulk updates, but also
comparisons between artifact tuples.

Our \emph{fourth contribution} concerns the implementation of backward
reachability techniques for \ras{s}. Specifically, we have extended the
well-known \textsc{mcmt} model checker for array-based systems~\cite{mcmt},
obtaining a fully operational counterpart to all the foundational results
presented in the paper. Even though implementation and experimental evaluation
are not central in this paper,  we note that our model checker correctly handles the
examples produced to test \textsc{verifas}~\cite{verifas}, as well as
additional examples that go beyond the verification capabilities of
\textsc{verifas}, and report some interesting case here.
 The performance of
\textsc{mcmt} to conduct verification of these examples is very encouraging,
and indeed provides the first stepping stone towards effective, SMT-based
verification techniques for artifact-centric systems.

\section{Preliminaries}
\label{sec:preliminaries}

We adopt the usual first-order syntactic notions of signature, term, atom,
(ground) formula, and so on.  We use $\uu$ to represent a tuple
$\tup{u_1,\ldots,u_n}$.  Our signatures $\Sigma$ are multi-sorted and include
equality for every sort, which implies that variables are sorted as well.
Depending on the context, we keep the sort of a variable implicit, or we
indicate explicitly in a formula that variable $x$ has sort $S$ by employing
notation $x:S$.  The notation $t(\ux)$, $\phi(\ux)$ means that the term $t$,
the formula $\phi$ has free variables included in the tuple $\ux$.
Constants and function symbols $f$ have
 \emph{sources} $\uS$ and a \emph{target} $S'$, denoted as
$f:\uS\longrightarrow S'$ (relation symbols $r$ only have sources $r:\uS$).  
We assume that terms and formulae are well-typed,
in the sense that the sorts of variables, constants, and 
relations, function
sources/targets match.  A formula is said to be \emph{universal} (resp.,
\emph{existential}) if it has the form $\forall \ux\, (\phi(\ux))$ (resp.,
$\exists \ux\, (\phi(\ux))$), where $\phi$ is a quantifier-free
formula. Formulae with no free variables are called \emph{sentences}.

From the semantic side, we use the standard notions of a
\emph{$\Sigma$-structure} $\cM$ and of \emph{truth} of a formula in a
$\Sigma$-structure under an assignment to the free variables.
A \emph{$\Sigma$-theory} $T$ is a set of $\Sigma$-sentences; a \emph{model} of
$T$ is a $\Sigma$-structure $\cM$ where all sentences in $T$ are true.  We use
the standard notation $T\models \phi$ to say that $\phi$ is true in all models
of $T$ for every assignment to the free variables of $\phi$.  We say that
$\phi$ is \emph{$T$-satisfiable} iff there is a model $\cM$ of $T$ and an
assignment to the free variables of $\phi$ that make $\phi$ true in $\cM$.


In the following (cf.~Section~\ref{sec:artifact}) we specify transitions of an
artifact-centric system using first-order formulae.  To obtain a more compact
representation, we make use there of definable extensions as a means for
introducing so-called \emph{case-defined functions}.  We fix a signature
$\Sigma$ and a $\Sigma$-theory $T$; a \emph{$T$-partition} is a finite set
$\kappa_1(\ux), \dots, \kappa_n(\ux)$ of quantifier-free formulae
such that $T\models \forall \ux \bigvee_{i=1}^n \kappa_i(\ux)$ and
$T\models \bigwedge_{i\not=j}\forall \ux \neg (\kappa_i(\ux)\wedge
\kappa_j(\ux))$.  Given such a $T$-partition
$\kappa_1(\ux), \dots, \kappa_n(\ux)$ together with $\Sigma$-terms
$t_1(\ux), \dots, t_n(\ux)$ (all of the same target sort), a
\emph{case-definable extension} is the $\Sigma'$-theory $T'$, where
$\Sigma'=\Sigma\cup\{F\}$, with $F$ a ``fresh'' function symbol (i.e.,
$F\not\in\Sigma$)\footnote{Arity and source/target sorts for $F$ can be
 deduced from the context (considering that everything is well-typed).}, and
$T'=T \cup\bigcup_{i=1}^n \{\forall\ux\; (\kappa_i(\ux) \to F(\ux) =
t_i(\ux))\}$.
Intuitively, $F$ represents a case-defined function, which can be reformulated
using nested if-then-else expressions and can be written as
$ F(\ux) ~:=~ \mathtt{case~of}~ \{\kappa_1(\ux):t_1;\cdots;\kappa_n(\ux):t_n\}.
$ By abuse of notation, we identify $T$ with any of its case-definable
extensions $T'$.  In fact, it is easy to produce from a $\Sigma'$-formula
$\phi'$ a $\Sigma$-formula $\phi$ equivalent to $\phi'$ in all models of $T'$:
just remove (in the appropriate order) every occurrence $F(\uv)$ of the new
symbol $F$ in an atomic formula $A$, by replacing $A$ with
$\bigvee_{i=1}^n (\kappa_i(\uv) \land A(t_i(\uv)))$.
We also exploit $\lambda$-abstractions (see, e.g., formula~\eqref{eq:trans1}
below) for a more compact (still first-order) representation of some complex
expressions, and always use them in atoms like $b = \lambda y. F(y,\uz)$ as
abbreviations of $\forall y.~b(y)=F(y,\uz)$ (where, typically, $F$ is a symbol
introduced in a case-defined extension as above).

\section{Read-only Database Schemas}
\label{sec:readonly}

We now provide a formal definition of (read-only) DB-schemas by relying on an
algebraic, functional characterization, and derive some key model-theoretic
properties.

\begin{definition}\label{def:db}
  A \emph{DB schema} is a pair $\tup{\Sigma,T}$, where:
  \begin{inparaenum}[\itshape (i)]
  \item $\Sigma$ is a \emph{DB signature}, that is, a finite multi-sorted
    signature whose only symbols are equality, unary functions, and constants;
  \item $T$ is a \emph{DB theory}, that is, a set of universal
    $\Sigma$-sentences.
  \end{inparaenum}
\end{definition}
Next, we refer to a DB schema simply through its (DB) signature $\Sigma$ and
(DB) theory $T$, and denote by $\sorts{\Sigma}$ the set of sorts and by
$\functs{\Sigma}$ the set of functions in $\Sigma$.
Since $\Sigma$ contains only unary function symbols and equality, all atomic
$\Sigma$-formulae are of the form $t_1(v_1)=t_2(v_2)$, where $t_1$, $t_2$ are
possibly complex terms, and $v_1$, $v_2$ are either variables or constants.

\begin{remark}\label{rem:extdb}

If desired, we can freely extend DB schemas by adding arbitrary $n$-ary relation symbols to the signature $\Sigma$. For this purpose, we give the
following definition.

\begin{definition}\label{def:extdb}
  A \emph{DB extended-schema} is a pair $\tup{\Sigma,T}$, where:
  \begin{inparaenum}[\itshape (i)]
  \item $\Sigma$ is a \emph{DB extended-signature}, that is, a finite multi-sorted
    signature whose only symbols are equality, $n$-ary relations, unary functions, and constants;
  \item $T$ is a \emph{DB extended-theory}, that is, a set of universal
    $\Sigma$-sentences.
  \end{inparaenum}
\end{definition}

Since for our application we are only interested in relations with primary and foreign key dependencies (even if our implementation takes into account also the case of ``free'' relations, i.e. without key dependencies),
we restrict our focus on DB schemas, which are sufficient to capture those constraints (as explained in the following subsection).
We notice that, in case Assumption~\ref{ass}
discussed below holds for DB extended-theories, all the results presented in Section~\ref{sec:artifact} (and Theorem~\ref{thm:basic}) still hold even considering DB extended-schemas instead of DB schemas.
\end{remark}

We associate to a DB signature $\Sigma$ a characteristic graph $G(\Sigma)$
capturing the dependencies induced by functions over sorts.\footnote{The same definition can be adopted also for extended
DB signatures (relation symbols do not play a role in it).}
Specifically,
$G(\Sigma)$ is an edge-labeled graph whose set of nodes is $\sorts{\Sigma}$,
and with a labeled edge $S \xrightarrow{f} S'$ for each $f:S\longrightarrow S'$
in $\functs{\Sigma}$.
We say that $\Sigma$ is \emph{acyclic} if $G(\Sigma)$ is so. The \emph{leaves}
of $\Sigma$ are the nodes of $G(\Sigma)$ without outgoing edges.  These
terminal sorts are divided in two subsets, respectively representing
\emph{unary relations} and \emph{value sorts}. Non-value sorts (i.e., unary
relations and non-leaf sorts) are called \emph{id sorts}, and are conceptually
used to represent (identifiers of) different kinds of objects. Value sorts,
instead, represent datatypes such as strings, numbers, clock values, etc. We
denote the set of id sorts in $\Sigma$ by $\ids{\Sigma}$, and that of value
sorts by $\vals{\Sigma}$, hence
$\sorts{\Sigma} = \ids{\Sigma}\uplus\vals{\Sigma}$.

We now consider extensional data.
\begin{definition}
  \label{def:instance}
  A \emph{DB instance} of DB schema $\tup{\Sigma,T}$ is a $\Sigma$-structure
  $\cM$ that is a model of $T$ and such that every id sort of $\Sigma$ is
  interpreted in $\cM$ on a \emph{finite} set.
\end{definition}
Contrast this to arbitrary \emph{models} of $T$, where no finiteness assumption
is made.  What may appear as not customary in Definition~\ref{def:instance} is
the fact that value sorts can be interpreted on infinite sets. This allows us,
at once, to reconstruct the classical notion of DB instance as a finite model
(since only finitely many values can be pointed from id sorts using functions),
at the same time supplying a potentially infinite set of fresh values to be
dynamically introduced in the working memory during the evolution of the
artifact system. More details on this will be given in
Section~\ref{sec:relational-view}.

We respectively denote by $S^\cM$, $f^\cM$, and $c^\cM$ the interpretation in
$\cM$ of the sort $S$ (this is a set), of the function symbol $f$ (this is a
set-theoretic function), and of the constant $c$ (this is an element of the
interpretation of the corresponding sort).  Obviously, $f^\cM$ and $c^\cM$ must
match the sorts in $\Sigma$. E.g., if $f$ has source $S$ and target $U$, then
$f^\cM$ has domain $S^\cM$ and range $U^\cM$.

\begin{figure*}[tbp]
  \centering
  \scalebox{.7}{
   \begin{tikzpicture}[node distance=3mm and 5mm]
     \node[idnode] (userid) {\userid};
     \node[functnode, right=of userid] (username) {\username};

     \node[idnode, above=of userid] (empid) {\empid};
     \node[functnode, right=of empid] (empname) {\empname};

     \node[idnode, above=7mm of empid] (compinid) {\compinid};
     \node[functnode, right=of compinid, yshift=-1mm] (compemp) {\compemp};
     \node[functnode, right=of compemp, yshift=2mm] (compjob) {\compjob};

     \node[idnode, above=7mm of compinid] (jobcatid) {\jobcatid};
     \node[functnode, right=of jobcatid] (jobcatdescr) {\jobdescr};

     \node[valnode, above right=0mm and 12mm of empname] (stringval) {\stringval};

     \draw[f]
       (userid) edge (username);
     \draw[fd,out=0,in=180]
       (username) edge ($(stringval.west)-(0,2mm)$);

     \draw[f]
       (empid) edge (empname);
     \draw[fd,out=0,in=180]
       (empname) edge (stringval);

     \draw[f,out=0,in=180]
       ($(compinid.east)+(0mm,-1mm)$) edge (compemp.west);
     \draw[fd]
       (compemp) |- ($(compemp.south)+(-3mm,-3mm)$) -| (empid);
     \draw[f,out=0,in=180]
       ($(compinid.east)+(0mm,1mm)$) edge (compjob);
     \draw[fd]
       (compjob) |- ($(compjob.north)+(-3mm,3mm)$) -| (jobcatid);

     \draw[f]
       (jobcatid) edge (jobcatdescr);
     \draw[fd,out=0,in=180]
       (jobcatdescr) edge ($(stringval.west)+(0,2mm)$);;

     \node[
       relation=2,
       rectangle split horizontal,
       rectangle split part fill={lightgray!50},
       right=8.2cm of userid,
       ultra thick,
     ] (user)
     {
       \nodepart{one} \underline{$\mathit{id}$} : \userid %
       \nodepart{two}  {\username} : \stringval
     };
     \node[left=0mm of user] {\relname{User}};

     \node[
       relation=2,
       rectangle split horizontal,
       rectangle split part fill={lightgray!50},
       above= of user.north west,
       anchor=south west,
       ultra thick
     ] (emp)
     {
       \nodepart{one} \underline{$\mathit{id}$} : \empid %
       \nodepart{two}  {\empname} : \stringval
     };
     \node[left=0mm of emp] {\relname{Employee}};

     \node[
       relation=3,
       rectangle split horizontal,
       rectangle split part fill={lightgray!50},
       above=7mm of emp.north west,
       anchor=south west,
       ultra thick
     ] (compin)
     {
       \nodepart{one} \underline{$\mathit{id}$} : \compinid %
       \nodepart{two}  {\compemp} : \empid
       \nodepart{three}  {\compjob} : \jobcatid
     };
     \node[left=0mm of compin] {\relname{CompetentIn}};

     \node[
       relation=2,
       rectangle split horizontal,
       rectangle split part fill={lightgray!50},
       above=7mm of compin.north west,
       anchor=south west,
       ultra thick
     ] (jobcat)
     {
       \underline{$\mathit{id}$} : \jobcatid %
       \nodepart{two}  {\jobdescr} : \stringval
     };
     \node[left=0mm of jobcat] {\relname{JobCategory}};

     \draw[fd,out=-90,in=90] (compin.two south) |-
      ($(compin.two south)-(3mm,3mm)$) -| (emp.one north);

     \draw[fd,out=90,in=-90] (compin.three north) |-
      ($(compin.two north)+(-3mm,3mm)$) -| (jobcat.one south);
   \end{tikzpicture}
  }
  \caption{On the left: characteristic graph of the human resources DB
   signature from Example~\ref{ex:hr}. On the right: relational view of the DB
   signature; each cell denotes an attribute with its type, underlined
   attributes denote primary keys, and directed edges capture foreign keys.}
  \label{fig:hr}
\end{figure*}
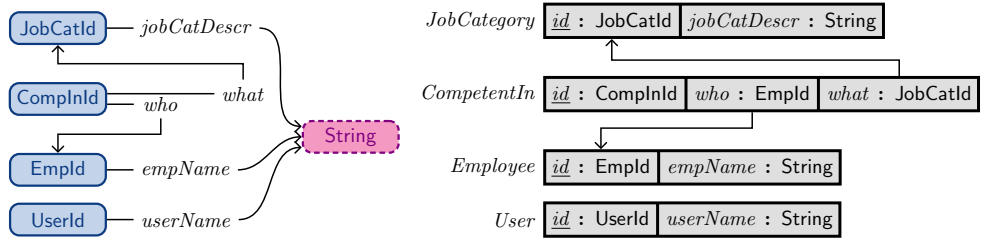

\begin{example}
  \label{ex:hr}
  The human resource (HR) branch of a company stores the following information
  inside a relational database:
  \begin{inparaenum}[\itshape (i)]
  \item users registered to the company website, who are potential job
    applicants;
  \item the different, available job categories;
  \item employees belonging to HR, together with the
    job categories they are competent in.
  \end{inparaenum}
  To formalize these different aspects, we make use of a DB signature $\hrsig$
  consisting of:
  \begin{inparaenum}[\itshape (i)]
  \item four id sorts, used to respectively identify users, employees, job
    categories, and the competence relationship connecting employees to job
    categories;
  \item one value sort containing strings used to name users and
    employees, and describe job categories.
  \end{inparaenum}
  In addition, $\hrsig$ contains five function symbols mapping:
  \begin{inparaenum}[\itshape (i)]
  \item user identifiers to their corresponding names;
  \item employee identifiers to their corresponding names;
  \item job category identifiers to their corresponding descriptions;
  \item competence identifiers to their corresponding employees and job
    categories.
  \end{inparaenum}
  The characteristic graph of $\hrsig$ is shown in Figure~\ref{fig:hr} (left
  part).
  \hfill\ensuremath{\triangleleft}
\end{example}

We close the formalization of DB schemas by discussing DB theories, whose role
is to encode background axioms.
We illustrate a typical background axiom, required to handle the possible
presence of \emph{undefined identifiers/values} in the different sorts. This
axiom is essential to capture artifact systems whose working memory is
initially undefined, in the style of~\cite{DeLV16,verifas}.
To specify an undefined value we add to every sort $S$ of $\Sigma$ a constant
$\nullv_S$ (written from now on, by abuse of notation, just as $\nullv$, used
also to indicate a tuple).  Then, for each function symbol $f$ of $\Sigma$, we
add the following axiom to the DB theory:
\begin{equation}
  \label{eq:null}
  \forall x~(x = \nullv \leftrightarrow f(x) = \nullv)
\end{equation}
This axiom states that the application of $f$ to the undefined value produces
an undefined value, and it is the only situation for which $f$ is undefined.

\begin{remark}
  In the artifact-centric model in the style of~\cite{DeLV16,verifas} that we
  intend to capture, the DB theory consists of Axioms~\eqref{eq:null} only.
  %
  However, our technical results do not require this specific choice, and more
  general sufficient conditions will be discussed later. These conditions apply to natural variants of Axiom~\eqref{eq:null} (such variants might be used to model situations where we would like to have for instance 
  many undefined values).
\end{remark}

\subsection{Relational View of DB Schemas}
\label{sec:relational-view}

We now clarify how the algebraic, functional characterization of DB schema and
instance can be actually reinterpreted in the classical, relational
model. Definition~\ref{def:db} naturally corresponds to the definition of
relational database schema equipped with single-attribute \emph{primary keys}
and \emph{foreign keys} (plus a reformulation of
constraint~\eqref{eq:null}). To technically explain the correspondence, we
adopt the \emph{named perspective}, where each relation schema is defined by a
signature containing a \emph{relation name} and a set of \emph{typed attribute
 names}. Let $\tup{\Sigma,T}$ be a DB schema. Each id sort $S \in \ids{\Sigma}$
corresponds to a dedicated relation $R_S$ with the following attributes:
\begin{inparaenum}[\itshape (i)]
\item one identifier attribute $id_S$ with type $S$;
\item one dedicated attribute $a_f$ with type $S'$ for every function symbol
  $f \in \functs{\Sigma}$ of the form $f: S \longrightarrow S'$.
\end{inparaenum}

The fact that $R_S$ is built starting from functions in $\Sigma$ naturally
induces different database dependencies in $R_S$. In particular, for each
non-id attribute $a_f$ of $R_S$, we get a \emph{functional dependency} from
$id_S$ to $a_f$; altogether, such dependencies in turn witness that
$\mathit{id}_S$ is the \emph{(primary) key} of $R_S$. In addition, for each
non-id attribute $a_f$ of $R_S$ whose corresponding function symbol $f$ has id
sort $S'$ as image, we get an \emph{inclusion dependency} from $a_f$ to the id
attribute $id_{S'}$ of $R_{S'}$; this captures that $a_f$ is a \emph{foreign
 key} referencing $R_{S'}$.

\begin{example}
  The diagram on the right in Figure~\ref{fig:hr} graphically depicts the
  relational view corresponding to the DB signature of Example~\ref{ex:hr}.
  \hfill\ensuremath{\triangleleft}
\end{example}
Given a DB instance $\M$ of $\tup{\Sigma,T}$, its corresponding
\emph{relational instance} $\I$ is the minimal set satisfying the following
property: for every id sort $S \in \ids{\Sigma}$, let $f_1,\ldots,f_n$ be all
functions in $\Sigma$ with domain $S$; then, for every identifier
$\constant{o} \in S^\M$, $\I$ contains a \emph{labeled fact} of the form
$R_S(\lentry{id_S}{\constant{o}^\M}, \lentry{a_{f_1}}{f_1^\M(\constant{o}^\M)},
\ldots, \lentry{a_{f_n}}{f_n^\M(\constant{o}^\M)})$.
With this interpretation, the \emph{active domain of $\I$} is the  set
\[
    \bigcup_{S \in \ids{\Sigma}} (S^\M \setminus \set{\nullv^\M}) 
    \cup
    \left\{
      \constant{v} \in \bigcup_{V \in \vals{\Sigma}} V^\M ~\left|~
        \begin{array}[c]{@{}l@{}}
          \constant{v}\neq\nullv^\M \text{ and there exist } f \in \functs{\Sigma}\\ \text{and }
          \constant{o}\in \domain{f^\M} \text{ s.t.~} f^\M(\constant{o}) =
          \constant{v}
        \end{array}
      \right.
      \!
    \right\}
\]
consisting of all (proper) identifiers assigned by $\M$ to id sorts, as well as
all values obtained in $\M$ via the application of some function. Since such
values are necessarily \emph{finitely many}, one may wonder why in
Definition~\ref{def:instance} we allow for interpreting value sorts over
infinite sets. The reason is that, in our framework, an evolving artifact
system may use such infinite provision to inject and manipulate new values into
the working memory.
From the definition of active domain above, exploiting Axioms~\eqref{eq:null} we get that the membership of a tuple
 $(x_0,\dots, x_n)$ to a generic $n+1$-ary relation $R_S$  with key dependencies
(corresponding to an id sort $S$)  can be expressed in our setting by using just unary function symbols and equality:

 \begin{equation}\label{eq:relations}
  R_S(x_0, \dots, x_n) \mbox{ iff } x_0 \neq \nullv \land x_1=f_1(x_0) \land \dots \land x_n = f_n(x_0) 
  \end{equation}
  
  Hence, the representation of negated atoms is the one that directly follows from negating~\eqref{eq:relations}:

 \begin{equation}\label{eq:not-relations}
 \neg R_S(x_0,\dots,x_n) \mbox{ iff } x_0 = \nullv \lor x_1 \neq f1(x_0) \lor \dots \lor x_n \neq f_n(x_0)
 \end{equation}

This relational interpretation of DB schemas exactly reconstructs the
requirements posed by~\cite{DeLV16,verifas} on the schema of the
\emph{read-only} database:
\begin{inparaenum}[\itshape (i)]
\item each relation schema has a single-attribute primary key;
\item attributes are typed;
\item attributes may be foreign keys referencing other relation schemas;
\item the primary keys of different relation schemas are pairwise disjoint.
\end{inparaenum}


We stress that all such requirements are natively captured in our functional
definition of a DB signature, and do not need to be formulated as axioms in the
DB theory. The DB theory is used to express additional constraints, like that
in Axiom~\eqref{eq:null}. In the following subsection, we thoroughly discuss
which properties must be respected by signatures and theories to guarantee that
our verification machinery is well-behaved.

One may wonder why we have not directly adopted a relational view for DB
schemas. This will become clear during the technical development. We anticipate
the main, intuitive reasons. First, our functional view allows us to
reconstruct in a single, homogeneous framework, some important results on
verification of artifact systems, achieved on different models that have been
unrelated so far~\cite{boj,DeLV16}. 
 Second, our
functional view makes the dependencies among different types explicit. In fact,
our notion of characteristic graph, which is readily computed from a DB
signature, exactly reconstructs the central notion of foreign key graph used
in~\cite{DeLV16} towards the main decidability results. 
Finally, we underline, once again, that \emph{free 
$n$-ary relation symbols can be added to our signatures} (see Remark~\ref{rem:extdb} and Definition~\ref{def:extdb} above) without compromising the results underlying our techniques.



%



\begin{remark}
%
%
 In some situations, it is useful to have many undefined keys and possibly also incomplete relations 
 with some undefined values. In such cases,
%
then one can only assume the left-to-right side
of~\eqref{eq:null}, which is equivalent to the ground axiom
\begin{equation}\label{eq:null1}
  f(\nullv) = \nullv
\end{equation}
In order to preserve the condition of being a foreign
key (i.e., the requirement that, for each
non-id attribute $a_f$ of a relation $R_S$ whose corresponding function symbol $f$ has id
sort $S'$ as image, we want an \emph{inclusion dependency} from $a_f$ to the id
attribute $id_{S'}$ of the relation $R_{S'}$), the axioms
\begin{equation}\label{eq:null2}
  \forall x~(f(x) \neq \nullv \rightarrow g(f(x)) \neq \nullv)~~
\end{equation}
are also needed. 
\end{remark}

\subsection{Formal Properties of DB Schemas}
\label{sec:theories}

The theory $T$ from Definition~\ref{def:db} must satisfy few crucial
requirements for our approach to work. In this section, we define such
requirements and show that they are matched, e.g., when the signature $\Sigma$ is
acyclic (as in~\cite{verifas}) and $T$ consists of Axioms~\eqref{eq:null} only. Actually, acyclicity is a stronger requirement than needed, which,
however, simplifies our exposition.

\medskip
\noindent
\textbf{Finite Model Property}. 
A $\Sigma$-formula $\phi$ is a $\Sigma$-\emph{constraint} (or just a
constraint) iff it is a conjunction of literals.
The constraint satisfiability problem for $T$ asks: given an existential
formula $\exists \uy \,\phi(\ux,\uy)$ (with $\phi$ a constraint\footnote{For
 the purposes of this definition, we may equivalently take $\phi$ to be
 quantifier-free.}), are there a model $\cM$ of $T$ and an assignment $\alpha$
to the free variables $\ux$ such that
$\cM, \alpha \models \exists \uy \,\phi(\ux,\uy)$?

We say that $T$ has the \emph{finite model property} (for constraint
satisfiability) iff every constraint $\phi$ that is satisfiable in a model of
$T$ is satisfiable in a DB instance of $T$.\footnote{This directly implies that $\phi$ is satisfiable also in a DB instance that interprets value sorts into finite sets.} The finite model property implies decidability of the constraint satisfiability
problem in case $T$ is recursively axiomatized. The following is proved in
Appendix~\ref{app_sec2}:

\begin{proposition}\label{prop:fmp}
  $T$ has the finite model property in  case $\Sigma$ is acyclic.
\end{proposition}

\medskip
\noindent
\textbf{Quantifier Elimination.} 
A $\Sigma$-theory $T$ has \emph{quantifier elimination} iff for every
$\Sigma$-formula $\phi(\ux)$ there is a quantifier-free formula $\phi'(\ux)$
such that $T\models \phi(\ux)\leftrightarrow \phi'(\ux)$.  It is known that
quantifier elimination holds if quantifiers can be eliminated from
\emph{primitive} formulae, i.e., formulae of the kind
$\exists \uy \,\phi(\ux, \uy)$, with $\phi$ a constraint. We assume that when
quantifier elimination is considered, there is an effective procedure that
eliminates quantifiers.

A DB theory $T$ does not necessarily have quantifier elimination; it is however
often possible to strengthen $T$ in a conservative way (with respect to
constraint satisfiability) and get quantifier elimination. We say that $T$ has
a \emph{model completion} iff there is a stronger theory $T^*\supseteq T$
(still within the same signature $\Sigma$ of $T$) such that
\begin{inparaenum}[\itshape (i)]
\item every $\Sigma$-constraint satisfiable in a model of $T$ is also so in a
  model of $T^*$;
\item $T^*$ has quantifier elimination. $T^*$ is called a \emph{model
   completion} of $T$.
\end{inparaenum}

\begin{proposition}\label{prop:mc}
  $T$ has a model completion in case it is axiomatized by universal
  one-variable formulae and $\Sigma$ is acyclic.
\end{proposition}

In Appendix~\ref{app_sec2} we prove the above proposition and give an algorithm
for quantifier elimination. This algorithm can be improved (and behaves much
better than their linear arithmetics counterparts) using a suitable version of
the Knuth-Bendix procedure~\cite{BaNi98} (studied in a dedicated
paper~\cite{new}, even if our \mcmt implementation already partially takes into
account such future development). Moreover, acyclicity is not needed in
general: when, for instance, $T:=\emptyset$ or when $T$ contains only
Axioms~\eqref{eq:null}, a model completion can be proved to exist, even if
$\Sigma$ is not acyclic, by using the Knuth-Bendix version of the
quantifier elimination algorithm.

\begin{remark}
	Proposition~\ref{prop:mc} holds also for DB extended-schemas, in case the universal
	one-variable formulae do not involve the relation symbols (so, the relations are ``free''): as explained in~\cite{new}, our implementation of the quantifier 
	elimination algorithm takes into account also this case. More generally, the model completion exists whenever we consider an acyclic
	DB extended-schema with a DB extended-theory $T$ that enjoys the amalgamation property.

\end{remark}

Hereafter, we make the following assumption:
\begin{assumption}
  \label{ass}
  The DB theories we consider have decidable constraint satisfiability problem,
  finite model property, and admit a model completion.
\end{assumption}

This assumption is matched, for instance, in the following three cases:
\begin{inparaenum}[\itshape (i)]
\item when $T$ is empty;
\item when $T$ is axiomatized by Axioms~\eqref{eq:null};
\item when $\Sigma$ is acyclic and $T$ is axiomatized by finitely many universal
  one-variable formulae (such as Axioms~\eqref{eq:null},\eqref{eq:null1},\eqref{eq:null2}, etc.).
\end{inparaenum}

 \begin{remark}
    Notice that the DB extended-schemas obtained by adding ``free'' relations to the DB schemas of \textit{(i)}, \textit{(ii)}, \textit{(iii)} above match Assumption~\ref{ass}.
  \end{remark}

\section{Relational Artifact Systems}
\label{sec:artifact}

We are now in the position to define our formal model of \emph{\RAS{s}} (\ras{s}), and to study parameterized safety problems over \ras{s}.
Since \ras{s} are array-based systems, we start by recalling the intuition
behind them.

In general terms, an array-based system is described using a multi-sorted
theory that contains two types of sorts, one accounting for the indexes of
arrays, and the other for the elements stored therein. Since the content of an
array changes over time, it is referred to using a second-order function
variable, whose interpretation in a state is that of a total function mapping
indexes to elements (so that applying the function to an index denotes the
classical \emph{read} operation for arrays). The definition of an array-based
system with array state variable $a$ always requires: a formula $I(a)$
describing the \emph{initial configuration} of the array $a$, and a formula
$\tau(a,a')$ describing a \emph{transition} that transforms the content of the
array from $a$ to $a'$. In such a setting, verifying whether the system can
reach unsafe configurations described by a formula $K(a)$ amounts to check
whether the formula
$I(a_0)\wedge \tau(a_0, a_1) \wedge \cdots \wedge \tau(a_{n-1}, a_n)\wedge
K(a_n)$ is satisfiable for some $n$.
Next, we make these ideas formally precise by grounding array-based systems in
the artifact-centric setting.

\medskip
\noindent
\textbf{The \ras Formal Model.}
 Following the tradition of artifact-centric systems \cite{DHPV09,DaDV12,boj,DeLV16}, a \ras consists of a read-only DB, a read-write working memory for artifacts, and a finite set of actions (also called services) that inspect the relational database and the working memory, and determine the new configuration of the working memory. In a \ras, the working memory consists of \textit{individual} and
\textit{higher order} variables. These variables (usually called \textit{arrays})
are supposed to model evolving relations, so-called \emph{artifact
relations} in \cite{DeLV16,verifas}.  The idea is to treat artifact relations in a uniform way as
we did for the read-only DB: we need extra sort symbols
(recall that each sort symbol corresponds to a database relation symbol) and
extra unary function symbols, the latter being treated as second-order
variables.

Given a DB schema $\Sigma$, an \emph{artifact extension} of $\Sigma$ is a
signature $\ext{\Sigma}$ obtained from $\Sigma$ by adding to it some extra sort
symbols\footnote{By `signature' we always mean 'signature with equality', so
 as soon as new sorts are added, the corresponding equality predicates are
 added too.}. These new sorts (usually indicated with letters $E, F, \dots$)
are called \emph{artifact sorts} (or \emph{artifact relations} by some abuse of
terminology), while the old sorts from $\Sigma$ are called \emph{basic sorts}.  
In \ras, artifacts and basic sorts correspond, respectively, to the index and the elements sorts mentioned in the literature on  array-based systems.
Below, given $\tup{\Sigma,T}$ and an artifact extension $\ext{\Sigma}$ of
$\Sigma$, when we speak of a $\ext{\Sigma}$-model of $T$, a DB instance of
$\tup{\ext{\Sigma},T}$, or a $\ext{\Sigma}$-model of $T^*$, we mean a
$\ext{\Sigma}$-structure $\cM$ whose reduct to $\Sigma$ respectively is a model
of $T$, a DB instance of $\tup{\Sigma,T}$, or a model of $T^*$.

An \emph{artifact setting} over $\ext{\Sigma}$ is a pair $(\ux,\ua)$ given by a
finite set $\ux$ of individual variables and a finite set $\ua$ of unary
function variables: \emph{the latter are required to have an artifact sort as
 source sort and a basic sort as target sort}. Variables in $\ux$ are called
\emph{artifact variables}, and variables in $\ua$ \emph{artifact
 components}. Given a DB instance $\cM$ of $\ext{\Sigma}$, an \emph{assignment} to an
artifact setting $(\ux, \ua)$ over $\ext{\Sigma}$ is a map $\alpha$ assigning
to every artifact variable $x_i\in \ux$ of sort $S_i$ an element
$x^\alpha\in S_i^\cM$ and to every artifact component
$a_j: E_j\longrightarrow U_j$ (with $a_j\in \ua$) a set-theoretic function
$a_j^\alpha: E_j^\cM\longrightarrow U_j^\cM$. 
In \ras, artifact components and artifact variables correspond, respectively, to \textit{arrays}
 and \textit{constant arrays} (i.e., arrays with all equal elements)
mentioned in the literature on  array-based systems.

We can view an assignment to an artifact setting $(\ux, \ua)$ as a DB instance
\emph{extending} the DB instance $\cM$ as follows.  Let all the artifact
components in $(\ux, \ua)$ having source $E$ be
$a_{i_1}: E\longrightarrow S_1, \cdots, a_{i_n}:E\longrightarrow S_n$. Viewed
as a relation in the artifact assignment $(\cM,\alpha)$, the artifact relation
$E$ ``consists'' of the set of tuples $ \{\tup{e, a_{i_1}^\alpha(e), \dots, a_{i_n}^\alpha(e)} \mid e\in E^{\cM} \}$.
Thus each element of $E$ is formed by an ``entry'' $e\in E^\cM$ (uniquely
identifying the tuple) and by ``data'' $\ua_i^\alpha(e)$ taken from the
read-only database $\cM$.  When the system evolves, the set $E^\cM$ of entries
remains fixed, whereas the components $\ua_i^\alpha(e)$ may change: typically,
we initially have $\ua_i^\alpha(e)=\nullv$, but these values are changed when
some defined values are inserted into the relation modeled by $E$; the values
are then repeatedly modified (and possibly also reset to $\nullv$, if the tuple
is removed and $e$ is re-set to point to undefined values)\footnote{In
 accordance with \textsc{mcmt} conventions, we denote
 the application of an artifact component $a$ to a term (i.e., constant or
 variable) $v$ also as $a[v]$ (standard notation for arrays), instead of $a(v)$.}.

In order to introduce verification problems in the symbolic setting of array-based systems, one first
has to specify which formulae are used to represent
\begin{inparablank}
        \item sets of states,
        \item the system initializations, and
        \item system evolution.
\end{inparablank}
To introduce \ras{s} we discuss the kind of formulae we
use.  In such formulae, we use notations like $\phi(\uz,\ua)$ to mean that
$\phi$ is a formula whose free individual variables are among the $\uz$ and
whose free unary function variables are among the $\ua$.  Let $(\ux,\ua)$ be an
artifact setting over $\ext{\Sigma}$, where $\ux=x_1,\dots, x_n$ are the
artifact variables and $\ua=a_1,\dots,a_m$ are the artifact components (their
source and target sorts are left implicit).

An \emph{initial formula} is a formula $\iota(\ux)$ of the
  form\footnote{Recall that $a_j =\lambda y. d_{j}$ abbreviates
   $\forall y\, a_{j}(y)=d_{j}$.}
    $\textstyle
    (\bigwedge_{i=1}^n x_i= c_i) \land
    (\bigwedge_{j=1}^m a_j =\lambda y. d_j)$,
  where $c_i$, $d_j$ are constants from $\Sigma$ (typically, $c_i$  and $d_j$ are $\nullv$).
A \emph{state formula} has the form
 $\exists \ue\, \phi(\ue, \ux,\ua)$,
  where $\phi$ is quantifier-free and the $\ue$ are individual variables of
  artifact sorts.
  A \emph{transition formula} $\hat\tau$ has the form
  \begin{equation}\label{eq:trans1}
  \textstyle
    \exists \ue\,(
      \gamma(\ue,\ux,\ua)
       \land\bigwedge_i x'_i= F_i(\ue,\ux,\ua)
       \land \bigwedge_j a'_j=\lambda y. G_j(y,\ue,\ux,\ua)
    )
  \end{equation}
  where the $\ue$ are individual variables (of \emph{both} basic and artifact
  sorts), $\gamma$ (the `guard') is quantifier-free, $\ux'$, $\ua'$ are renamed
  copies of $\ux$, $\ua$, and the $F_i$, $G_j$ (the `updates') are case-defined
  functions.
Transition formulae as above can express, e.g.,
\begin{inparaenum}[\itshape (i)]
\item insertion (with/without duplicates) of a tuple in an artifact relation,
\item removal of a tuple from an artifact relation,
\item transfer of a tuple from an artifact relation to artifact variables (and
  vice-versa), and
\item bulk removal/update of \emph{all} the tuples satisfying a certain
  condition from an artifact relation.
\end{inparaenum}
All the above operations can also be constrained: 
the formalization of the above operations in the formalism of our transition is 
straightforward (the reader can see all the details in Appendix~\ref{app_sec4}).

\begin{definition}\label{def:ras}
  A \emph{\RAS} (\ras) is
  $$
    \cS ~=~\tup{\Sigma,T,\ext{\Sigma}, \ux, \ua, \iota(\ux,\ua),
     \tau(\ux,\ua,\ux',\ua')}
$$
  where:
  \begin{inparaenum}[\it (i)]
  \item $\tup{\Sigma,T}$ is a (read-only) DB schema,
  \item $\ext{\Sigma}$ is an artifact extension of $\Sigma$,
  \item $(\ux, \ua)$ is an artifact setting over $\ext{\Sigma}$,
  \item $\iota$ is an intitial formula, and
  \item $\tau$ is a disjunction of transition formulae.
  \end{inparaenum}
\end{definition}

\begin{example}\label{ex:hr-short}
  We present here a \ras $\hrras$
  containing a multi-instance artifact accounting for the evolution of
  \emph{job applications}. Each job category may receive multiple applications
  from registered users. Such applications are then evaluated, finally deciding
  which to accept or reject. The example is inspired by the job
  hiring process presented in~\cite{Silv11} to show the intrinsic difficulties
  of capturing real-life processes with many-to-many interacting business
  entities using conventional process modeling notations (e.g., BPMN).  An
  extended version of this example is presented in
  Appendix~\ref{sec:hiring-example}.

  As for the read-only DB, $\hrras$ works over the DB schema of
  Example~\ref{ex:hr}, extended with a further value sort $\scoreval$ used to
  score job applications. $\scoreval$ contains $102$ values in the range
  $[\constant{-1},\constant{100}]$, where $\constant{-1}$ denotes the
  non-eligibility of the application, and a score from $\constant{0}$ to
  $\constant{100}$ indicates the actual one assigned after evaluating the
  application. For readability, we use as syntactic sugar usual predicates $<$, $>$,
  and $=$ to compare variables of type $\scoreval$. 

  As for the working memory, $\hrras$ consists of two artifacts.
  The first single-instance \emph{job hiring} artifact employs a dedicated $\pstatev$
variable to capture main phases that the running process goes through: initially, hiring is disabled ($\pstatev=\nullv$),
and, if there is at least one registered user in the HR DB,  $\pstatev$ becomes enabled.
The second multi-instance artifact accounts for the evolution of of \emph{user
   applications}.
 To model applications, we take the DB signature $\hrsig$ of
  the read-only HR DB, and enrich it with an artifact
  extension containing an artifact sort $\appidx$ used to \emph{index} (i.e.,
  \emph{``internally'' identify}) job applications. The management of job
  applications is then modeled by an artifact setting with:
  \begin{inparaenum}[\itshape (i)]
  \item artifact components with domain $\appidx$ capturing the artifact
    relation storing different job applications;
  \item additional individual variables as temporary memory to manipulate the
    artifact relation.
  \end{inparaenum}
  Specifically, each application consists of a job category, the identifier of
  the applicant user and that of an HR employee responsible for the
  application, the application score, and the final result (indicating whether
  the application is accepted or not). These
  information slots are encapsulated into dedicated artifact components, i.e.,
  function variables with domain $\appidx$ that collectively realize the
  application artifact relation:
  \[
    \small
    \begin{array}{l@{~:~}r@{~\longrightarrow~}l@{~~~~~~~~~~~~}l@{~:~}r@{~\longrightarrow~}l}
      \appjob    & \appidx  & \jobcatid & \appscore  & \appidx  & \scoreval\\
      \appuser   & \appidx  & \userid   &    \appresp   & \appidx  & \empid\\
        \appres  & \appidx  & \stringval\\
      \end{array}
    \]
  We now discuss the relevant transitions for inserting and evaluating job
  applications. When writing transition formulae, we make the following
  assumption: if an artifact variable/component is not mentioned at all, it is
  meant that is updated identically; otherwise, the relevant update function
  will specify how it is updated.\footnote{Non-deterministic
   updates can be formalized using existentially quantified variables in the
   transition.}  The insertion of an application into the system can be
  executed when the hiring process is enabled,
  and consists of two consecutive steps. To indicate when a step can be
  applied, also ensuring that the insertion of an application is not
  interrupted by the insertion of another one, we manipulate a string artifact
  variable $\recstate$. The first step is executable when $\recstate$ is
  $\nullv$, and aims at loading the application data into dedicated artifact
  variables through the following simultaneous effects:
  \begin{inparaenum}[\itshape (i)]
  \item the identifier of the user who wants to submit the application, and
    that of the targeted job category, are selected and respectively stored
    into variables $\uid$ and $\jid$;
  \item the identifier of an HR employee who becomes responsible for the
    application is selected and stored into variable $\eid$, with the
    requirement that such an employee must be competent in the job category
    targeted by the application;
  \item $\recstate$ evolves into state $\appreceived$.
   \end{inparaenum}
  Formally:
 \[
    \small
    \begin{array}{@{}l@{}}
      \exists \typedvar{u}{\userid},
      \typedvar{j}{\jobcatid},\typedvar{e}{\empid},\typedvar{c}{\compinid}\\
      \left(
        \begin{array}{@{}l@{}}
          \pstatev = \enab \land \recstate = \nullv
          \land u \neq \nullv \land j \neq \nullv \land e \neq \nullv \land
         c \neq \nullv
         \land \compemp(c) = e \\
          {}\land \compjob(c) = j
         \land \pstatev' = \enab \land \recstate' = \appreceived
         \land \uid' = u \land \jid' = j \land \eid' = e \land \cid'=c
        \end{array}
      \right)
    \end{array}
  \]

  The second step transfers the application data into the application artifact
  relation (using its corresponding function variables), and
  resets all application-related artifact variables to $\nullv$ (including
  $\recstate$, so that new applications can be inserted). For the insertion, a
  ``free'' index (i.e., an index pointing to an undefined applicant) is
  picked. The newly inserted application gets a default score of
  $\constant{-1}$ (``not eligible''), and an $\nullv$ final result:%
  \[
    \scriptsize
    \begin{array}{@{}l@{}}
      \exists \typedvar{i}{\appidx}\\
      \!\!\left(
        \begin{array}{@{}l@{}}
          \pstatev = \enab \land \recstate = \appreceived
          \land \appuser[i]=\nullv
          \land \pstatev'=\enab \land \recstate'=\nullv \land \cid'=\nullv
          \\{}\land
          \appjob' =
                  \smtlambda{j}{
                    \smtifinline{j=i}
                      {\jid}
                      {\appjob[j]}
                  }
          \land
          \appuser' =
                  \smtlambda{j}{
                    \smtifinline{j=i}
                      {\uid}
                      {\appuser[j]}
                  }
          \\{}\land
          \appresp' =
                  \smtlambda{j}{
                     \smtifinline{j=i}
                      {\eid}
                      {\appresp[j]}
                  }
          \land
          \appscore' =
                  \smtlambda{j}{
                    \smtifinline{j=i}
                      {\constant{-1}}
                      {\appscore[j]}
                  }
          \\{}\land
          \appres' =
                  \smtlambda{j}{
                    \smtifinline{j=i}
                      {\nullv}
                      {\appres[j]}
                  }
          \land \jid' = \nullv \land \uid' = \nullv \land \eid'=\nullv
        \end{array}
      \right)
    \end{array}
  \]
  Notice that such a transition does not prevent the possibility of inserting
  exactly the same application twice, at different indexes. If this is not
  wanted, the transition can be suitably changed so as to guarantee that no two
  identical applications can coexist in the same artifact relation (see
  Appendix~\ref{sec:hiring-example} for an example).

  Each application currently considered as not eligible can be made eligible by
  assigning a proper score to it:
  \[
    \small
    \begin{array}{@{}l@{}}
      \exists \typedvar{i}{\appidx}, \typedvar{s}{\scoreval}
      \left(
        \begin{array}{@{}l@{}}
          \pstatev = \enab \land \recstate=\nullv \\
           \appscore[i] = \constant{-1} \land \recstate'=\nullv \\
          s \geq \constant{0}
          \land \pstatev' = \enab \land \appscore'[i] = s
        \end{array}
      \right)
    \end{array}
  \]
  Finally, application results are computed when the process moves to state
  $\notified$.  This is handled by
  the \emph{bulk} transition:
  \[
    \small
    \begin{array}{@{}l@{}}
      \pstatev = \enab \land \recstate=\nullv\\
       \land \pstatev'=\notified \land \recstate'=\nullv \\
      \land \appres' = \smtlambda{j}{
                           \smtif{c}{\appscore[j] > \constant{80}}
                           {\win}{\los}
                         }
    \end{array}
  \]
  which declares applications with a score above $\constant{80}$ as winning,
  and the others as losing.
  \hfill\ensuremath{\triangleleft}
\end{example}

\medskip
\noindent
\textbf{Parameterized Safety via Backward Reachability.}
A \emph{safety} formula for $\cS$ is a state formula
$\upsilon(\ux)$ describing undesired states of $\cS$. 
As usual in array-based systems, we say that $\cS$ is \emph{safe with
        respect to} $\upsilon$ if intuitively the system has no finite run leading from
$\iota$ to $\upsilon$.  Formally, there is no DB-instance $\cM$ of $\tup{\ext{\Sigma},T}$, no $k\geq 0$, and
no assignment in $\cM$ to the variables $\ux^0,\ua^0 \dots, \ux^k, \ua^k$ such
that the formula
\begin{equation}\label{eq:smc1}
    \iota(\ux^0, \ua^0)
    \land \tau(\ux^0,\ua^0, \ux^1, \ua^1)
    \land \cdots
    \land\tau(\ux^{k-1},\ua^{k-1}, \ux^k,\ua^{k})
    \land \upsilon(\ux^k,\ua^{k})
\end{equation}
is true in $\cM$ (here $\ux^i$, $\ua^i$ are renamed copies of $\ux$, $\ua$). The \emph{safety
        problem} for $\cS$ is the following: \emph{given a safety formula $\upsilon$
        decide whether $\cS$ is safe with respect to $\upsilon$}.

\begin{example}\label{ex:hr-prop}
The following property expresses the undesired situation that, in the \ras from Example~\ref{ex:hr-short}, once the evaluation is notified there is an applicant with unknown result:
  \[
    \begin{array}{@{}l@{}}
      \exists  \typedvar{i}{\appidx}\\
      \left(
        \begin{array}{@{}l@{}}
          \pstatev=\notified \land \appuser[i]\neq \nullv
          \land \appres[i]\neq \win \land \appres[i]\neq \los
        \end{array}
      \right)
    \end{array}
  \]
  The job hiring \ras $\hrras$ turns out to be safe with respect to this
  property (cf.~Section~\ref{sec:first-exp}).
  \hfill\ensuremath{\triangleleft}
\end{example}

Algorithm~\ref{alg1} describes the \emph{backward reachability algorithm} (or,
\emph{backward search}) for handling the safety problem for $\cS$.  An integral
part of the algorithm is to compute \textit{symbolic} preimages.
For that purpose, we define for any $\phi_1(\uz,\uz')$ and $\phi_2(\uz)$,
$\mathit{Pre}(\phi_1,\phi_2)$ as the formula
$\exists \uz'(\phi_1(\uz, \uz')\land \phi_2(\uz'))$.
The \emph{preimage} of the set of states described by a state formula
$\phi(\ux)$ is the set of states described by
$\mathit{Pre}(\tau,\phi)$.\footnote{Notice that, when $\tau=\bigvee\hat\tau$,
        then $\mathit{Pre}(\tau,\phi)=\bigvee\mathit{Pre}(\hat\tau,\phi)$.}
$\mathsf{QE}(T^*,\phi)$ in Line~6 is a subprocedure that extends the quantifier
elimination algorithm of $T^*$ so as to convert the preimage $\mathit{Pre}(\tau,\phi)$
of a state formula $\phi$ into a state formula (equivalent to it modulo the
axioms of $T^*$), witnessing its \emph{regressability}: this is possible since $T^*$ eliminates from primitive
formulae the existentially quantified variables over the basic sorts,
whereas elimination of quantified variables over artifact sorts is not
possible, because these variables occur as arguments of artifact
components (see Lemma~\ref{lem:eq1} and Lemma~\ref{lem:eq2} in Appendix~\ref{app_sec4b} for more details).
Algorithm~\ref{alg1} computes iterated preimages of $\upsilon$ and
applies to them the above explained quantifier elimination over basic sorts, until a fixpoint is reached or until a set
intersecting the initial states (i.e., satisfying $\iota$) is found.\footnote{\textit{Inclusion} (Line~2) and \textit{disjointness} 
(Line~3) tests can be discharged via proof obligations to be handled by SMT solvers. 
The fixpoint is reached when the test in Line~2 returns \textit{unsat}, which means that the preimage of the set of the current states is included
in the set of states reached by the backward search so far.}
We obtain the following theorem, proved in Appendix~\ref{app_sec4b}:

\begin{theorem}\label{thm:nonsimple}
  Backward search (cf.\ Algorithm~\ref{alg1}) is effective and partially
  correct\footnote{\emph{Partial correctness} means that, when
                        the algorithm terminates, it gives a correct answer. \emph{Effectiveness}
                        means that all subprocedures in the algorithm can be effectively
                        executed.} for solving safety problems for \ras\/s.
\end{theorem}
\begin{proof}[Proof sketch]
  Algorithm~\ref{alg1}, to be effective, requires the availability of decision procedures for discharging the satisfiability tests in Lines~2-3. Thanks to
        the subprocedure $\mathsf{QE}(T^*,\phi)$, the only formulae we need to test in these lines
        have a specific form (i.e. $\exists\forall$-formulae\footnote{As defined in Appendix~\ref{app_sec4b}, we call $\exists\forall$-formulae
         the ones of the kind $\exists \ue\; \forall \ui \; \phi(\ue, \ui, \ux, \ua)$, where $\ue, \ui$ are variables whose sort is an artifact sort and $\phi$ is quantifier-free.}). By our hypotheses in Assumption~\ref{ass},
        we can freely assume that all the runs we are interested in take place inside
        models of $T^*$ (where we can eliminate quantifiers binding variables of basic sorts): in fact, a technical lemma (Lemma~\ref{lem:sat}) shows that formulae
        of the kind $\exists\forall$  are satisfiable in a model of $T$ iff they are satisfiable
        in a DB instance iff they are satisfiable in a model of $T^*$. The fact that a preimage of a state formula is a state formula is exploited
to make both safety and fixpoint tests effective (in fact, we prove that the entailment between
state formulae - and more generally satisfiability of $\exists\forall$ sentences - can be decided via finite instantiation techniques).
\end{proof}
\begin{wrapfigure}[11]{r}{0.45\textwidth}
\vspace{-20pt}
      \begin{algorithm}[H]
        \SetKwProg{Fn}{Function}{}{end}
        \Fn{$\mathsf{BReach}(\upsilon)$}{
                \setcounter{AlgoLine}{0}
                \ShowLn$\phi\longleftarrow \upsilon$;  $B\longleftarrow \bot$\;
                \ShowLn\While{$\phi\land \neg B$ is $T$-satisfiable}{
                        \ShowLn\If{$\iota\land \phi$ is $T$-satisfiable}
                        {\textbf{return}  $\mathsf{unsafe}$}
                        \setcounter{AlgoLine}{3}
                        \ShowLn$B\longleftarrow \phi\vee B$\;
                        \ShowLn$\phi\longleftarrow \mathit{Pre}(\tau, \phi)$\;
                        \ShowLn$\phi\longleftarrow \mathsf{QE}(T^*,\phi)$\;
                }
                \textbf{return} $(\mathsf{safe}, B)$;}{
                \caption{Schema of the backward reachability algorithm}\label{alg1}
        }
      \end{algorithm}
  \end{wrapfigure}

Theorem~\ref{thm:nonsimple} shows that backward search is a semi-decision procedure: if the system is
unsafe, backward search always terminates and discovers it; if the system is
safe, the procedure can diverge (but it is still correct).
Notice that the role of quantifier elimination (Line~6 of Algorithm~\ref{alg1})
is twofold:
\begin{inparaenum}[\itshape (i)]
\item It allows to discharge the fixpoint test of Line~2 (see
  Lemma~\ref{lem:sat}).
\item It ensures termination in significant cases, namely those where
  \emph{(strongly) local formulae}, introduced in the next section, are
  involved.
\end{inparaenum}

\vspace{10mm}

\section{Termination Results for \ras\/\lowercase{s}}
\label{sec:termination}


We now present three termination results, two relating \ras{s} to fundamental previous results, and one genuinely novel. All the proofs are given in the appendix.

\medskip
\noindent
\textbf{Termination for ``Simple'' Artifact Systems.} An interesting class of \ras{s} is the one where the working memory consists \textit{only} of artifact variables (without artifact relations). We call systems of this type \sas{s} (\textit{Simple Artifact Systems}). For \sas{s}, the following termination result holds.

\begin{theorem}\label{thm:basic}
        Let $\tup{\Sigma,T}$ be a DB schema with $\Sigma$ acyclic.  Then, for every \sas
        $\cS=\tup{\Sigma,T,\ux,\iota,\tau}$, backward search terminates and decides safety
                problems for $\cS$ in \PSPACE in the combined size of $\ux$, $\iota$, and
                $\tau$.
\end{theorem}


\begin{remark}
We remark that Theorem~\ref{thm:basic} holds also for DB extended-schemas (so, even adding ``free relations''
to the DB signatures). Moreover, notice that it can be shown that every existential
formula $\phi(\ux,\ux')$ can be turned into the form of
Formula~\eqref{eq:transition1}. Furthermore, we highlight that the proof of the decidability result of
Theorem~\ref{thm:basic} requires that the considered background theory $T$:
\begin{inparaenum}[\itshape (i)]
\item admits a model completion;
\item is \emph{locally finite}, i.e., up to $T$-equivalence, there are only
  finitely many atoms involving a fixed finite number of variables (this
  condition is implied by acyclicity);
\item is universal; and
\item enjoys decidability of constraint satisfiability.
\end{inparaenum}
Conditions~\textit{(iii)} and~\textit{(iv)} imply that one can decide whether a
finite structure is a model of $T$.  If
\textit{(ii)} and \textit{(iii)} hold, it is well-known that \textit{(i)} is
equivalent to amalgamation~\cite{wheeler}. Moreover, \textit{(ii)} alone always
holds for relational signatures and \textit{(iii)} is equivalent to $T$ being
closed under substructures (this is a standard preservation theorem in model
theory~\cite{CK}).  It follows that \textit{arbitrary relational signatures} (or \textit{locally finite
theories} in general, even allowing $n$-ary relation and $n$-ary function symbols) require only amalgamability and closure under
substructures. Thanks to these observations, Theorem~\ref{thm:basic} is reminiscent of an analogous result in~\cite{boj}, i.e., Theorem~5,  the crucial hypotheses of which are exactly amalgamability and closure under substructures, although the setting in that paper is different (there, key dependencies are not discussed, whereas we are interested only in DB (extended-)theories).
\end{remark}

In our
first-order setting, we can perform verification in a \emph{purely symbolic} way,
using (semi-)decision procedures provided by SMT-solvers, even when local
finiteness fails.  As mentioned before, local finiteness is guaranteed in the
relational context, but it does not hold anymore when \emph{arithmetic
 operations} are introduced.  Note that the theory of a single uninterpreted
binary relation (i.e., the theory of directed graphs) has a model completion, whereas it can be
easily seen that the theory of one binary relation endowed with primary key
dependencies 
(i.e. the theory of a binary relation which is a partial function) 
\textit{has not}, since it is \textit{not} amalgamable. So, the second distinctive feature of our setting naturally follows from
this observation: thanks to our functional representation of DB schemas (with
keys), the amalgamation property, required by Theorem~\ref{thm:basic}, holds, witnessing that our framework
 remains well-behaved even in the presence of key dependencies.

\medskip
\noindent
\textbf{Termination with Local Updates.}
Consider an acyclic signature $\Sigma$, a DB theory $T$ (satisfying our
Assumption~\ref{ass}), and an artifact setting $(\ux,\ua)$ over an artifact
extension $\ext{\Sigma}$ of $\Sigma$.
We call a state formula \emph{local} if it is a disjunction of the formulae
\begin{equation}\label{eq:local}
  \textstyle
  \exists e_1\cdots \exists e_k\, ( \delta(e_1,\dots, e_k)  \land
    \bigwedge_{i=1}^k \phi_i(e_i,\ux,\ua)),
\end{equation}
and \emph{strongly local} if it is a disjunction of the formulae
\begin{equation}\label{eq:localstrong}
  \textstyle
  \exists e_1\cdots\exists e_k\, ( \delta(e_1, \dots, e_k) \land
  \psi(\ux) \land \bigwedge_{i=1}^k \phi_i(e_i, \ua)).
\end{equation}
In~\eqref{eq:local} and~\eqref{eq:localstrong}, $\delta$ is a conjunction of
variable equalities and inequalities, $\phi_i$, $\psi$ are quantifier-free, and
$e_1,\ldots,e_k$ are individual variables varying over artifact sorts.
The key limitation of local state formulae is that they cannot
compare entries from different tuples of artifact relations:
each $\phi_i$ in~\eqref{eq:local} and~\eqref{eq:localstrong} can contain only
the existentially quantified variable $e_i$.

A transition formula $\hat\tau$ is \emph{local} (resp., \emph{strongly local})
if whenever a formula $\phi$ is local (resp., strongly local), so is
$\mathit{Pre}(\hat\tau,\phi)$ (modulo the axioms of $T^*$). Examples of
(strongly) local $\hat\tau$ are discussed in Appendix~\ref{app_sec4}.

\begin{theorem}\label{thm:term1}
  If $\Sigma$ is acyclic, backward search (cf.\ Algorithm~\ref{alg1})
  terminates when applied to a local safety formula in a \ras whose $\tau$ is a
  disjunction of local transition formulae.
\end{theorem}

\begin{proof}[Proof sketch]
Let $\tilde\Sigma$ be $\ext{\Sigma} \cup \{\ua, \ux\}$, i.e., $\ext{\Sigma}$
expanded with function symbols $\ua$ and constants $\ux$ ($\ua$ and $\ux$ are
treated as symbols of $\tilde\Sigma$, but not as variables anymore).
We call a $\tilde\Sigma$-structure
\emph{cyclic}\footnote{This is unrelated to cyclicity of $\Sigma$ defined in
 Section~\ref{sec:readonly}, and comes from universal algebra terminology.}  if
it is generated by one element belonging to the interpretation of an artifact
sort.  Since $\Sigma$ is acyclic, so is $\tilde\Sigma$, and then one can show
that there are only finitely many cyclic $\tilde\Sigma$-structures
$\cC_1,\ldots,\cC_N$ up to isomorphism.  With a $\tilde\Sigma$-structure $\cM$
we associate the tuple of numbers
$k_1(\cM), \dots, k_N(\cM)\in \mathbb N\cup\{\infty\}$ counting the numbers of
elements generating (as singletons) the cyclic substructures isomorphic to
$\cC_1, \dots, \cC_N$, respectively.  Then we show that, if the tuple
associated with $\cM$ is componentwise bigger than the one associated with
$\cN$, then $\cM$ satisfies all the local formulae satisfied by $\cN$.  Finally
we apply Dikson Lemma~\cite{BaNi98}.
\end{proof}

Note that Theorem~\ref{thm:term1} can be used to reconstruct the decidability
results of~\cite{verifas} concerning safety problems. Specifically, one needs
to show that transitions in~\cite{verifas} are strongly local which, in turn,
can be shown using quantifier elimination (see Appendix~\ref{app_sec4} for more
details).  Interestingly, Theorem~\ref{thm:term1} can be applied to more cases
not covered in~\cite{verifas}.
For example, one can provide transitions
enforcing \emph{updates over unboundedly many} tuples (bulk updates) that are
strongly local (cf.\ Appendix~\ref{app_sec4}).  One can also see that the
safety problem for our running example is decidable since all its transitions
are strongly local.
Another case considers coverability problems for broadcast
protocols~\cite{bro1,bro2}, which can be encoded using local formulae over the trivial one-sorted signature
containing just one basic sort, finitely many constants and one artifact sort with one artifact
component.  These problems can be decided with a
non-primitive recursive lower bound~\cite{SchmitzS13} (whereas the problems
in~\cite{verifas} have an \textsc{ExpSpace} upper bound). Recalling that \cite{verifas} handles verification of LTL-FO,
thus going beyond safety problems, this shows that the two settings are incomparable. Notice that Theorem~\ref{thm:term1}
implies also the decidability of the safety problem for \sas{s}, in case of $\Sigma$ acyclic.

\medskip
\noindent
\textbf{Termination for Tree-like Signatures.}
$\Sigma$ is \emph{tree-like} if it is acyclic and all non-leaf nodes have
outdegree~1.  An artifact setting over $\Sigma$ is tree-like if
$\tilde\Sigma:=\ext{\Sigma}\cup\{\ua, \ux\}$ is tree-like.  In tree-like
artifact settings, artifact relations have a single ``data'' component, and
basic relations are unary or binary.

\begin{theorem}
  \label{thm:term2}
  Backward search (cf.\ Algorithm~\ref{alg1}) terminates when applied to a
  safety problem in a \ras with a tree-like artifact setting.
\end{theorem}

\begin{proof}[Proof sketch]
The crux is to show, using Kruskal's Tree Theorem~\cite{kruskal},
that the finitely generated $\tilde\Sigma$-structures are a well-quasi-order
w.r.t.~the embeddability partial order.
\end{proof}

While tree-like \ras restrict artifact relations to be unary, their transitions
are not subject to any locality restriction. This allows for expressing rich
forms of updates, including general bulk updates (which allow us to capture
non-primitive recursive verification problems) and transitions comparing at
once different tuples in artifact relations. Notice that tree-like \ras{s} are
incomparable with the ``tree'' classes of \cite{boj}, since the former use
artifact relations, whereas the latter only individual variables. In
Appendix~\ref{sec:app-examples} we show the power of such advanced features in
a flight management process example.

%

\section{First experiments}
\label{sec:first-exp}

We implemented a prototype of the backward reachability algorithm for \ras{s} on top of the \textsc{mcmt} model checker for array-based systems.
%
Starting from its first version~\cite{mcmt}, \textsc{mcmt} was successfully
applied to a variety of settings: cache coherence and mutual exclusions protocols~\cite{lmcs},
timed~\cite{VERIFY} and fault-tolerant~\cite{jsat,disc} distributed systems, and imperative programs~\cite{atva,FI}. Interesting case studies
concerned waiting time bounds synthesis in parameterized timed
networks~\cite{nasa} and internet protocols~\cite{ifm}. Further related tools
include \textsc{safari}~\cite{cav}, \textsc{asasp}~\cite{asasp}, and \textsc{Cubicle}~\cite{cubicle_cav}. The latter
relies on a parallel architecture with further
powerful extensions.
The work principle of \textsc{mcmt} is rather simple: the tool generates the
proof obligations arising from the safety and fixpoint tests in backward search
(Lines~2-3 of Algorithm~\ref{alg1}) and passes them to the background
SMT-solver (currently it is \textsc{Yices}~\cite{Dutertre}).  In practice, the
situation is more complicated because SMT-solvers are quite efficient in
handling satisfiability problems in combined theories at quantifier-free level,
but may encounter difficulties with quantifiers. For this reason, \textsc{mcmt}
implements modules for \emph{quantifier elimination} and \emph{quantifier
 instantiation}.  A \emph{specific module} for the quantifier elimination
problems mentioned in Line~6 of Algorithm~\ref{alg1} has been added to
Version~2.8 of \textsc{mcmt}.

We produced a benchmark consisting of eight realistic business process examples
and ran it in \textsc{mcmt} (detailed explanations and results are given in
Appendix~\ref{app_experiments}). The examples are partially made by hand and
partially obtained from those supplied in~\cite{verifas}. A thorough comparison
with \textsc{Verifas}~\cite{verifas} is matter of future work, and is
non-trivial for a variety of reasons.  In particular, the two systems tackle
incomparable verification problems: on the one hand, we deal with safety
problems, whereas \textsc{Verifas} handles more general LTL-FO properties.  On
the other hand, we tackle features not available in \textsc{Verifas}, like bulk
updates and comparisons between artifact tuples.
Moreover, the two verifiers implement completely different state space
construction strategies: \textsc{mcmt} is based on backward reachability and
makes use of declarative techniques that rely on decision procedures, while
\textsc{Verifas} employs forward search via VASS encoding.
%

The benchmark is available as part of the last distribution~2.8 of
\textsc{mcmt}.\footnote{\url{http://users.mat.unimi.it/users/ghilardi/mcmt/},
 subdirectory \texttt{/examples/dbdriven} of the distribution. The user manual
 contains a new section (pages 36--39) on how to encode \ras{s} in MCMT
 specifications.}  Table~\ref{tab:exp-results} shows the very encouraging
results (the first row tackles Example~\ref{ex:hr-prop}). While a systematic
evaluation is out of scope, \textsc{mcmt} effectively handles the benchmark
with a comparable performance shown in other, well-studied systems, with
verification times below 1s in most cases.


%

\setlength{\dashlinedash}{1pt}
\setlength{\dashlinegap}{1pt}
\begin{table}
  \caption{Experimental results. The input system size is reflected by
   columns \textbf{\#AC}, \textbf{\#AV}, \textbf{\#T},
   indicating, resp., the number of artifact components, artifact variables, and
   transitions.}
  \label{tab:exp-results}
  \renewcommand{\arraystretch}{1.2}
  \begin{tabular}{@{}c||c@{}}
    \begin{minipage}{.48\linewidth}
      \centering{
       \scriptsize
  \begin{tabularx}{\linewidth}{
  >{\hsize=0.3\hsize}X
  >{\hsize=0.6\hsize\raggedleft\arraybackslash}X
  >{\hsize=0.6\hsize\raggedleft\arraybackslash}X
  >{\hsize=0.6\hsize\raggedleft\arraybackslash}X
  >{\hsize=0.5\hsize}X
  >{\hsize=0.6\hsize\raggedleft\arraybackslash}X
  >{\hsize=0.6\hsize\raggedleft\arraybackslash}X
    }
    \textbf{Exp.} & \textbf{\#AC} & \textbf{\#AV} & \textbf{\#T}
    & \textbf{Prop.} & \textbf{Res.} & \textbf{Time (sec)}
    \\\midrule
    E1 & 9 & 18 & 15 &
         E1P1 & \safe   & 0.06 \\
    &&&& E1P2 & \unsafe & 0.36 \\
    &&&& E1P3 & \unsafe & 0.50 \\
    &&&& E1P4 & \unsafe & 0.35 \\
   \hdashline
    E2 & 6 &13 &28&
         E2P1 & \safe   & 0.72 \\
    &&&& E2P2 & \unsafe & 0.88 \\
    &&&& E2P3 & \unsafe & 1.01 \\
    &&&& E2P4 & \unsafe & 0.83 \\
    \hdashline
    E3 & 4 & 14 & 13 &
         E3P1 & \safe   & 0.05 \\
    &&&& E3P2 & \unsafe & 0.06
\end{tabularx}}
         \end{minipage}
  &
  \renewcommand{\arraystretch}{1.2}
  \begin{minipage}{0.48\linewidth}
    \centering{
     \scriptsize
 \begin{tabularx}{\linewidth}{
  >{\hsize=0.3\hsize}X
  >{\hsize=0.6\hsize\raggedleft\arraybackslash}X
  >{\hsize=0.6\hsize\raggedleft\arraybackslash}X
  >{\hsize=0.6\hsize\raggedleft\arraybackslash}X
  >{\hsize=0.5\hsize}X
  >{\hsize=0.6\hsize\raggedleft\arraybackslash}X
  >{\hsize=0.6\hsize\raggedleft\arraybackslash}X}
    \textbf{Exp.} & \textbf{\#AC} & \textbf{\#AV} & \textbf{\#T}
    & \textbf{Prop.}&\textbf{Res.}&\textbf{Time (sec)}
   \\\midrule
    E4 & 9 & 11 & 21 &
         E4P1 & \safe   & 0.12 \\
    &&&& E4P2 & \unsafe & 0.13 \\
 \hdashline
    E5 & 6 & 17 & 34 &
         E5P1 & \safe   & 4.11 \\
    &&&& E5P2 & \unsafe & 0.17 \\
   \hdashline
    E6 & 2 & 7 &15 &
         E6P1 & \safe   & 0.04 \\
    &&&& E6P2 & \unsafe & 0.08 \\
    \hdashline
    E7 & 2 & 28 & 38 &
         E7P1 & \safe   & 1.00 \\
    &&&& E7P2 & \unsafe & 0.20 \\
   \hdashline
    E8 & 3 & 20 & 19 &
         E8P1 & \safe   & 0.70 \\
    &&&& E8P2 & \unsafe & 0.15
                 \end{tabularx}}
        \end{minipage}
 \end{tabular}
\end{table}


\section{Conclusion}

We have laid the foundations of SMT-based verification for artifact systems,
focusing on safety problems and relying on array-based systems as underlying
formal model. We have exploited the model-theoretic
machinery of model completion to overcome the main technical difficulty arising
from this approach, i.e., showing how to reconstruct quantifier elimination in
the rich setting of artifact systems. On top of this framework, we have
identified three classes of systems for which safety is decidable, which impose
different combinations of restrictions on the form of actions and the shape of
DB constraints.
The presented techniques have been implemented on top of the well-established
\textsc{mcmt} model checker, making our approach fully operational.

We consider the present work as the starting point for a full line of research
dedicated to SMT-based techniques for the effective verification of data-aware
processes, addressing richer forms of verification beyond safety (such as
liveness, fairness, or full LTL-FO) and richer classes of artifact systems,
(e.g., with concrete data types and arithmetics), while identifying novel
decidable classes (e.g., by restricting the structure of the DB and of
transition and state formulae).
Implementation-wise, we want to build on the reported encouraging results and
benchmark our approach using the \textsc{Verifas} system as a baseline, while
incorporating the plethora of optimizations available in SMT-based model
checking.  Finally, we plan to tackle more conventional process modeling
notations, in particular data-aware extensions of the de-facto standard BPMN.

\bibliographystyle{plainurl}
\bibliography{main-bib}

\newpage

\appendix

\section{Examples}
\label{sec:app-examples}

In this section, we present two full examples of \ras for which our backward
reachability technique terminates. In particular, they are meant to highlight
the expressiveness of our approach, even in presence of the restrictions
imposed by Theorems~\ref{thm:term1} and~\ref{thm:term2} towards decidability of
reachability. When writing transition formulae in the examples, we make the
following assumption: when an artifact variable or component is not mentioned
at all in a transition, it is meant that is updated identically; if it is
mentioned, the relevant update function in the transition will specify how it
is updated.\footnote{Notice that non-deterministic updates can be formalized
 using the existential quantified variables in the transition.}

\subsection{Job Hiring Process}
\label{sec:hiring-example}

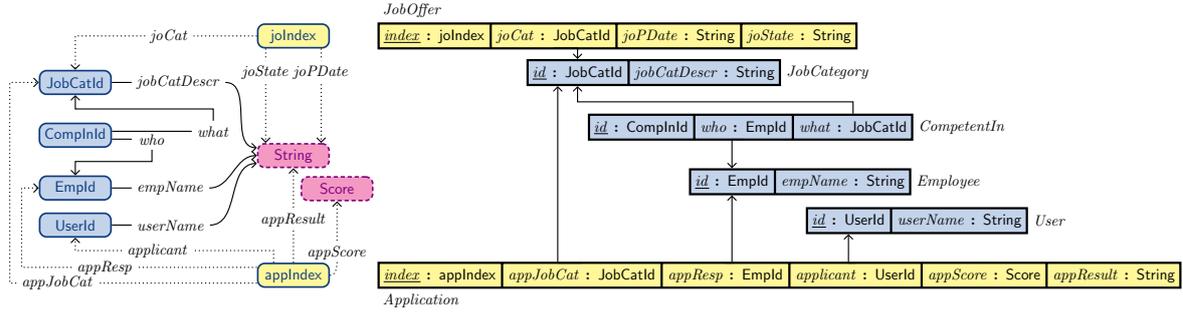
\begin{figure*}
\centering
\resizebox{\textwidth}{!}{
  \begin{tikzpicture}[node distance=3mm and 5mm]
    \node[idnode] (userid) {\userid};
    \node[functnode, right=of userid] (username) {\username};
    
    \node[idnode, above=of userid] (empid) {\empid};
    \node[functnode, right=of empid] (empname) {\empname};
    
    \node[idnode, above=7mm of empid] (compinid) {\compinid};
    \node[functnode, right=of compinid, yshift=-1mm] (compemp) {\compemp};
    \node[functnode, right=of compemp, yshift=2mm] (compjob) {\compjob};
    
    \node[idnode, above=7mm of compinid] (jobcatid) {\jobcatid};
    \node[functnode, right=of jobcatid] (jobcatdescr) {\jobdescr};
    
    \node[valnode, above right=2mm and 12mm of empname] (stringval) {\stringval};
    
	\node[artnode, above = of stringval, yshift = 21mm] (joidx)  {\joidx};
	\node[functnode, below=of joidx,xshift=7mm] (jopdate) {\jodate};
	\node[functnode, below=of joidx,xshift=-7mm] (jostate) {\jostate};
	\node[functnode,left=16mm of joidx] (jocat) {\jocat};
	
	\node[artnode, below = of stringval, yshift = -21mm] (appidx)  {\appidx};
	\node[functnode, below=1cm of stringval] (appres) {\appres};
	\node[functnode, above left=0mm and 16mm of appidx] (appuser) {\appuser};

	\node[valnode, below=2mm of stringval,xshift=11mm] (scoreval) {\scoreval};
	\node[functnode,below=1cm of scoreval] (appscore) {\appscore};

	\node[functnode,left=3cm of appidx,yshift=2mm] (appresp) {\appresp};
	\node[functnode,left=4cm of appidx,yshift=-2mm] (appjob) {\appjob};

    \draw[f] 
      (userid) edge (username);
    \draw[fd,out=0,in=180] 
      (username) edge ($(stringval.west)-(0,2mm)$);
    
    \draw[f]
      (empid) edge (empname);
    \draw[fd,out=0,in=180]
      (empname) edge (stringval);
      
    \draw[f,out=0,in=180]
      ($(compinid.east)+(0mm,-1mm)$) edge (compemp.west);
    \draw[fd]
      (compemp) |- ($(compemp.south)+(-3mm,-3mm)$) -| (empid);  
    \draw[f,out=0,in=180]
      ($(compinid.east)+(0mm,1mm)$) edge (compjob);
    \draw[fd]
      (compjob) |- ($(compjob.north)+(-3mm,3mm)$) -| (jobcatid); 
    
    \draw[f] 
      (jobcatid) edge (jobcatdescr);
    \draw[fd,out=0,in=180] 
      (jobcatdescr) edge ($(stringval.west)+(0,2mm)$);
      
       \draw[f,dotted] 
      ($(joidx.south)-(7mm,0)$) edge (jostate);
    \draw[fd,dotted] 
      (jostate) edge ($(stringval.north)-(7mm,0)$);
      
       \draw[f,dotted] 
      ($(joidx.south)+(7mm,0)$) edge (jopdate);
    \draw[fd,dotted] 
      (jopdate) edge ($(stringval.north)+(7mm,0)$);
      
      \draw[f,dotted] 
      (joidx) edge (jocat); 
      \draw[fd,dotted] 
      (jocat) -| (jobcatid); 
      
	   \draw[f,dotted] 
	   (appidx) edge ($(appres.south)+(0,1mm)$); 
      \draw[fd,dotted] 
      (appres) edge (stringval.south); 
      
	   \draw[f,dotted] 
	   ($(appidx.north)-(5mm,0)$) |- (appuser.east); 
      \draw[fd,dotted,out=-180,in=0] 
      (appuser) -| (userid.south);

      \draw[f,dotted] 
	   ($(appidx.west)+(0,2mm)$)  edge (appresp);  
     \draw[fd,dotted] 
	   (appresp) -- ++(-21mm,0) |- (empid);  
	   
	    \draw[f,dotted] 
	   ($(appidx.west)-(0,2mm)$) edge (appjob);  
     \draw[fd,dotted] 
	   (appjob) -- ++(-12mm,0) |- (jobcatid);  
	   
	  \draw[f,dotted,out=0,in=-90] 
	   (appidx.east)  edge ($(appscore.south)+(0,0.7mm)$);  
     \draw[fd,dotted] 
	   (appscore) edge (scoreval);


    \node[
      relation=4, 
      rectangle split horizontal, 
      rectangle split part fill={yellow!50},
      right=12mm of joidx.east,
      anchor=west,
      ultra thick
    ] (joboffer)
    {
      \underline{$\mathit{index}$} : \joidx %
      \nodepart{two} {\jocat} : \jobcatid
      \nodepart{three} {\jodate} : \stringval
      \nodepart{four} {\jostate} : \stringval
    };
    \node[above=0mm of joboffer.north west,anchor=south west] {\relname{JobOffer}};
    
        \node[
      relation=6, 
      rectangle split horizontal, 
      rectangle split part fill={yellow!50},
      right=12mm of appidx.east,
      anchor=west,
      ultra thick
    ] (app)
    {
      \nodepart{one} \underline{$\mathit{index}$} : \appidx %
      \nodepart{two} {\appjob} : \jobcatid
      \nodepart{three} {\appresp} : \empid
      \nodepart{four} {\appuser} : \userid
      \nodepart{five} {\appscore} : \scoreval
      \nodepart{six} {\appres} : \stringval
    };
    \node[below=0mm of app.south west,anchor=north west] {\relname{Application}};

      \node[
      relation=2, 
      rectangle split horizontal, 
      rectangle split part fill={cyan!50},
      above=7mm of app.four north,
      xshift=15.5mm,
      ultra thick
    ] (user)
    {
      \nodepart{one} \underline{$\mathit{id}$} : \userid %
      \nodepart{two}  {\username} : \stringval
    };
    \node[right=0mm of user] {\relname{User}};
 
    \node[
      relation=2, 
      rectangle split horizontal, 
      rectangle split part fill={cyan!50},
      above= of user.north west,
      xshift=-29.5mm,
      anchor=south west,
      ultra thick
    ] (emp)
    {
      \nodepart{one} \underline{$\mathit{id}$} : \empid %
      \nodepart{two}  {\empname} : \stringval
    };
    \node[right=0mm of emp] {\relname{Employee}}; 
    
    \node[
      relation=3, 
      rectangle split horizontal, 
      rectangle split part fill={cyan!50},
      above=7mm of emp.north west,
      xshift=-25.5mm,
      anchor=south west,
      ultra thick
    ] (compin)
    {
      \nodepart{one} \underline{$\mathit{id}$} : \compinid %
      \nodepart{two}  {\compemp} : \empid
      \nodepart{three}  {\compjob} : \jobcatid
    };
    \node[right=0mm of compin] {\relname{CompetentIn}};
    
    \node[
      relation=2, 
      rectangle split horizontal, 
      rectangle split part fill={cyan!50},
      above=7mm of compin.north west,
      xshift=-15.5mm,
      anchor=south west,
      ultra thick
    ] (jobcat)
    {
      \nodepart{one} \underline{$\mathit{id}$} : \jobcatid %
      \nodepart{two} {\jobdescr} : \stringval
    };
    \node[right=0mm of jobcat] {\relname{JobCategory}};

    \draw[fd] 
      ($(joboffer.two south)+(0,1pt)$) 
      -| (jobcat.one north);
  
    \draw[fd] 
      ($(compin.two south)+(0,1pt)$) 
      -| (emp.one north);
    
    \draw[fd,out=90,in=-90] 
      (compin.three north) 
      |- ($(compin.two north)+(-3mm,3mm)$) 
      -| (jobcat.one south);
    
    \draw[fd] 
      ($(app.two north)-(0,1pt)$) 
      -|($(jobcat.one south)-(5mm,0)$);
    
    \draw[fd]
      ($(app.four north)-(0,1pt)$)
      -| (user.one south);

    \draw[fd]
      ($(app.three north)-(0,1pt)$)
      -| (emp.one south);

  \end{tikzpicture} 
}

%
%
%

%
%
  \caption{On the left: characteristic graph of the human resources DB signature from Example~\ref{ex:hr}, augmented with the signature of the artifact extension for the job hiring process; value sorts are shown in pink, basic id sorts in blue, and artifact id sorts in yellow.
      On the right: relational view of the DB signature and the corresponding artifact relations; each cell denotes an attribute with its type, underlined attributes denote primary keys, and directed edges capture foreign keys.}
  \label{fig:hr_cg}
\end{figure*}

We present a \ras $\hrras$ capturing a job hiring process where multiple job categories may be turned into actual job offers, each one receiving many applications from registered users. Such applications are then evaluated, finally deciding which are accepted and which are rejected. The example is inspired by the job hiring process presented in \cite{Silv11} to show the intrinsic difficulties of capturing real-life processes with many-to-many interacting business entities using conventional process modeling notations (such as BPMN). 
Note that this example is also demonstrating the co-evolution of multiple instances of two different artifacts (namely, job offer and application).   
 
 As for the read-only DB, $\hrras$ works over the DB schema of Example~\ref{ex:hr}, extended with a further value sort $\scoreval$ used to score the applications sent for job offerings. $\scoreval$ contains $102$ different values, intuitively corresponding to the integer numbers from $-1$ to $100$ (included), where $-1$ denotes that the application is considered to be not eligible, while a score between $0$ and $100$ indicates the actual score assigned after evaluating the application. For the sake of readability, we make use of the usual integer comparison predicates to compare variables of type $\scoreval$. This is simply syntactic sugar and does not require the introduction of rigid predicates in our framework. In fact, given two variables $x$ and $y$ of type $\scoreval$, $x<y$ is a shortcut for the finitary disjunction testing that $x$ is one of the scores that are ``less than'' $y$ (similarly for the other comparison predicates).

As for the working memory, $\hrras$ consists of three artifacts: a single-instance \emph{job hiring} artifact tracking the three main phases of the overall process, and two multi-instance artifacts accounting for the evolution of \emph{job offers}, and that of corresponding \emph{user applications}. The job hiring artifact simply requires a dedicated $\pstatev$ variable to store the current process state. The job offer and user application multi-instance artifacts are instead modeled by enriching the DB signature $\hrsig$ of the read-only database of human resources. In particular, an artifact extension is added containing two artifact sorts $\joidx$ and $\appidx$ used to respectively \emph{index} (i.e., \emph{``internally'' identify}) job offers and applications. The management of job offers and applications is then modeled by a full-fledged artifact setting that adopts:
\begin{compactitem} 
\item artifact components with domains $\joidx$ and $\appidx$ to capture the artifact relations storing multiple instances of job offers and applications; 
\item individual variables used as temporary memory to manipulate the artifact relations.
\end{compactitem}
The actual components of such an artifact setting will be introduced when needed.

 We now describe how the process works, step by step. Initially, hiring is disabled, which is captured by initially setting the $\pstatev$ variable to $\nullv$. A transition of the process from disabled to \emph{enabled} may occur  provided that the read-only HR DB contains at least one registered user (who, in turn, may decide to apply for job offers created during this phase).  Technically, we introduce a dedicated artifact variable $\uidv$ initialized to $\nullv$, and used to load the identifier of such a registered user, if (s)he exists. The enablement task is then captured by the following transition formula:
  \begin{align*}
  \exists y:\userid
  \left(
  \begin{array}{@{}l@{}l@{}}
    \pstatev = \nullv \land y \neq \nullv \\ 
    \recstate=\nullv \land \recstate'=\nullv \\
    \land \pstatev' = \constant{enabled} \land \uidv' = y    
  \end{array}
  \right)
 \end{align*}
	  
We now focus on the creation of a job offer. When the overall hiring process is enabled, some job categories present in the read-only DB may be published into a corresponding job offer, consequently becoming ready to receive applications. This is done in two steps. In the first step, we transfer the id of the job category to be published to the artifact variable $\jid$, and the string representing the publishing date to the artifact variable $\dval$. Thus, $\jid$ is filled with the identifier of a job category picked from $\jobcatid$ (modeling a nondeterministic choice of category), while $\dval$ is filled with a $\stringval$ (modeling a \emph{user input} where one of the infinitely many strings is injected into $\dval$).
 
 In addition, the transition interacts with a further artifact variable $\ostate$ capturing the publishing state of offers, and consequently used to synchronize the two steps for publishing a job offer. In particular, this first step can be executed only if $\ostate$ is \emph{not} in state $\publishing$, and has the effect of setting it to such a value, thus preventing the first step to be executed twice in a row (which would actually overwrite what has been stored in $\jid$ and $\dval$). Technically, we have:
\[
\small
\begin{array}{@{}l@{}}
\exists \typedvar{j}{\jobcatid},  \typedvar{d}{\stringval} 
  \left( 
    \begin{array}{@{}l@{}}
      \pstatev=\enab \land \ostate\neq\publishing\land j\neq\nullv\\
      {}\land
       \pstatev'=\enab \land \ostate' = \publishing
       \land
        \jid' = j \land \dval' = d \\
        \recstate=\nullv \land \recstate'=\nullv
    \end{array}
  \right)
\end{array}
\]
The second step consists in transferring the content of these three variables into corresponding artifact components that keep track of all active job offers, at the same time resetting the content of the artifact variables to $\nullv$. This is done by introducing three function variables with domain $\joidx$, respectively keeping track of the category, publishing date, and state of job offers:
\[
\small
  \begin{array}{l@{~:~}r@{~\longrightarrow~}l}
    \jocat    & \joidx  & \jobcatid\\
    \jodate   & \joidx  & \stringval\\
    \jostate  & \joidx  & \stringval\\
  \end{array}
\]
With these artifact components at hand, the second step is then realized as follows:
\[
\small
  \begin{array}{@{}l@{}}
    \exists \typedvar{i}{\joidx}\\
    \left(
      \begin{array}{@{}l@{}}
        \pstatev = \enab \land \ostate=\publishing 
        \land
         \jodate[i]=\nullv \land  \jocat[i] = \nullv \land \jostate[i]=\nullv \\
          {}\land \recstate'=\nullv
         \land
        \pstatev' = \enab \land \ostate'=\published\\
        {}\land
        \jocat' = \smtlambda{j}{
                    \smtif{c}{j=i}
                      {\jid}
                      {
                        \smtif{t}{\jocat[j]=\jid}
                        {\nullv}
                        {
                          \jocat[j]                        
                        }
                      }
                  }
        \land
        \jodate' = \smtlambda{j}{
                    \smtif{c}{j=i}
                      {\dval}
                      {
                        \smtif{t}{\jocat[j]=\jid}
                        {\nullv}
                        {
                          \jodate[j]                        
                        }
                      }
                  }\\
        {}\land
        \jostate' = \smtlambda{j}{
                    \smtif{c}{j=i}
                      {\joopen}
                      {
                        \smtif{t}{\jocat[j]=\jid}
                        {\nullv}
                        {
                          \jostate[j]                        
                        }
                      }
                  }
          \\{}\land
          \uid'= \nullv \land \eid'=\nullv \land \jid'=\nullv \land \dval'=\nullv \land \cid'=\nullv
      \end{array}
    \right)
  \end{array}
\]
The ``if-then-else'' pattern is used to create an entry for the job offer artifact relation containing the information stored into the artifact variables populated in the first step, at the same time \emph{making sure that only one entry exists for a given job category}. This is done by picking a job offer index $i$ that is not already pointing to an actual job offer, i.e., such that  the $i$-th element of $\jocat$ is $\nullv$. Then, the transition updates the whole 
    content of the three artifact components $\jocat$, $\jodate$, and $\jostate$ as follows:
    \begin{compactitem}[$\bullet$]
      \item The $i$-th entry of such variables is respectively assigned to the job category stored in $\jobcatid$, the string stored in $\dval$, and the constant $\joopen$ (signifying that this entry is ready to receive applications).
      \item All other entries are kept unaltered, with the exception of a possibly existing entry $j$ with $j\neq i$ that points to the same job category contained in $\jobcatid$. If such an entry $j$ exists, its content is reset, by assigning to the $j$-th component of all  three artifact components the value $\nullv$.     Obviously, other strategies to resolve this possible conflict can be seamlessly captured in our framework.
    \end{compactitem}
  A similar conflict resolution strategy will be used in the other transitions of this example.
  
  We now focus on the evolution of applications to job offers. Each application consists of a job category, the identifier of the applicant user, the identifier of an employee from human resources who is responsible for the application, the score assigned to the application, and the application final result (indicating whether the application is among the winners or the losers for the job offer). These five information types are encapsulated into five dedicated function variables with domain $\appidx$, collectively realizing the application artifact relation:
\[
\small
  \begin{array}{l@{~:~}r@{~\longrightarrow~}l}
    \appjob    & \appidx  & \jobcatid\\
    \appuser   & \appidx  & \userid\\
    \appresp   & \appidx  & \empid\\
    \appscore  & \appidx  & \scoreval\\
    \appres  & \appidx  & \stringval\\
  \end{array}
\]
  
With these function variables at hand, we discuss the insertion of an application into the system for an open job offer. This is again managed in multiple steps, first loading the necessary information into dedicated artifact variables, and finally transferring them into the function variables that collectively realize the application artifact relation. To synchronize these multiple steps and define which step is applicable in a given state, we make use of a string artifact variable called $\recstate$. The first step to insert an application is executed when $\recstate$ is $\nullv$, and has the effect of loading into $\jid$ the identifier of a job category that has a corresponding open job offer, at the same time putting $\recstate$ in state $\joselected$. 
\[
\small
  \begin{array}{@{}l@{}}
  \exists \typedvar{i}{\joidx}\\
  \left(
    \begin{array}{@{}l@{}}
    \pstatev = \enab \land \recstate = \nullv \land \ostate\neq\publishing 
    \land \jocat[i] \neq \nullv \land \jostate[i] = \joopen\\
    {}\land \pstatev' = \enab \land \recstate' = \joselected
    \land \jid' = \jocat[i]\land \jocat'=\jocat \land \ostate'=\nullv\\
    {}\land  \uid'= \nullv \land \eid'=\nullv \land \jid'=\nullv \land \dval'=\nullv \land \cid'=\nullv
    \end{array}
  \right)
  \end{array}
\]     
The last row of the transition resets the content of all artifact variables, cleaning the working memory for the forthcoming steps (avoiding that stale values are present there). This is also useful from the technical point of view, as it guarantees that the transition is \emph{strongly local} (cf.~Section~\ref{sec:termination}, and the discussion in Appendix~\ref{sec:deletion-updates}).

    The second step has a twofold purpose: picking the identifier of the user who wants to submit an application for the selected job offer, and assigning to its application an employee of human resources who is competent in the category of the job offer.  This also results in an update of variable $\recstate$:
\[
\small
  \begin{array}{@{}l@{}}
  \exists \typedvar{u}{\userid}, \typedvar{e}{\empid},\typedvar{c}{\compinid}\\
  \left(
    \begin{array}{@{}l@{}}
    \pstatev = \enab \land \recstate = \joselected\land \ostate\neq\publishing 
    \land \compemp(c) = e\\
     \land \compjob(c) = \jid\land \jid\neq \nullv
    \land u\neq \nullv \land c\neq\nullv
    {}\land \pstatev' = \enab\\ \land \recstate' = \appreceived
    \land \jid' = \jid \land \uid' = u \land \eid' = e \land \cid'=c
    \end{array}
  \right)
  \end{array}  
\]

The last step transfers the application data into the application artifact relation, making sure that no two applications exist for the same user and the same job category. The transfer is done by assigning the artifact variables to corresponding components of the application artifact relation, at the same resetting all application-related artifact variables to $\nullv$ (including $\recstate$, so that new applications can be inserted). For the insertion, a ``free'' index (i.e., an index pointing to an undefined applicant, with an undefined job category and an undefined responsible) is picked. The newly inserted application gets a default score of $\constant{-1}$ (thus initializing it to ``not eligible''), while the final result is $\nullv$:    
\[
\small
  \begin{array}{@{}l@{}}
    \exists \typedvar{i}{\appidx}\\
    \left(
      \begin{array}{@{}l@{}}
        ~\pstatev = \enab \land \recstate = \appreceived \land \ostate\neq\publishing \\
        {}\land \appjob[i] = \nullv \land \appuser[i]=\nullv \land \appresp[i]=\nullv\\ 
            \land \pstatev' = \enab \land \recstate' = \nullv \land \ostate'=\nullv\\
        {}\land
        \appjob' =
                  \smtlambda{j}{
                    \smtif{c}{j=i}
                      {\jid}
                      {
                        \smtif{t}{\left(\appuser[j]=\uid \land \appresp[j] = \eid\right)}
                        {\nullv}
                        {
                          \appjob[j]                        
                        }
                      }
                  }\\
        {}\land
        \appuser' = 
                  \smtlambda{j}{
                    \smtif{c}{j=i}
                      {\uid}
                      {
                        \smtif{t}{\left(\appuser[j]=\uid \land \appresp[j] = \eid\right)}
                        {\nullv}
                        {
                          \appuser[j]                        
                        }
                      }
                  }\\
          {}\land
          \appresp' = 
                  \smtlambda{j}{
                    \smtif{c}{j=i}
                      {\eid}
                      {
                        \smtif{t}{\left(\appuser[j]=\uid \land \appresp[j] = \eid\right)}
                        {\nullv}
                        {
                          \appresp[j]                        
                        }
                      }
                  }\\
          {}\land
          \appscore' = 
                  \smtlambda{j}{
                    \smtif{c}{j=i}
                      {\constant{-1}}
                      {
                        \smtif{t}{\left(\appuser[j]=\uid \land \appresp[j] = \eid\right)}
                        {\nullv}
                        {
                          \appscore[j]                        
                        }
                      }
                  }\\
          {}\land
          \appres' = 
                  \smtlambda{j}{
                    \smtif{c}{j=i \lor \left(\appuser[j]=\uid \land \appresp[j] = \eid\right)}
                      {\nullv}
                      {\appres[j]}
                  }\\  
                  {}\land  \uid'= \nullv \land \eid'=\nullv \land \jid'=\nullv \land \dval'=\nullv \land \cid'=\nullv
    
      \end{array}
    \right)
  \end{array}
\]

Each single application that is currently considered as not eligible can be made eligible by carrying out an evaluation that assigns a proper score to it. This is managed by the following transition:
\[
\small
  \begin{array}{@{}l@{}}
  \exists \typedvar{i}{\appidx}, \typedvar{s}{\scoreval}
  \left(
    \begin{array}{@{}l@{}}
    \pstatev = \enab \land \appuser[i] \neq \nullv \land \ostate\neq\publishing \\
    \appscore[i] = \constant{-1} \land s \geq 0
    {}\land \pstatev' = \enab \land \appscore'[i] = s
    \end{array}
  \right)
  \end{array}
\] 
Evaluations are only possible as long as the process is in the $\enab$ state. The process moves from enabled to \emph{final} once the deadline for receiving applications to job offers is actually reached. This event is captured with pure nondeterminism, and has the additional \emph{bulk} effect of turning all open job offers to \emph{closed}:
\[
\small
  \begin{array}{@{}l@{}}
    \pstatev = \enab \land \pstatev'=\final \land \ostate\neq\publishing \land \ostate'=\nullv\\ 
    \recstate=\nullv \land \recstate'=\nullv \land \dval'=\nullv\\
    \land \jostate' =   \smtlambda{j}{
                            \smtif{c}{\jostate[j] = \joopen}
                            {\joclosed}{\jostate[j]}
                          }    
  \end{array}
\]

Finally, we consider the determination of winners and losers, which is carried out when the overall hiring process moves from final to \emph{notified}. This is captured by the following \emph{bulk} transition, which declares all applications with a score above $\constant{80}$ as winning, and all the others as losing:
\[
\small
  \begin{array}{@{}l@{}}
    \pstatev = \final \land \pstatev'=\notified\land \dval'=\nullv \land \ostate\neq\publishing \\ 
    \recstate=\nullv \land \recstate'=\nullv \land \ostate'=\nullv\\
    \land \appres' =    \smtlambda{j}{
                              \smtif{c}{\appscore[j] > \constant{80}}
                            {\win}{\los}
                          }
  \end{array}
\]  

We close the example with the following key observation. All transitions of the hiring process are, in their current form, strongly local, with the exception of those operating over artifact relations in a way that ensures no repeated entries are inserted. Such transitions can be turned into strongly local ones if \emph{repetitions in the artifact relations are allowed}. That is, multiple identical job offers and applications can be inserted in the corresponding relations, using different indexes. This is the strategy adopted in Example~\ref{ex:hr-short} in the main text of the paper. This approach realizes a sort of multiset semantics for artifact relations. The impact of this variant to verification of safety properties is discussed in Appendix~\ref{sec:insertion-updates}.

\subsection{Flight Management Process}
\newcommand{\fmsig}{\Sigma_{\mathit{fm}}}
\newcommand{\cityid}{\type{CityId}}
\newcommand{\flightid}{\type{FlightId}} 
\newcommand{\sfdest}{\funct{safeCity}}
\newcommand{\flightidx}{\type{FligthIndex}}
\newcommand{\passidx}{\type{PassengerIndex}}
\newcommand{\overbooked}{\funct{overbooked}}
\newcommand{\regpasger}{\funct{regdPassenger}}
\newcommand{\cityidx}{\type{CityIndex}}
\newcommand{\destination}{\funct{destination}}

\newcommand{\sdid}{\artvar{safeCitytId}}
\newcommand{\usdid}{\artvar{unsafeCityId}}

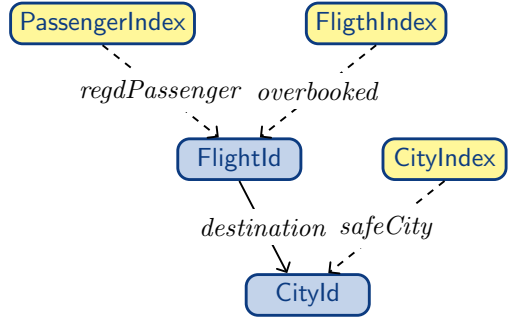
\begin{figure}[h!]
\centering
\scalebox{.9}{
  \begin{tikzpicture}[node distance=1mm,auto,>=stealth]
    \node[idnode] (cityid) at (0,0) {\cityid};
    \node[idnode] (flightid) at (-1,2) {\flightid};
    \node[artnode] (passidx) at (-3,4) {\passidx};
    \node[artnode] (flightidx) at (1,4) {\flightidx};
    \node[artnode] (cityidx) at (2,2) {\cityidx};

%
%
%
%
%
%
    \draw[fd] 
      ($(flightid.south)$) -- node[fill=white, anchor=center, pos=0.5] {\destination} ($(cityid.north)+(-0.3,0)$);
       \draw[fd,dashed] 
      ($(flightidx.south)$) -- node[fill=white, anchor=center, pos=0.5] {\overbooked} ($(flightid.north)+(0.3,0)$);
       \draw[fd,dashed] 
      ($(passidx.south)$)  -- node[fill=white, anchor=center, pos=0.5] {\regpasger} ($(flightid.north)+(-0.3,0)$);
       \draw[fd,dashed] 
      ($(cityidx.south)$) -- node[fill=white, anchor=center, pos=0.5] {\sfdest} ($(cityid.north)+(0.3,0)$);
  \end{tikzpicture} 
}
  \caption{A characteristic graph of the flight management process, where blue and yellow boxes respectively represent basic and artifact sorts.}
  \label{fig:fm}
\end{figure}
In this section we consider a simple \ras that falls in the scope of the decidability result described in Section~\ref{sec:termination}. 
Specifically, this example has a tree-like artifact setting (see Figure~\ref{fig:fm}), thus assuring that, when solving the safety problem  
for it, the backward search algorithm is guaranteed to terminate. Note, however, that the termination result adopted here is the one of  Theorem~\ref{thm:term2} due to the non-locality of certain transitions, as explained in detail below. 

The flight management process represents a simplified version of a flight management
system adopted by an airline.
To prepare a flight, the company picks a corresponding destination (that meets the aviation safety compliance indications)
and consequently reports on a number of passengers that are going to attend the flight. 
Then, an airport dispatcher may pick a manned flight 
and put it in the airports flight plan.
In case the flight destination becomes unsafe (e.g., it was stroke by a hurricane or the hosting 
airport had been seized by terrorists), the dispatcher uses the system to inform 
the airline about this condition. In turn, the airline notifies all the passengers of the affected destination about the 
contingency, and temporary cancels their flights.

To formalize these different aspects, we make use of a DB signature $\fmsig$ that consists of: 
\begin{inparaenum}[\it(i)]
\item  two id sorts, used to identify flights and cities;
\item one function symbol $\destination : \flightid\longrightarrow\cityid$ mapping flight identifiers to their corresponding destinations (i.e., city identifiers).
 \end{inparaenum}
Note that, in a classical relational model (cf. Section~\ref{sec:relational-view}), our signature would contain two relations: one binary $R_{\flightid}$ that defines flights and their destinations, and another unary $R_{\cityid}$ identifying cities, that are referenced by $R_{\flightid}$ using $\destination$. 

We assume that the read-only flight management database contains data about at 
least one flight and one city. To start the process, one needs at least one city 
to meet the aviation safety compliances. It is assumed that, initially, all the cities 
are unsafe. An airport dispatcher, at once, may change the safety status only of one 
city.

We model this action by performing two consequent actions. 
First, we select the city identifier and 
store it in the designated artifact variable $\sdid$:
\[
\small
\exists \typedvar{c}{\cityid}
  \left(c\neq \nullv {}\land  \sdid=\nullv {}\land \sdid'=c    \right)
\]
Then, we place the extracted city identifier into a unary artifact 
relation $\sfdest   : \cityidx  \longrightarrow \cityid$, that is used to represent 
safe cities and where $\cityidx$ is its artifact sort.
\[
\small
  \begin{array}{@{}l@{}}
    \exists \typedvar{i}{\cityidx}\\
    \left(
      \begin{array}{@{}l@{}}
        \sfdest[i] = \nullv \land \sdid \neq \nullv \land  \sdid'=\nullv\\
        {}\land \sfdest' = \smtlambda{j}{
                    \smtif{c}{j=i}
                      {\sdid}
                      {
                        \smtif{t}{\sfdest[j]=\sdid}
                        {\nullv}
                        {
                          \sfdest[j]                        
                        }
                      }
                  }
      \end{array}
    \right)
  \end{array}
\]

Note that two previous transitions can be rewritten as a unique one, hence showing a more
compact way of specifying \ras transitions. This, in turn, can augment the performance of the 
verifier while working with large-scale cases. The unified transition actually looks as follows:

\[
\small
  \begin{array}{@{}l@{}}
    \exists \typedvar{c}{\cityid},\exists \typedvar{i}{\cityidx}\\
    \left(
      \begin{array}{@{}l@{}}
      c\neq \nullv {}\land
        \sfdest[i] = \nullv
        \\
        {}\land \sfdest' = \smtlambda{j}{
                    \smtif{c}{j=i}
                      {c}
                      {
                        \smtif{t}{\sfdest[j]=c}
                        {\nullv}
                        {
                          \sfdest[j]                        
                        }
                      }
                  }
      \end{array}
    \right)
  \end{array}
\]

Then, to register passengers with booked tickets on a flight, the airline needs to make sure that a corresponding flight destination is 
actually safe. 
To perform the passenger registration, the airline selects a flight identifier that is assigned to the route 
and uses it to populate entries in an unary artifact relation  $\regpasger  : \passidx  \longrightarrow \flightid$. 
Note that there may be more than one passenger taking the flight, and therefore, more than one entry in $\regpasger$
with the same flight identifier. 
\[
\small
  \begin{array}{@{}l@{}}
    \exists \typedvar{i}{\cityidx},\typedvar{f}{\flightid},\typedvar{p}{\passidx}\\
    \left(
      \begin{array}{@{}l@{}}
        f\neq \nullv \land \destination(f)=\sfdest[i] \land \regpasger[p]=\nullv\\
        {}\land \regpasger' = \smtlambda{j}{
                    \smtif{c}{j=p}
                      {f}
                      {
                          \regpasger[j]                        
                      }
                  }
      \end{array}
    \right)
  \end{array}
\]

We also assume that the airline owns aircraft of one type that can contain no more than $k$ passengers. 
In case there were more than $k$ passengers registered on the flight, the airline receives a notification about its 
overbooking and temporary suspends all passenger registrations associated to this flight. 
This is modelled by checking whether there are at least $k+1$ entries in $\regpasger$. If so, the flight identifier is
added to a unary artifact relation $\overbooked  : \flightidx  \longrightarrow \flightid$ and all the passenger registrations
in  $\regpasger$ that reference this flight identifier are nullified by updating unboundedly many entries in the corresponding artifact relation:\footnote{For simplicity of presentation, we simply remove such data from the artifact relation. In a real setting, this information would actually be transferred to a dedicated, historical table, so as to reconstruct the status of past, overbooked flights.}
\[
\small
  \begin{array}{@{}l@{}}
    \exists \typedvar{p_1}{\passidx},\ldots\typedvar{p_{k+1}}{\passidx}, \typedvar{m}{\flightidx}\\
    \left(
      \begin{array}{@{}l@{}}
        \bigwedge_{i,i' \in \set{1,\ldots,k+1}, i\neq i'}\left(p_i\neq p_{i'} {}\land \regpasger[p_i]\neq\nullv  
        \land  \regpasger[p_i]=\regpasger[p_{i'}]\right)\\
         {}\land \overbooked[m]=\nullv\\
        {}\land \regpasger' = \smtlambda{j}{
                    \smtif{c}{\regpasger[j]=\regpasger[p_1]}
                     {\nullv}{\regpasger[j]}             
                      
                  }\\
		{}\land \overbooked'[m] = \regpasger[p_1]                  
      \end{array}
    \right)
  \end{array}
\]
Notice that this transition is not local, since its guard contains literals of the form $\regpasger[p_i]=\regpasger[p_{i'}]$ (with 
$p_i\neq p_{i'}$), which involve more than one element of one artifact sort.

In case of any contingency, the airport dispatcher may change the city status from \emph{safe} to \emph{unsafe}. 
To do it, we first select one of the safe cities, make it unsafe (i.e., remove it from $\sfdest$ relation) and 
store its identifier in the artifact variable $\usdid$:
\[
\footnotesize
\exists \typedvar{i}{\cityidx}
  \left(\usdid=\nullv {}\land \sfdest[i]\neq\nullv  {}\land \usdid'=\sfdest[i] {}\land \sfdest'[i]=\nullv\right)
\]

Then, we use the remembered city identifier to cancel all the passenger registrations for flights that use this city as their destination:\footnote{Similarly to the previous case, the corresponding transition performs the intended action by updating unboundedly many entries in the artifact relation.}
\[
\small
  \begin{array}{@{}l@{}}
    \left(
      \begin{array}{@{}l@{}}
        \usdid \neq \nullv \land \usdid'=\nullv \\
        {}\land \regpasger' = \smtlambda{j}{
                    \smtif{c}{\destination(\regpasger[j])=\usdid}
                      {\nullv}
                      {
                        {
                          \regpasger[j]                        
                        }
                  }}
      \end{array}
    \right)
  \end{array}
\]

Also in this case, we can shrink the transitions into a single transition:
\[
\footnotesize
\exists \typedvar{i}{\cityidx}
  \begin{array}{@{}l@{}}
    \left(
      \begin{array}{@{}l@{}}
      \sfdest[i]\neq\nullv 
        {}\land \regpasger' = \smtlambda{j}{
                    \smtif{c}{\destination(\regpasger[j])=\sfdest[i]}
                      {\nullv}
                      {
                        {
                          \regpasger[j]                        
                        }
                  }}
      \end{array}
    \right)
  \end{array}
\]

However, as in the previous case, the transition turns out to be not local. Specifically, it is due to the literal $\destination(\regpasger[j])=\sfdest[i]$ that involves more than one element with different artifact sorts.

\section{Proofs and Complements for Section~3}
\label{app_sec2}

We fix a signature $\Sigma$ and a universal theory $T$ as in Definition~\ref{def:db}.

Observe that if $\Sigma$ is acyclic, there are only finitely many terms involving a single variable $x$: in fact, there are as many terms as paths in $G(\Sigma)$
starting from the sort of $x$. If $k_\Sigma$ is the maximum number of terms involving a single variable, then (since all function symbols are unary) there are at most
$k_\Sigma^n$ terms involving $n$ variables.

\vskip 2mm
\noindent
\textbf{Proposition~\ref{prop:fmp}}.~\emph{$T$ has the finite model property in  case $\Sigma$ is acyclic.}
\vskip 2mm

\begin{proof}
If $T:=\emptyset$, then congruence closure ensures that the finite model property holds and decides  constraint satisfiability in time $O(n\log n)$ \cite{bradley}.

Otherwise, we reduce the argument to the Herbrand Theorem.  Indeed, suppose to have a set $\Phi$ of universal formulae. Herbrand Theorem states that $\Phi$ has a model iff the set of ground instances of $\Phi$ has a model. These ground instances are finitely many by acyclicity, so we can reduce to the case where $T$ is empty.
%
%
%
\end{proof}

\begin{remark} {\rm
If $T$ is finite, Proposition~\ref{prop:fmp} ensures decidability of constraint satisfiability. In order to obtain a decision procedure, it is sufficient to instantiate the axioms of $T$ and the axioms of equality (reflexivity, transitivity, symmetry, congruence) and to use a SAT-solver to decide constraint satisfiability. Alternatively, one can decide constraint satisfiability via congruence closure \cite{bradley} and avoid instantiating the equality axioms.
}
\end{remark}

\begin{remark} {\rm
Acyclity is  a  strong condition, often too strong. However, some condition must be imposed (otherwise we have undecidability, and then failure of finite model property, by reduction to word problem for finite presentations of monoids).
In fact, the empty theory and the theory  axiomatized by axiom~\ref{eq:null} both have the finite model property even without acyciclity assumptions.
}
\end{remark}

\begin{remark} {\rm
 It is evident from the above proof that {Proposition~\ref{prop:fmp}} still holds whenever $n$-ary relation symbols are added to the signature, so it applies also to the extended DB-theories considered in Definition~\ref{def:extdb}.
 }
\end{remark}

\vspace{3mm}

 We recall some basic definitions and notions from logic and model theory. We focus on the definitions of diagram, embedding, substructure and amalgamation.

 We adopt the usual first-order syntactic notions of signature, term,
 atom, (ground) formula, sentence, and so on.

 Let $\Sigma$ be a first-order signature. The signature
 obtained from $\Sigma$ by adding to it a set $\ua$ of new constants
 (i.e., $0$-ary function symbols) is denoted by $\Sigma^{\ua}$. We indicate by $\vert \cA\vert$ the support of a $\Sigma$-structure $\cA$: this is the disjoint union of the sets $S^\cA$, varying $S$ among the sort symbols of $\cA$.
Analogously, given a $\Sigma$-structure $\cA$, the signature $\Sigma$ can be expanded to a new signature $\Sigma^{|\cA|}:=\Sigma\cup \{\bar{a}\ |\ a\in |\cA| \}$ by adding a set of new constants $\bar{a}$ (the \textit{name} for $a$), one for each element $a$ in $\cA$, with the convention that two distinct elements are denoted by different "name" constants. $\cA$ can be expanded to a $\Sigma^{|\cA|}$-structure $\cA^{\prime}:=(\cA, a)_{a\in \vert\cA\vert}$ just interpreting the additional costants over the corresponding elements. From now on, when the meaning is clear from the context, we will freely use the notation  $\cA$ and  $\cA^{\prime}$ interchangeably: in particular, given a $\Sigma$-structure $\cM$ and a $\Sigma$-formula $\phi(\ux)$ with free variables that are all in $\ux$, we will write, by abuse of notation,  $\cA\models \phi(\ua)$ instead of $\cA^{\prime}\models \phi(\bar{\ua})$.


A {\it $\Sigma$-homomorphism} (or, simply, a
homomorphism) between two $\Sigma$-structu\-res $\cM$ and
$\cN$ is any mapping $\mu: \vert \cM \vert \lra \vert \cN\vert $ among the
support sets $\vert \cM \vert $ of $\cM$ and $\vert \cN \vert$  of $\cN$ satisfying the condition
\begin{equation}\label{eq:emb}
\cM \models \varphi \quad \Rightarrow \quad \cN \models \varphi
\end{equation}
for all $\Sigma^{\vert \cM\vert}$-atoms $\varphi$ (here $\cM$ is regarded as a
$\Sigma^{\vert \cM\vert}$-structure, by interpreting each additional constant $a\in
\vert \cM\vert $ into itself and $\cN$ is regarded as a $\Sigma^{\vert \cM\vert}$-structure by
interpreting each additional constant $a\in \vert \cM\vert $ into $\mu(a)$).
In case condition \eqref{eq:emb} holds for all $\Sigma^{|\cM|}$-literals,
the homomorphism $\mu$ is said to be an {\it
        embedding} and if it holds for all
first order
formulae, the embedding $\mu$ is said to be {\it
        elementary}.
Notice the following facts:
\begin{description}
        \item[(a)] since we have equality in the
        signature, an embedding is an injective function;
        \item[(b)] an embedding $\mu:\cM\longrightarrow \cN$ must be an algebraic homomorphism,
        that is for every $n$-ary function symbol $f$ and for every $m_1,..., m_n$ in $|\cM|$, we must have
        $f^{\cN}(\mu(m_1), ... , \mu(m_n)) = \mu(f^{\cM}(m_1, ... , m_n))$;
        \item[(c)] for an $n$-ary predicate symbol $P$ we must have $(m_1, ... , m_n) \in P^{\cM}$ iff  $(\mu(m_1), ... ,
        \mu(m_n)) \in P^{\cN}$.
\end{description}
It is easily seen that an embedding  $\mu:\cM\longrightarrow \cN$ can be equivalently
defined as a map  $\mu:\vert\cM\vert\longrightarrow\vert \cN\vert$ satisfying the conditions (a)-(b)-(c) above.
If  $\mu: \cM \lra \cN$ is an embedding which is just the
identity inclusion $\vert \cM\vert\subseteq\vert \cN\vert$, we say that $\cM$ is a {\it
        substructure} of $\cN$ or that $\cN$ is an {\it extension} of
$\cM$. A $\Sigma$-structure $\cM$ is said to be \emph{generated by}
a set $X$ included in its support $\vert \cM\vert$ iff there are no proper substructures of $\cM$ including $X$.

The notion of substructure can be equivalently defined as follows: given a $\Sigma$-structure $\cN$ and a $\Sigma$-structure $\cM$ such that $|\cM| \subseteq |\cN|$, we say that $\cM$ is a $\Sigma$-\textit{substructure} of $\cN$ if:

\begin{itemize}
        \item for every function symbol $f$ inf $\Sigma$, the interpretation of $f$ in $\cM$ (denoted using $f^{\cM}$) is the restriction of the interpretation of $f$ in $\cN$ to $|\cM|$ (i.e. $f^{\cM}(m)=f^{\cN}(m)$ for every $m$ in $|\cM|$); this fact implies that a substructure $\cM$ must be a subset of $\cN$ which is closed under the application of $f^{\cN}$.
        \item  for every relation symbol $P$ in $\Sigma$ and every tuple $(m_1,...,m_n)\in |\cM|^{n}$, $(m_1,...,m_n)\in P^{\cM}$ iff $(m_1,...,m_n)\in P^{\cN}$, which means that the relation $P^{\cM}$ is the restriction of $P^{\cN}$ to the support of $\cM$.
\end{itemize}

We recall that a substructure \textit{preserves} and \textit{reflects} validity of ground formulae, in the following sense:
given a $\Sigma$-substructure $\cA_1$ of a $\Sigma$-structure $\cA_2$, a ground $\Sigma^{\vert\cA_1\vert}$-sentence $\theta$ is true in $\cA_1$ iff $\theta$ is true in $\cA_2$.

Let $\cA$ be a $\Sigma$-structure. The \textit{diagram} of $\cA$, denoted by $\Delta_{\Sigma}(\cA)$, is defined as the set of ground $\Sigma^{|\cA|}$-literals (i.e. atomic formulae and negations of atomic formulae) that are true in $\cA$.  For the sake of simplicity, once again by abuse of notation, we will freely say that  $\Delta_{\Sigma}(\cA)$ is the set of $\Sigma^{|\cA|}$-literals which are true in $\cA$.

An easy but nevertheless important basic result, called
\emph{Robinson Diagram Lemma}~\cite{CK},
says that, given any $\Sigma$-structure $\cB$,  the embeddings $\mu: \cA \longrightarrow \cB$ are in bijective correspondence with
expansions of $\cB$
to   $\Sigma^{\vert \cA\vert}$-structures which are models of
$\Delta_{\Sigma}(\cA)$. The expansions and the embeddings are related in the obvious way: $\bar a$ is interpreted as $\mu(a)$.

Amalgamation is a classical algebraic concept. We give the formal definition of this notion.
\begin{definition}[Amalgamation]
        \label{def:amalgamation}
        A theory $T$ has the \emph{amalgamation property} if for every couple of embeddings $\mu_1:\cM_0\longrightarrow\cM_1$, $\mu_2:\cM_0\longrightarrow\cM_2$ among models of $T$, there exists a model $\cM$ of
        $T$ endowed with embeddings $\nu_1:\cM_1 \longrightarrow \cM$ and
        $\nu_2:\cM_2 \longrightarrow \cM$ such that $\nu_1\circ\mu_1=\nu_2\circ\mu_2$
        \begin{center}
                \begin{tikzcd}
                        & \cM \arrow[rd, leftarrow, "\nu_2"] & \\
                        \cM_{1} \arrow[ru, "\nu_1"] \arrow[rd, leftarrow, "\mu_1"]	& & \cM_{2} &\\
                        & \cM_{0} \arrow[ur, "\mu_2"] &
                \end{tikzcd}
        \end{center}
\end{definition}
The triple   $(\cM, \mu_1,\mu_2)$ (or, by abuse,   $\cM$  itself) is said to be a $T$-amalgama of  $\cM_1,\cM_2$ over $\cM_0$

The following Lemma gives a useful folklore technique for finding model completions:

\begin{lemma}\label{lem:qe}
 Suppose that for every primitive $\Sigma$-formula $\exists x\, \phi(x, \uy)$ it is possible to find a quantifier-free formula $\psi(\uy)$ such that
 \begin{description}
  \item[{\rm (i)}] $T\models \forall x\, \forall \uy\,( \phi(x, \uy) \to \psi(\uy))$;
  \item[{\rm (ii)}] for every model $\cM$ of $T$, for every tuple of  elements $\ua$ from the support of $\cM$ such that $\cM\models \psi(\ua)$ it is possible to find
  another model $\cN$ of $T$ such that $\cM$ embeds into $\cN$ and $\cN\models \exists x \phi(x, \ua)$.
 \end{description}
  Then $T$ has a model completion $T^*$ axiomatized by the infinitely many sentences~\footnote{Notice that our $T$ is assumed to be universal according to Definition~\ref{def:db}, whereas $T^*$ turns out to be universal-existential. }
  \begin{equation}\label{eq:qe}
  \forall \uy \,(\psi(\uy) \to \exists x\, \phi(x, \uy))~.
  \end{equation}
\end{lemma}
\begin{proof}
 From {\rm (i)} and (\ref{eq:qe}) we clearly get that $T^{\star}$ admits quantifier elimination: in fact, in order to prove that a theory enjoys quantifier elimination, it is sufficient to teliminate quantifiers from \textit{primitive} formulae (then the quantifier elimination for all formulae can be easily shown by an induction over their complexity). This is exactly what is guaranteed by {\rm (i)} and (\ref{eq:qe}).

 Let $\cM$ be a model of $T$. We show (by using a chain argument) that there exists a model $\cM^{\prime}$ of $T^{\star}$ such that $\cM$ embeds into $\cM^{\prime}$. For every primitive formula $\exists x \phi(x,\uy)$, consider the set $\{(\ua, \exists x\phi(x, \ua))\}$ such that $\cM\models \psi(\ua)$ (where $\psi$ is related to $\phi$ as in {\rm (i)-(ii)}). 
 By Zermelo's Theorem, the set  $\{(\ua, \exists \ue\,\phi(\ue, \ua))\}$ can be well-ordered: let $\{(\ua_i,\exists \ue\,\phi_i(\ue, \ua_i))\}_{i\in I}$ be such a well-ordered set (where $I$ is an ordinal). By transfinite induction on this well-order, we define $\cM_{0}:=\cM$ and, for each $i\in I$, 
 $\cM_{i}$ as the extension of $\bigcup_{j<i}\cM_j$ such that $\cM_{i}\models\exists \ue\, \phi_i(\ue, \uy)$, which exists for {\rm (ii)}   since $\bigcup_{j<i}\cM_j\models \psi_i(\ua)$ (remember that validity of ground formulae is preserved passing through substructures and superstructures, and $\cM_0\models \psi_i(\ua)$).

 Now we take the chain union $\cM^{1}:=\bigcup_{i\in I} \cM_{i}$: since $T$ is universal, $\cM^{1}$ is again a model of $T$, and it is possible to construct an analogous chain $\cM^{2}$ as done above, starting from $\cM^{1}$ instead of $\cM$. Clearly, we get $\cM_{0}:=\cM\subseteq \cM^{1} \subseteq \cM^2$ by construction. At this point, we iterate the same argument countably many times, so as to define a new chain of models of $T$: $$\cM_{0}:=\cM\subseteq \cM^{1}\subseteq...\subseteq \cM^n\subseteq...$$
 
 Defining $\cM^{\prime}:=\bigcup_n \cM^{n}$, we trivially get that $\cM^{\prime}$ is a model of $T$ such that $\cM\subseteq \cM^{\prime}$ and satisfies all the sentences of type (\ref{eq:qe}). The last fact can be shown using the following finiteness argument.
 
 Fix $\phi, \psi$ as in~\eqref{eq:qe}.
 For every tuple $\ua^{\prime}\in \cM^{\prime}$ such that $\cM^{\prime}\models \psi(\ua^{\prime})$, by definition of $\cM^{\prime}$ there exists a natural number $k$ such that $\ua^{\prime}\in \cM^{k}$: since $\psi(\ua^{\prime})$ is a ground formula, we get that also $\cM^{k}\models \psi(\ua^{\prime})$. Therefore, we consider the step $k$ of the countable chain: there, we have that the pair $(\ua^{\prime},\psi(\ua^{\prime}))$ appears in the enumeration given by the well-ordered set 
 of pairs
 $\{(\ua_i,\exists \ue\, \phi_i(\ue,\ua_i))\}_{i\in I}$ (for some ordinal $I$) such that $\cM^{k}\models \psi_i(\ua)$. Hence, by construction and since $\psi(\ua^{\prime})$ is a ground formula, we have that there exists a $j\in I$ such that
 $\cM^{k}_{j}\models \exists \ue\,\phi(\ue,\ua^{\prime})$.
 In conclusion, since the existential formulae are preserved passing to extensions, we obtain $\cM^{\prime}\models\exists \ue\,\phi(\ue,\ua^{\prime})$, as wanted.
\end{proof}

\vskip 2mm
\noindent
\textbf{Proposition~\ref{prop:mc}}.~\emph{$T$ has a model completion in  case it is axiomatized by universal one-variable formulae and $\Sigma$ is  acyclic.}
\vskip 2mm

\begin{proof}
We freely take inspiration from an analogous result in~\cite{wheeler}.  We preliminarly show that $T$ is amalgamable. Then, for a suitable choice of $\psi$ suggested by the acyclicity assumption, the amalgamation property will be used to prove the validy of the condition {\rm (ii)} of Lemma~\ref{lem:qe}: this fact (together with condition {\rm (i)}) yields that $T$ has a model completion which is axiomatized by the infinitely many sentences (\ref{eq:qe}).

Let $\cM_1$ and $\cM_2$ two models of $T$ with a submodel $\mathcal{M}_0$ of $T$ in common (we suppose for simplicity that $|\cM_1|\cap|\cM_2|=|\cM_0|)$. We define a $T$-amalgam $\cM$ of $\cM_1, \cM_2$ over $ \cM_0$ as follows (we use in an essential way the fact that $\Sigma$ contains only \textit{unary} function symbols).\footnote{ 
Adding $n$-ary relations symbols would not compromize the argument either. 
}
Let the support of $\cM$ be the set-theoretic union of the supports of $\cM_1$ and $\cM_2$, i.e. $|\cM|:= |\cM_1|\cup |\cM_2|$. $\cM$ has a natural $\Sigma$-structure inherited by the $\Sigma$-structures $\cM_1$ and $\cM_2$: for every function symbol $f$ in $\Sigma$, we define, for each $m_i\in |\cM_i| (i=1,2)$, $f^{\cM}(m_i):=f^{\cM_1}(m_i)$, i.e. the interpretation of $f$ in $\cM$ is the restriction of the interpretation of $f$ in $\cM_i$ for every element $m_i\in |\cM_i|$. This is well-defined since, for every $a\in |\cM_1|\cap|\cM_2|=|\cM_0|$, we have that $f^{\cM}(a):=f^{\cM_1}(a)=f^{\cM_0}(a)=f^{\cM_2}(a)$. It is clear that $\cM_1$ and $\cM_2$ are substructures of $\cM$, and their inclusions agree on $\cM_0$.

We show that the $\Sigma$-structure $\cM$, as defined above, is a model of $T$. By hypothesis, $T$ is axiomatized by universal one-variable formulae: so, we can consider $T$ as a theory formed by axioms $\phi$ which are universal closures of clauses with just one variable, i.e. $\phi:=\forall x (A_1(x)\wedge...\wedge A_n(x)\rightarrow B_1(x)\vee...\vee B_m(x))$, where $A_j$ and $B_k$ ($j=1,...,n$ and $k=1,...,m$) are atoms.

We show that $\cM$ satisfies all such formulae $\phi$. In order to do that, suppose that, for every $a\in |\cM|$, $\cM\models A_j(a)$ for all $j=1,...,n$. If $a\in |\cM_i|$, then $\cM\models A_j(a)$ implies $\cM_i\models A_j(a)$, since $A_j(a)$ is a ground formula. Since $\cM_i$ is model of $T$ and so $\cM_i\models \phi$, we get that $\cM_i\models B_k(a)$ for some $k=1,...,m$, which means that $\cM\models B_k(a)$, since $B_k(a)$ is a ground formula. Thus, $\cM\models\phi$ for every axiom $\phi$ of $T$, i.e. $\cM\models T$ and, hence, $\cM$ is a $T$-amalgam of $\cM_1, \cM_2$ over $ \cM_0$, as wanted

Now, given a primitive formula $\exists x \phi(x,\uy)$, we find a suitable $\psi$ such that the hypothesis of Lemma~\ref{lem:qe} holds. We define $\psi(\uy)$ as the conjunction of the set of all quantifier-free $\chi(\uy)$-formulae such that $\phi(x,\uy)\rightarrow \chi(\uy)$ is a logical consequences of $T$ (they are finitely many - up to $T$-equivalence - because $\Sigma$ is acyclic). By definition, clearly we have that {\rm (i)} of Lemma~\ref{lem:qe} holds.

 We show that also condition {\rm (ii)} is satisfied. Let $\cM$ be a model of $T$ such that $\cM\models \psi(\ua)$ for some tuple of elements $\ua$ from the support of $\cM$. Then, consider the $\Sigma$-substructure $\cM[\ua]$ of $\cM$ generated by the elements $\ua$: this substructure is finite (since $\Sigma$ is acyclic), it is a model of $T$ and we trivially have that $\cM[\ua]\models \psi(\ua)$, since $\psi(\ua)$ is a ground formula. In order to prove that there exists an extension $\cN^{\prime}$ of $\cM[\ua]$ such that $\cN\models \exists x\phi(x,\ua)$, it is sufficient to prove (by the Robinson Diagram Lemma) that the $\Sigma^{|\cM[\ua]|\cup\{e\}}$-theory $\Delta(\cM[\ua])\cup\{\phi(e,\ua)\}$ is $T$-consistent. For reduction to absurdity, suppose that the last theory is $T$-inconsistent. Then, there are finitely many literals $l_1(\ua),..., l_m(\ua)$ from $\Delta(\cM[\ua])$ (remember that $\Delta(\cM[\ua])$ is a finite set of literals since $\cM[\ua]$ is a finite structure) such that $\phi(e,\ua)\models_T \neg(l_1(\ua)\wedge...\wedge l_m(\ua))$. Therefore, defining $A(\ua):=l_1(\ua)\wedge...\wedge l_m(\ua)$, we get that $\phi(e,\ua)\models_T \neg A(\ua)$, which implies that $\neg A(\ua)$ is one of the $\chi(\uy)$-formulae appearing in $\psi(\ua)$. Since $\cM[\ua]\models \psi(\ua)$, we also have that $\cM[\ua]\models \neg A(\ua)$, which is a contraddiction: in fact, by definition of diagram, $\cM[\ua]\models A(\ua)$ must hold. Hence, there exists an extension $\cN^{\prime}$ of $\cM[\ua]$ such that $\cN^{\prime}\models\exists x\phi(x,\ua)$.
Now, by amalgamation property, there exists a $T$-amalgam $\cN$ of $\cM$ and $\cN^{\prime}$ over $\cM[\ua]$: clearly, $\cN$ is an extension of $\cM$ and, since $\cN^{\prime}\hookrightarrow \cN$ and $\cN^{\prime}\models\exists x\phi(x,\ua)$, also $\cN\models\exists x\phi(x,\ua)$ holds, as required.

\end{proof}

\begin{remark}\label{rem:qe}
  {\rm The proof of Proposition~\ref{prop:mc} gives an algorithm for quantifier elimination
in the model completion. The algorithm works as follows (see the formula~\eqref{eq:qe}): to eliminate the quantifier $x$ from
$\exists x\,\phi(x, \uy)$ take the conjunction of
  the  clauses  $\chi(\uy)$ implied by $\phi(x, \uy)$. This algorithm is far from optimal from two points
of view. First, contrary to what happens in linear arithmetics, the quantifier
elimination needed to prove Proposition~\ref{prop:mc} has a much better
behaviour (from the complexity point of view) if obtained via a suitable
version of the Knuth-Bendix procedure~\cite{BaNi98}. Since these aspects
concerning quantifier elimination are rather delicate, we address them in a
dedicated paper~\cite{new} (our \mcmt implementation, however, already
partially takes into account such future development).

Secondly, the algorithm presented in Appendix~\ref{app_sec2} uses the
acyclicity assumption, whereas such assumption is in general not needed for
Proposition~\ref{prop:mc} to hold: for instance, when $T:=\emptyset$ or when
$T$ contains only Axiom~\eqref{eq:null}, a model completion can be proved to
exist, even if $\Sigma$ is not acyclic, by using the Knuth-Bendix version of
the quantifier elimination algorithm.}
\end{remark}

\section{Proofs of Theorem~5.1}\label{app_sas}

In this section we present Theorems~\ref{thm:basic-1} and ~\ref{thm:basic-2} that constitute the proof of Theorem~\ref{thm:basic} from Section~\ref{sec:termination}. 

First, we specify the definition of \ras in the particular case of \sas. Given a DB schema
$\tup{\Sigma,T}$ and a tuple $\ux=x_1,\dots,x_n$ of variables, we consider the
following classes of $\Sigma$-formulae:
\begin{compactitem}[\textbf{--}]
\item a \emph{state formula} is a quantifier-free $\Sigma$-formula $\phi(\ux)$;
\item an \emph{initial formula} is a conjunction of equalities of the form
  $\bigwedge_{i=1}^n x_i=c_i$, where each $c_i$ is a
  constant;\footnote{Typically, $c_i$ is an $\nullv$ constant mentioned above.}
\item a \emph{transition formula} $\hat\tau$ is an existential formula
  \begin{equation}
    \label{eq:transition1}
    \textstyle
    \exists \uy \left(G(\ux, \uy) \wedge \bigwedge_{i=1}^n x'_i = F_i(\ux,
      \uy)\right)
  \end{equation}
  where $\ux'$ are renamed copies of $\ux$, $G$ is quantifier-free and
  $F_1,\ldots,F_n$ are case-defined functions.  We call $G$ the \emph{guard}
  and $F_i$ the \emph{updates} of Formula~\eqref{eq:transition1}.
\end{compactitem}

In view of Definition~\ref{def:ras}, we have:
\begin{definition}
  A \emph{\SAS} (\sas) has the form
  \[
    \cS ~=~\tup{\Sigma, T, \ux, \iota(\ux), \tau(\ux, \ux')}
  \]
  where:
  \begin{inparaenum}[\itshape (i)]
  \item $\tup{\Sigma,T}$ is a (read-only) DB schema,
  \item $\ux=x_1, \dots, x_n$ are variables (called \emph{artifact variables}),
  \item $\iota$ is an initial formula, and
  \item $\tau$ is a disjunction of transition formulae.
  \end{inparaenum}
\end{definition}

\vskip 2mm
\noindent
\begin{theorem}\label{thm:basic-1}
Let $\tup{\Sigma,T}$ be a DB schema. Then, for any a \sas $\mathcal S$ with $\tup{\Sigma,T}$ as its DB schema, backward search algorithm is effective and partially correct for solving safety problems for $\mathcal S$. If, in addition, $\Sigma$ is acyclic, backward search terminates and decides safety problems for $\mathcal S$.
\end{theorem}
\begin{proof} In the case of \sas, formula~\eqref{eq:smc1} has the following form
\begin{equation}\label{eq:smc}
           \iota(\ux^0) \wedge \tau(\ux^0, \ux^1) \wedge \cdots \wedge\tau(\ux^{k-1}, \ux^k)\wedge \upsilon(\ux^k)~~.
\end{equation}
By definition, $\mathcal S$ is unsafe iff for some $n$, the formula~\eqref{eq:smc} is satisfiable in a DB-instance of $\tup{\Sigma,T}$. Thanks to Assumption~\ref{ass}, $T$ has the finite model property and consequently, as~\eqref{eq:smc} is an existential $\Sigma$-formula,  $\mathcal S$ is unsafe iff for some $n$, formula~\eqref{eq:smc} is satisfiable in a model of $T$; furthermore, again by Assumption~\ref{ass},  $\mathcal S$ is unsafe iff for some $n$, formula~\eqref{eq:smc} is satisfiable in a model of $T^*$.
Thus, we shall concentrate on satisfiability in models of $T^*$ in order to prove the Theorem.

Let us call $B_n$ (resp. $\phi_n$) the status of the variable $B$ (resp. $\phi$) after $n$ executions in line 4 (resp. line 6) of Algorithm ~\ref{alg1}.
Notice that we have
 $T^*\models \phi_{j+1}\leftrightarrow Pre(\tau, \phi_j)$ for all $j$ and that
\begin{equation}\label{eq:invariant}
 T\models B_n \leftrightarrow \bigvee_{0\leq j<n} \phi_j
\end{equation}
 is an invariant of the algorithm.

Since we are considering satisfiability in models of $T^*$, we can apply quantifier elimination and so the satisfiability of~\eqref{eq:smc} is equivalent to the satisfiability of
$\iota\wedge \phi_n$: this is a quantifier-free formula (because in line 6 of Algorithm~\ref{alg1}), whose satisfiability (wrt $T$ or equivalently wrt $T^*$)\footnote{$T$-satisfiability and $T^*$-satisfiability
are equivalent, by the definition of $T^*$, as far as existential (in particular, quantifier-free) formulae are concerned.}
is decidable by Assumption 1, so if Algorithm~\ref{alg1} terminates with an
 $\mathsf{unsafe}$ outcome, then $\mathcal S$ is really unsafe.

 Now consider the satisfiability test in line 2. This is again a satisfiability test for a quantifier-free formula, thus it is decidable. In case of a $\mathsf{safe}$ outcome, we have that $T\models \phi_n\to B_n$; this means that, if we could continue executing the loop of Algorithm~\ref{alg1}, we would nevertheless get
 $T^*\models B_m\leftrightarrow B_n$ for all $m\geq n$.\footnote{ In more detail:
 recall the invariant~\eqref{eq:invariant} and that
 $T^*\models \phi_{j+1}\leftrightarrow Pre(\tau, \phi_j)$ holds for all $j$.
 Thus, from $T\models \phi_n\to B_n$, we get $T\models\phi_{n+1} \to Pre(\tau,B_n)$; since $Pre$ commutes with disjunctions, we have $T^*\models Pre(\tau,B_n)\leftrightarrow \bigvee_{1\leq j\leq n} \phi_j$. Now (using $T\models \phi_n\to B_n$ again), we get  $T^*\models\phi_{n+1} \to B_n$, that is $T^*\models B_{n+1}\leftrightarrow B_n$.
 Since then $T^*\models \phi_{n+1} \to B_{n+1}$, we can repeat the argument for all $m\geq n$.
 } This would entail that $\iota\wedge \phi_m$ is always unsatisfiable (because of~\eqref{eq:invariant} and because $\iota\wedge \phi_j$ was unsatisfiable
 for all $j<n$), which is the same (as remarked above) as saying that all formulae~\eqref{eq:smc} are unsatisfiable. Thus $\mathcal S$ is safe.

 In case $\Sigma$ is acyclic, there are only finitely many 
 quantifier-free formulae (in which the finite set of variables $\ux$ occur), so it is evident that the algorithm must terminate:
 because of~\eqref{eq:invariant},
 the unsatisfiability test of Line 2 must eventually succeed, if the unsatisfiability test of Line 3 never does so. 
\end{proof}

\vskip 2mm
For complexity questions, we have the following result:

\begin{theorem}\label{thm:basic-2}
  Let $\Sigma$ be an acyclic DB signature and $\tup{\Sigma,T}$ a DB schema
  built on top of it. Then, for every \sas
  $\mathcal S=\langle \Sigma,T,\ux,\iota,\tau\rangle$, deciding safety problems
  for $\mathcal S$ is in PSPACE in the size of $\ux$, of $\iota$ and of $\tau$.
\end{theorem}

\begin{proof}
We need to modify Algorithm~\ref{alg1} (we make
it nondeterministic and use Savitch's Theorem saying that PSPACE $=$ NPSPACE).

Since $\Sigma$ is acyclic, there are only finitely many terms involving a single variable, let this number be $k_\Sigma$ (we consider $T,\Sigma$ and hence $k_\Sigma$ constant for our problems).
Then, since all function symbols are unary, it is clear that we have at most $2^{O(n^2)}$ conjunctions of sets of literals involving at most $n$ variables and that if the system is unsafe, unsafety can be detected
with a run whose length is at most $2^{O(n^2)}$. Thus we introduce a counter to be incremented during the main loop (lines 2-6) of  Algorithm~\ref{alg1}.
The fixpoint test in line 2 is removed and loop is executed only until the maximum length of an unsafe run is not exceeded (notice that an exponential counter requires polynomial space).

Inside the loop, line 4 is removed (we do not need anymore the variable $B$) and line 6 is modified as follows.
%
We replace line 6 of the algorithm by
\[
6'.~~~   \phi\longleftarrow \alpha(\ux);
\]
where $\alpha$ is a non-deterministically chosen conjunction of literals
implying $\mathsf{QE}(T^*,\phi)$. Notice that to check the latter, there is no need to compute $\mathsf{QE}(T^*,\phi)$: recalling the proof of Proposition~\ref{prop:mc}   and Remark~\ref{rem:qe} it is sufficient to check that $T\models \alpha\to C$ holds for every clause $C(\ux)$ such that
$T\models\phi \to C$.

The algorithm is now in PSPACE, because all the satisfiability
tests we need are, as a consequence of the proof of Proposition~\ref{prop:fmp}, in NP: all such tests are reducible to $T$-satisfiability tests for quantifier-free $\Sigma$-formulae involving the variables $\ux$ and the additional (skolemized) quantified variables occurring in the transitions \footnote{ For the
test in line 3,
we just need replace in $\phi$ the $\ux$ by their values given by $\iota$, conjoin the result with all the ground instances of the axioms of $T$ and finally decide satisfiability with congruence closure algorithm of a polynomial size ground conjunction of literals.}. In fact, all these satisfiability tests are applied to formulae whose length is polynomial in the size of $\ux$, of $\iota$ and of $\tau$.
\end{proof}
\vskip 2mm

The proof of Theorem~\ref{thm:basic} shows that, whenever $\Sigma$ is not
acyclic, backward search is still a semi-decision procedure: if the system is
unsafe, backward search always terminates and discovers it; if the system is
safe, the procedure can diverge (but it is still correct).

\section{Proof of Theorem~4.2}\label{app_sec4b}

The technique used for proving Theorem~\ref{thm:nonsimple} is similar to that used in~\cite{ijfcs} (but here we have to face some additional complications, due to the fact that our quantifier elimination is not directly available, it is only indirectly available via model completions).

When introducing our transition formulae in~\eqref{eq:trans1} we made use of definable extensions
 and also of some function definitions via $\lambda$-abstraction. We already observed that such uses are due to notational convenience and do not really go beyond first-order logic. We are clarifying one more point now, before going into formal proofs.
 The lambda-abstraction definitions in~\eqref{eq:trans1} will make the proof of Lemma~\ref{lem:eq1} below smooth.
 Recall that an expression like
 \[
 b = \lambda y. F(y,\uz)
 \]
 can be seen as a mere abbreviation of $\forall y~b(y)=F(y,\uz)$.
 However, the use of such abbreviation makes clear that e.g. a formula like
 \[
  \exists b~( b = \lambda y. F(y,\uz) \wedge \phi(\uz, b))
 \]
 is equivalent to
 \begin{equation}\label{eq:aux}
 \phi(\uz, \lambda y. F(y,\uz)/b)~~.
 \end{equation}
 Since our $\phi(\uz, b)$ is in fact a first-order formula, our $b$ can occur in it only in terms like $b(t)$, so that in~\eqref{eq:aux} all
 occurrences of $\lambda$ can be eliminated by the so-called $\beta$-conversion: replace $\lambda y F(y,\uz)(t)$ by $F(t, \uz)$. Thus, in the end, either we use definable extensions or definitions via lambda abstractions, \emph{the formulae we manipulate can always be converted into plain first-order $\Sigma$- or $\ext{\Sigma}$-formulae}.

Let us call \emph{extended state formulae} the formulae of the kind
$ \exists \ue~ \phi(\ue, \ux,\ua)$,
where $\phi$ is quantifier-free and the $\ue$ are individual variables of both artifact and basic sorts.

\begin{lemma}\label{lem:eq1}
 The preimage of an extended state formula is logically equivalent to an extended state formula.
\end{lemma}
\begin{proof} We manipulate the formula
\begin{equation}\label{eq:pre}
 \exists \ux'\,\exists \ua'\, (\tau(\ux, \ua, \ux', \ua') \wedge \exists \ue~ \phi(\ue, \ux',\ua'))
\end{equation}
up to logical equivalence, where $\tau$ is given by\footnote{Actually, $\tau$ is a disjunction of such formulae, but it easily seen that disjunction can be
accommodated by moving existential quantifiers back-and-forth through them.}
\begin{equation}\label{eq:trans2}
         \exists \ue_0\left(\gamma(\ue_0, \ux, \ua) \wedge
            \ux'= \uF(\ue_0, \ux, \ua) \wedge \ua'=\lambda y. \uG(y,\ue_0,\ux, \ua)
           \right)
\end{equation}
(here we used plain equality  for conjunctions of equalities, e.g. $\ux'= \uF(\ue_0, \ux, \ua)$ stands for $\bigwedge_i x'_i= F_i(\ue, \ux, \ua)$).
Repeated substitutions  show that~\eqref{eq:pre} is equivalent to
\begin{equation}
 \exists \ue\,\exists \ue_0\, \left(\gamma(\ue_0, \ux, \ua) \wedge \phi(\ue, \uF(\ue_0, \ux,\ua)/\ux',\lambda y.\uG(y,\ue_0,\ux, \ua)/\ua' )\right)
\end{equation}
which is an extended state formula.
\end{proof}
\vskip 2mm

\begin{lemma}\label{lem:eq2}
 For every extended state formula there is a state formula equivalent to it in all $\ext{\Sigma}$-models of $T^*$.
\end{lemma}
\begin{proof} Let $ \exists \ue\,\exists \uy~ \phi(\ue,\uy, \ux,\ua)$, be an extended state formula, where $\phi$ is quantifier-free, the $\ue$ are variables whose sort is an artifact sort and the $\uy$ are variables whose sort is a basic sort.

Now observe that, according to our definitions,  the artifact components have an artifact sort as source sort and a basic sort as target sort; since equality
is the only predicate, the literals in $\phi$ can be divided into equalities/inequalities between variables from $\ue$ and literals where the $\ue$ can only occur as arguments of an artifact component.
Let $\ua[\ue]$ be the tuple of the terms among the terms of the kind $a_j[e_s]$ which are well-typed; using disjunctive normal forms,
our extended state formula can be written as a disjunction of formulae of the kind
\begin{equation}\label{eq:s}
 \exists \ue\,\exists \uy~(\phi_1(\ue) \wedge \phi_2(\uy, \ux,\ua[\ue]/\uz))
\end{equation}
where $\phi_1$ is a conjunction of equalities/inequalities, $\phi_2(\uy,\ux,\uz)$ is a quantifier-free $\Sigma$-formula and $\phi_2(\uy,\ux, \ua[\ue]/\uz)$ is obtained from  $\phi_2$ by replacing the variables $\uz$ by the terms $\ua[\ue]$. Moving inside the existential quantifiers $\uy$, we can rewrite~\eqref{eq:s} to
\begin{equation}\label{eq:s1}
 \exists \ue~(\phi_1(\ue) \wedge \,\exists \uy\,\phi_2(\uy, \ux,\ua[\ue]/\uz))
\end{equation}
Since $T^*$ has quantifier elimination, we have that
there is $\psi(\ux,\uz)$ which is equivalent to $\exists \uy\,\phi_2(\uy, \ux,\uz))$ in all models of $T^*$;
thus in all $\ext{\Sigma}$-models of $T^*$, the formula~\eqref{eq:s1} is equivalent to
\[
\exists \ue~(\phi_1(\ue) \wedge \,\psi( \ux,\ua[\ue]/\uz))
\]
which is a state formula.
\end{proof}
\vskip 2mm
We underline that Lemmas~\ref{lem:eq1} and~\ref{lem:eq2} both give an explicit effective procedure
for computing equivalent (extended) state formulae. Used one after the other, such procedures extends the procedure $QE(T^*, \phi)$ in line 6 of Algorithm~\ref{alg1} to (non simple) artifact systems. Thanks to such procedure, the only formulae we need to test for satisfiability in lines 2 and 3 of  the backward reachability algorithm are
the $\exists\forall$-formulae introduced below.

Let us call $\exists\forall$-formulae the formulae of the kind 
\begin{equation}\label{eq:s02}
 \exists \ue\; \forall \ui \; \phi(\ue, \ui, \ux, \ua)
\end{equation}
where the variables $\ue, \ui$ are variables whose sort is an artifact sort and $\phi$ is quantifier-free.
The crucial point for the following lemma to hold is that the \emph{universally} quantified variables in $\exists\forall$-formulae are all of artifact sorts:

\begin{lemma}\label{lem:sat}
 The satisfiability of a $\exists\forall$-formula in a $\ext{\Sigma}$-model of $T$ is decidable; moreover, a $\exists\forall$-formula is satisfiable in a $\ext{\Sigma}$-model of $T$ iff it is satisfiable in a DB-instance of $\tup{\ext{\Sigma},T}$ iff it is satisfiable in a $\ext{\Sigma}$-model of $T^*$.
\end{lemma}
\begin{proof}
 First of all, notice that a $\exists\forall$-formula~\eqref{eq:s02} is equivalent to a disjunction of formulae of the kind
 \begin{equation}\label{eq:s03}
 \exists \ue\; ({\rm Diff}(\ue) \wedge \forall \ui \; \phi(\ue, \ui, \ux, \ua))
\end{equation}
where ${\rm Diff}(\ue)$ says that any two variables of the same sort from the $\ue$ are distinct (to this aim, it is sufficient to guess a partition and to keep, via a substitution,
only one element for each equivalence class).\footnote{In the MCMT implementation, state formulae are always maintained so that all existential variables occurring in  them are differentiated, so that there is no need of this expensive computation step.} So we can freely assume that $\exists\forall$-formulae are all of the kind~\eqref{eq:s03}.

Now, by the way $\ext{\Sigma}$ is built, the only atoms occurring in $\phi$ whose arguments involve terms of artifact sorts are of the kind $e_s=e_j$, so all such atoms can be 
replaced either by $\top$ or by $\bot$ (depending on whether we have $s=j$ or not). So we can assume that there are no such atoms in $\phi$ and as a result, the variables $\ue$, $\ui$ can only occur as arguments of the $\ua$.

Let us consider now the set of all (sort-matching) substitutions $\sigma$ mapping the $\ui$ to the $\ue$.
The formula~\eqref{eq:s03} is satisfiable (respectively: in a $\ext{\Sigma}$-model of $T$, in a DB-instance of $\tup{\ext{\Sigma},T}$, in a $\ext{\Sigma}$-model of $T^*$) iff so it is the formula
 \begin{equation}\label{eq:s04}
 \exists \ue\; ({\rm Diff}(\ue) \wedge \bigwedge_{\sigma}\phi(\ue, \ui\sigma, \ux, \ua))
\end{equation}
(here $\ui\sigma$ means the componentwise application of $\sigma$ to the $\ui$): this is because, if~\eqref{eq:s04} is satisfiable in $\cM$, then we can take as $\cM'$ the
same $\ext{\Sigma}$-structure as $\cM$, but with the interpretation of the artifact sorts restricted only to the elements named by the $\ue$ and get in this way a $\ext{\Sigma}$-structure $\cM'$ satisfying~\eqref{eq:s03} (notice that $\cM'$ is still a DB-instance of $\tup{\ext{\Sigma}, T}$ or a $\ext{\Sigma}$-model of $T^*$, if so was $\cM$).
Thus, we can freely concentrate on the satisfiability problem of formulae of the kind~\eqref{eq:s04} only.

Let now $\ua[\ue]$ be the tuple of the terms among the terms of the kind $a_j[e_s]$ which are well-typed. Since in~\eqref{eq:s04} the $\ue$ can only occur as arguments of the artifact components, as observed above, the formula~\eqref{eq:s04} is in fact of the kind
\begin{equation}\label{eq:s05}
 \exists \ue\; ({\rm Diff}(\ue) \wedge \psi(\ux, \ua[\ue]/\uz))
\end{equation}
where $\psi(\ux,\uz)$ is a quantifier-free $\Sigma$-formula and $\psi(\ux, \ua[\ue]/\uz)$ is obtained from  $\psi$ by replacing the variables $\uz$ by the terms $\ua[\ue]$
(notice that the $\uz$ are of basic sorts because the target sorts of the artifact components are basic sorts).

It is now evident that~\eqref{eq:s05} is satisfiable  (respectively: in a $\ext{\Sigma}$-model of $T$, in a DB-instance of $\tup{\ext{\Sigma},T}$, in a $\ext{\Sigma}$-model of $T^*$) iff the formula
\begin{equation}\label{eq:s06}
  \psi(\ux,\uz)
\end{equation}
is satisfiable (respectively: in a $\Sigma$-model of $T$, in a DB-instance of $\tup{\Sigma,T}$, in a $\Sigma$-model of $T^*$). In fact, if we are given a $\Sigma$-structure
$\cM$ and an assignment satisfying~\eqref{eq:s06}, we can easily expand $\cM$  to a $\ext{\Sigma}$-structure by taking the $e$'s themselves as the elements of the interpretation of the artifact sorts;  in the so-expanded $\ext{\Sigma}$-structure, we can interpret the artifact components $\ua$ by taking  the $\ua[\ue]$  to be the elements assigned to the $\uz$ 
 in the satisfying assignment for~\eqref{eq:s06}.

Thanks to Assumption~\ref{ass},
the satisfiability of~\eqref{eq:s06}  in a $\Sigma$-model of $T$, in a DB-instance of $\tup{\Sigma,T}$, or in a $\Sigma$-model of $T^*$
are all equivalent and decidable.
\end{proof}
\vskip 2mm
The instantiation algorithm of Lemma~\ref{lem:sat} can be used to discharge the satisfiability tests in lines 2 and 3 of Algorithm~\ref{alg1} because the conjunction of a state formula and of the negation of a state formula is a $\exists\forall$-formula (notice that
$\iota$ is itself the negation of a state formula, according to the definition of an \textit{initial} formula in \ras.

\vskip 2mm\noindent
\textbf{Theorem~\ref{thm:nonsimple}}~\emph{The backward search algorithm (cf. Algorithm~\ref{alg1}), applied to artifact systems, is effective and partially correct.}
\vskip 2mm
\begin{proof}
 Recall that $\mathcal S$ is unsafe iff there is
 no DB-instance $\cM$ of $\tup{\ext{\Sigma}, T}$, no $k\geq 0$ and no assignment in $\cM$ to the variables $\ux^0,\ua^0 \dots, \ux^k, \ua^k$ such that the formula~\eqref{eq:smc1}
         \[
           \iota(\ux^0, \ua^0) \wedge \tau(\ux^0,\ua^0, \ux^1, \ua^1) \wedge \cdots \wedge\tau(\ux^{k-1},\ua^{k-1},
           \ux^k,\ua^{k})\wedge \upsilon(\ux^k,\ua^{k})
         \]
         is true in $\cM$. It is sufficient to show that this is equivalent to saying that there is
 no $\ext{\Sigma}$-model $\cM$ of $T^*$, no $k\geq 0$ and no assignment in $\cM$ to the variables $\ux^0,\ua^0 \dots, \ux^k, \ua^k$ such that~\eqref{eq:smc1} is true in $\cM$
 (once this is shown, the proof goes in the same way as the proof of Theorem~\ref{thm:basic}).

 Now, the formula~\eqref{eq:smc1} is satisfiable in a $\ext{\Sigma}$-structure $\cM$ under a suitable assignment iff the formula
         \begin{eqnarray*}
           \iota(\ux^0, \ua^0) ~~\wedge
           & \exists \ua^1\exists\ux^1 (\tau(\ux^0,\ua^0, \ux^1, \ua^1) \wedge \cdots ~~~~~~~~~~~~~~~~~~~~~~~~~~~~~~~~~~~~~~~~
           \\ &
           \cdots
           \wedge \exists \ua^k\exists \ux^k(\tau(\ux^{k-1},\ua^{k-1},
           \ux^k,\ua^{k})\wedge \upsilon(\ux^k,\ua^{k}))\cdots)
         \end{eqnarray*}
         is satisfiable in $\cM$ under a suitable assignment; by Lemma~\ref{lem:eq1}, the latter is equivalent to a formula of the kind
         \begin{equation}\label{eq:a}
           \iota(\ux, \ua)~\wedge~\exists \ue\,\exists \uz\,\phi(\ue,\uz,\ux, \ua)
         \end{equation}
         where $\exists \ue\,\exists \uz\,\phi(\ue,\uz,\ux, \ua)$ is an extended state formula (thus $\phi$ is quantifier-free, the $\ue$
         are variables of artifact sorts and the $\uz$ are variables of  basic sorts  - we renamed $\ux^0, \ua^0$ as $\ux, \ua$).
         However the satisfiability of~\eqref{eq:a} is the same as the satisfiability of
         $\exists \ue\,(\iota(\ux,\ua) \wedge\phi(\ue,\uz,\ux, \ua))$;
         the latter, in view of the definition of \textit{initial} formula in \ras, is
         a $\exists\forall$-formula and so
         Lemma~\ref{lem:sat} applies and shows that its satisfiability in a DB-instance of $\tup{\ext{\Sigma}, T}$ is the same as its satisfiability in
         a $\ext{\Sigma}$-model of $T^*$.
\end{proof}

We remark that all the results in this Section (in particular, Theorem~\ref{thm:nonsimple}) \emph{hold also in case the read-only database is modeled via an extended DB-theory} (see Definition~\ref{def:extdb}) satisfying  Assumption~\ref{ass}.

\section{Proof of Termination Results: local updates and tree-like settings}
\label{app_sec5}

We begin by recalling some basic facts about well-quasi-orders.  Recall that a
\emph{well-quasi-order} (wqo) is a set $W$ endowed with a reflexive-transitive
relation $\leq$ having the following property: for every infinite succession
\[
w_0, w_1, \dots, w_i, \dots
\]
of elements from $W$ there are $i, j$ such that $i<j$ and $w_i\leq w_j$.

The fundamental result about wqo's is the following, which is a consequence of
the well-known Kruskal's Tree Theorem~\cite{kruskal}:

\begin{theorem} If $(W,\leq)$ is a  wqo, then so is the partial order of the finite lists over $W$, ordered by componentwise subword comparison (i.e. $w\leq w'$ iff there is a subword $w_0$ of $w'$ of the same length as $w$, such that the i-th entry of $w$ is
        less or equal to---in the sense of $(W,\leq)$---the $i$-th entry of $w_0$, for all $i=0, \dots \vert w\vert$).
\end{theorem}

Various  wqo's can be recognized by applying the above theorem; in particular, the theorem implies that the cartesian product of wqo's is a wqo.
As an application, notice that $\mathbb N$ is a wqo, hence the following
corollary (known as Dikson Lemma) follows:

\begin{corollary}
 The cartesian product of $k$-copies of $\mathbb N$ (and also of $\mathbb N \cup \{\infty\}$), with componentwise ordering, is a wqo.
\end{corollary}

Let us now turn to the terminology introduced in Section\ref{sec:termination}  and in particular to the numbers $k_1(\cM), \dots, k_N(\cM)\in \mathbb N\cup\{\infty\}$ counting the numbers of elements generating (as singletons) the cyclic substructures $\cC_1, \dots, \cC_N$, respectively (we assume the acyclicity of $\Sigma$ and consequently also of $\tilde \Sigma$).

\begin{lemma}\label{lem:localinv}
  Let $\cM, \cN$ be $\tilde \Sigma$-structures. If the inequalities
  \[
    k_1(\cM)\leq k_1(\cN), \dots, k_N(\cM)\leq k_N(\cN)
  \]
  hold, then all local formulae true in $\cM$ are also true in $\cN$.
\end{lemma}

\begin{proof}
Notice that local formulae (viewed in $\tilde\Sigma$) are sentences, because they do not have free variable occurrences - the $\ua, \ux$ are now constant function symbols and individual constants, respectively.  The proof of the lemma is fairly obvious: notice that, once we assigned some $\alpha(e_i)$ in $\cM$ to the variable $e_i$, the truth of a formula like $\phi(e_i, \ux, \ua)$ under such an assignment depends only on the $\tilde \Sigma$-substructure generated by $\alpha(e_i)$, because  $\phi$ is quantifier-free and $e_i$ is the only $\tilde \Sigma$-variable occurring in it. 
In fact, if a local state formula $ \exists e_1\cdots \exists e_k \left( \delta(e_1,\dots, e_k)  \land
  \bigwedge_{i=1}^k \phi_i(e_i,\ux,\ua)\right)$  is true in $\cM$, then there
exist elements $\bar{e}_1,\cdots, \bar{e}_k$
(in the interpretation of some artifact sorts),  each of which makes $\phi_i$ true. Hence, $\phi_i$ is also true  in the corresponding cyclic structure generated by $\bar{e}_i$. Since $k_1(\cM)\leq k_1(\cN), \dots, k_N(\cM)\leq k_N(\cN)$ hold, then also in $\cN$ there are at least as many elements in the interpretation of artifact sorts as there are in $\cM$ that validate all the $\phi_i$ . Thus, we get that $\exists e_1\cdots \exists e_k \left( \delta(e_1,\dots, e_k)  \land
    \bigwedge_{i=1}^k \phi_i(e_i,\ux,\ua)\right)$ is true also in $\cN$, as wanted.
\end{proof}

\vskip 2mm\noindent
\textbf{Theorem~\ref{thm:term1}}~\emph{If $\Sigma$ is acyclic,
the backward search algorithm (cf. Algorithm~\ref{alg1}) terminates when applied to a local safety formula in a \ras, whose transition formula is a disjunction of local transition formulae. }
\vskip 2mm
\begin{proof}
 Suppose the algorithm does not terminate. Then the fixpoint test of Line 2 fails infinitely often. Recalling that the $T$-equivalence of
 $B_n$  and of $\bigvee_{0\leq j<n} \phi_j$ is an invariant of the algorithm (here $\phi_n, B_n$ are the status of the variables 
  $\phi, B$ after $n$ execution of the main loop), this means that there are models
 \[
 \cM_0, \cM_1, \dots, \cM_i, \dots
 \]
 such that for all $i$, we have that  $\cM_i\models \phi_i$ and $\cM_i \not \models \phi_j$ (all $j<i$). But the $\phi_i$ are all local formulae, so considering the
 tuple of cardinals $k_1(\cM_i), \dots, k_N(\cM_i)$ and Lemma~\ref{lem:localinv}, we get a contradiction, in view of Dikson Lemma. This is because, by Dikson Lemma, $(\mathbb N \cup \{\infty\})^N$ is a wqo, so there exist $i$, $j$ such that $j<i$ and $k_1(\cM_j)\leq k_1(\cM_i), \dots, k_N(\cM_j)\leq k_N(\cM_i)$. Using Lemma~\ref{lem:localinv}, we get that $\phi_j$, which is local and true in $\cM_j$, is also true in $\cM_i$, which is a contradiction.
\end{proof}
\vskip 2mm

Proving termination for \ras with a tree-like artifact setting
 is more complex, but follows a similar schema as in the case of local transition formulae.

	If $(W,\leq)$ is a partial order, we consider the set $M(W)$ of finite multisets of $W$  as a partial order in the following way:\footnote{This is not the canonical ordering used for multisets, as introduced  eg in~\cite{BaNi98}.} say that $M\leq N$ holds iff there is an injection $p:M\longrightarrow N$ such that $m\leq p(m)$ holds for all $m\in M$ 
	(in other words, $p$ associates with
	every occurrence of $m$ an occurrence $p(m)$ of an element  of $N$ so that different occurrences are associated to different occurrences). 
	
	\begin{corollary}\label{coro:multiset}
		If $(W,\leq)$ is a wqo, then so is $(M(W), \leq)$ as defined above.
	\end{corollary}
	
	\begin{proof}
		This is due to the fact that one can convert a multiset $M$ to a list $L(M)$ so that if $L(M)\leq L(N)$ holds, then also $M\leq N$ holds (such a conversion $L$ can be obtained by ordering the occurrences of elements in $M$ in any arbitrarily chosen way).
	\end{proof}
	
	\vskip 2mm
	
	We assume that the graph $G(\tilde \Sigma)$ associated to $\tilde \Sigma$ is a tree (the generalization to the case where such a graph is a forest is trivial). 
	This means in particular that  each sort is the domain of at most one function symbol and that there just one sort which is not the domain of any function symbol
	(let us call it the \emph{root sort} of $\tilde \Sigma$ and let us denote it with $S_r$).
	
	By induction on the height\footnote{This is defined as the length of the longest path from $S$ to a leaf.} of a sort $S$ in the above graph, we  define a wqo $w(S)$ (in the definition we use the fact the cartesian product of wqo's is a wqo and Corollary~\ref{coro:multiset}).
	Let $S_1, \dots, S_n$ be the sons of $S$ in the tree; put
	\begin{equation}\label{eq:sortwqo}
	w(S)~:=~M(w(S_1))\times \cdots \times  M(w(S_n))
	\end{equation}
	(thus, if $S$ is a leaf, $w(S)$ is the trivial one-element wqo - its only element is the empty tuple).
	
	Let now $\cM$ be a finite $\tilde\Sigma$-structure; we indicate with $S^\cM$ the interpretation in $\cM$ of the sort $S$ (it is a finite set). For $a\in S^\cM$, we define the multiset $M_{\cM}(a)\in w(S)$, again by induction on the height of $S$. Suppose that
	$S_1, \dots, S_n$ are the sons of $S$ and that the arc from $S_i$ to $S$ is labeled by the function symbol $f_i$; then we put
	\begin{align*}
	M_{\cM}(a)~:=~ & \langle\{ M_{\cM}(b_1) \mid b_1\in S_1^\cM ~{\rm and}~f^\cM_1(b_1)=a\}, \dots
	\\ & \dots,
	\{ M_{\cM}(b_n) \mid b_n\in S_n^\cM ~{\rm and}~f^\cM_n(b_n)=a\}\rangle
	\end{align*}
	where $f_i^\cM$ ($i=1, \dots, n$) is the interpretation of the symbol $f_i$ in $\cM$.
	
	Moreover, for every sort $S$, we let
	\begin{equation}\label{eq:setwqo}
	M_{\cM}(S)~:=~\{ M_{\cM}(a) \mid a\in S^\cM\}~~~~.
	\end{equation}
	Finally, we define
	\begin{equation}\label{eq:modelwqo}
	M(\cM)~:=~M_{\cM}(S_r)~~~.
	\end{equation}
	For termination, the relevant lemma is the following:
	\begin{lemma} Suppose that $\tilde \Sigma$ is tree-like and does not contain constant symbols;
		given two finite $\tilde \Sigma$-structures $\cM$ and $\cN$, we have that if $M(\cM)\leq M(\cN)$, then $\cM$ embeds into $\cN$.
	\end{lemma}
	\begin{proof}
		Again, we make an induction on the height of $S$, proving the claim for the subsignature of $\tilde \Sigma$ having $S$ as a root
		(let us call this the $S$-subsignature).
		
		Let $\cM$ be a model over the $S$-subsignature. For every $a\in S^{\cM}$, and for every
		$f_i:S_i\longrightarrow S$, if we restrict
		$\cM$ to the elements in the $f_i$-fibers of $a$, we get a model $\cM_{f_i,a}$ for the $S_i$-subsignature (an element
		$c\in \tilde S^\cM$ is in the $f_i$-fiber of $a$ if, taking the term $t$ corresponding to the composition of the functions
		symbols going from
		$\tilde S$ to $S_i$, we have that $f_i^\cM(t^\cM(c))=a$). In addition, if $M_{\cM}(a)=(M_1, \dots, M_n)$, then
		$M_i=M(\cM_{f_i,a})$ by definition. Finally, observe that the restriction of $\cM$ to the $S_i$-subsignature is the disjoint union of the $f_i$-fibers models $\cM_{f_i,a}$, varying $a\in S^\cM$.

		Suppose now that $\cM,\cN$ are models over the $S$-subsignature such that $M(\cM)\leq M(\cN)$; this means that we can find an injective map $\mu$ mapping $S^{\cM}$ into
		$S^{\cN}$ so that $M_\cM(a)\leq M_\cN(\mu(a))$. If $M_{\cM}(a)=(M_1, \dots, M_n)$ and $M_{\cN}(\mu(a))=(N_1, \dots, N_n)$,
		we then have that $M_i\leq N_i$ for every $i=1, \dots, n$. Considering that, as noticed above, $M_i=\cM_{f_i,a}$ and
		$N_i=\cN_{f_i,\mu(a)}$, by induction hypothesis, we have  embeddings $\nu_{i,a}$ for  the $f_i$-fibers models of
		$a$ and $\mu(a)$ (for every $a\in S^\cM$ and $i=1, \dots,n$). Glueing these embeddings to the disjoint union (varying $i,a$)
		and adding them $\mu$ as $S$-component, we get the desired embedding of $\cM$ into $\cN$.
	\end{proof}
	
	\begin{proposition}\label{prop:wqotreelike}
		If $\tilde \Sigma$ is tree-like and does not contain constant symbols, then the finite $\tilde\Sigma$-structures are a wqo with respect to the embeddability quasi-order.
	\end{proposition}
	
	\begin{proof}
		An immediate consequence of the previous lemma. 
	\end{proof}

\vskip 2mm\noindent
\textbf{Theorem~\ref{thm:term2}}~\emph{ Backward search (cf.\ Algorithm~\ref{alg1}) terminates when applied to a
  safety problem in a \ras with a tree-like artifact setting.}
  \vskip 2mm
	\begin{proof} For simplicity, we give the argument for the case where we do not have constants and artifact variables (but see the footnote below for the general case).   Similarly to the proof of Theorem~\ref{thm:term1}, suppose the algorithm does not terminate. Then the fixpoint test of Line 2 fails infinitely often. Recalling that the $T$-equivalence of
		$B_n$  and of $\bigvee_{0\leq j<n} \phi_j$ is an invariant of the algorithm (here $\phi_n, B_n$ are the status of the variables 
		$\phi, B$ after $n$ execution of the main loop), this means that there are models
		\[
		\cM_0, \cM_1, \dots, \cM_i, \dots
		\]
		such that for all $i$, we have that  $\cM_i\models \phi_i$ and $\cM_i \not \models \phi_j$ (all $j<i$).
		The models can be taken to be all finite, by Lemma~\ref{lem:sat}.
		But the $\phi_i$ are all existential
		sentences in $\tilde \Sigma$, so this is
		incompatible to the fact that, by Proposition~\ref{prop:wqotreelike}, there are $j<i$ with $\cM_j$ embeddable into $\cM_i$.\footnote{
			The following observation shows how to extend the proof to the case where we have constants and artifact variables. Recall that in $\tilde \Sigma$ the artifact variables are seen as constants, so we need to consider only the case of constants. Let $\tilde \Sigma^+$ be $\tilde \Sigma$ where each constant symbol $c$ of sort $S$ is replaced by a new sort $S_c$ and a new function symbol $f_c: S_c \longrightarrow S$. Now every model $\cM$  of $\tilde \Sigma$ can be transformed into a model $\cM^+$ of $\tilde \Sigma^+$ by interpreting $S_c$
			as a singleton set $\{*\}$ and $f_c$  as the map sending $*$ to $c^\cM$. This transformation has the following property: $\tilde \Sigma$-embeddings of $\cM$ into $\cN$ are in bijective correspondence with $\tilde \Sigma^+$-embeddings of $\cM^+$ into $\cN^+$.
			Since $\tilde \Sigma^+$ is still tree-like and does not have constant symbols, this shows that Theorem~\ref{thm:term2} holds for $\tilde \Sigma$ too.
		}
	\end{proof}

\section{Complements for Section~5}
\label{app_sec4}

 Fix an acyclic signature $\Sigma$ and an artifact setting $(\ux, \ua)$ over it.
 In this section we analyze in our setting the transition formulae studied in~\cite{verifas}\footnote{For simplicity, since we are not considering hierarchical aspects, we assume that there is no input variable in the sense of \cite{verifas}} 
 (deletion, insertion and propagation updates). In addition, we discuss some modifications of the previous transitions and introduce new kinds of updates (like bulk updates). We prove that all these transitions are strongly local transitions.

\subsection{Deletion Updates}
\label{sec:deletion-updates}
We want to remove a tuple $\underline{t}:=(t_{1},...,t_{m})$ from an $m$-ary artifact relation $R$ and assign the values $t_{1},...,t_{m}$ to some of the artifact variables (let $\ux:=\ux_1,\ux_2$, where  $\ux_{1}:=(x_{i_{1}},...,x_{i_{m}})$ are the variables where we want to transfer the tuple $\underline{t}$). This operation has to be applied only if the current artifact variables $\ux$ satisfy the pre-condition $\pi(\ux_1, \ux_2)$ and the updated artifact variables $\ux^{\prime}:=\ux_1^{\prime}, \ux_2^{\prime}$ satisfy the post-condition $\psi(\ux_1^{\prime}, \ux_2^{\prime})$ ($\pi$ and $\psi$ are quantifier-free formulae). The variables $\ux_2$ are not propagated, i.e. they are non deterministically reassigned.
Let $\underline{r}:=r_1,...,r_m$ be the artifact components of $R$.
 Such an update can be formalized in a symbolic way as follows:
\begin{equation}\label{eq:del}
\exists \underline{d}\,\exists e \begin{pmatrix}
\pi(\ux_1, \ux_2)\;\wedge\; \psi(\ux_1^{\prime}, \ux_2^{\prime})
\; \wedge r_1[e]\neq \nullv\wedge...\\ \wedge\; r_n[e]\neq \nullv
\wedge (\ux_1^{\prime}:=\underline{r}[e]\;\wedge \;\ux_2^{\prime}:=\underline{d}\wedge \underline{s}^{\prime}:=\underline{s}\;\wedge \\
\;\wedge\; \underline{r}^{\prime}:=\lambda j.(\mathtt{ if }\  j=e \mathtt{~ then ~\nullv ~else ~} \underline{r}[j]))\end{pmatrix}
\end{equation}
where $\underline{s}$ are the artifact components of the artifact relations different from $R$. Notice that the $\underline{d}$ are non deterministically produced values for the updated $\ux^{\prime}_2$. In the terminology of \cite{verifas}, notice that no artifact variable is propagated in a deletion update.

Notice that in place of the condition $r_1[e]\neq \nullv\wedge... \wedge\; r_n[e]\neq \nullv$ one can consider the modified deletion update that is fired only if \textit{some} (and not all) artifact components are not $\nullv$, or even the case when the transition is fired if \textit{at least one} artifact component is not $\nullv$: the latter case can be expressed using a disjunction of transitions $\tau_i$ that, instead of $r_1[e]\neq \nullv\wedge... \wedge\; r_n[e]\neq \nullv$, involve only the literal $r_i[e]\neq \nullv$ (for $i=1,...,n$). These modified deletion updates can be proved to be strongly local transitions by using trivial adaptations of the arguments shown below.

The formula \eqref{eq:del} is not in the format \eqref{eq:trans1} but can be easily converted into it as follows:

\begin{equation}\label{eq:del1}
\exists \underline{d}\,\exists e \begin{pmatrix}\pi(\ux_1, \ux_2)\; \wedge \; \psi(\underline{r}[e], \underline{d}) \; \wedge \; r_1[e]\neq \nullv\; \wedge...\\ \wedge \; r_n[e]\neq \nullv
\; \wedge \; (\ux_1^{\prime}:=\underline{r}[e]\; \wedge\;  \ux_2^{\prime}:=\underline{d}\; \wedge\; \underline{s}^{\prime}:=\underline{s}\; \wedge \\ \wedge\;  \underline{r}^{\prime}:=\lambda j.(\mathtt{ if }\  j=e \mathtt{~ then ~ \nullv ~ else ~} \underline{r}[j]))\end{pmatrix}
\end{equation}

 We prove that the preimage along \eqref{eq:del1} of a strongly local formula is strongly local.
 Consider a strongly local formula
 \[
   K:=\psi^{\prime}(\ux)\wedge\exists \underline{e} \left(
     \text{Diff}(\underline{e})  \wedge \bigwedge_{e_r \in \underline{e}}
     \phi_{e_r}( \underline{r}[e_r])\wedge\Theta\right)
   \]
 where  $\Theta$ is a formula involving the artifact components $\underline{s}$ (which are not updated) such that no $e_r$ occurs in it.

 \begin{remark}{\rm
 Notice that equality is the only predicate, so a quantifier-free formula $\phi(e,\ua)$ involving a single variable $e$ must be obtained from atoms of the kind $b[e]=b'[e]$ (for $b, b'\in \ua$) by applying the Boolean connectives only: this is why we usually display such a formula as $\phi(\ua[e])$. In addition,  since the source sorts of the different artifact relations are different, we cannot employ the same variable as argument of artifact components of different artifact relations: in other words, we cannot employ the same variable $e$ in terms like $r_i[e]$ and $s_j[e]$,  in case  $r_i$ and $s_j$ are components of two different artifact relation  $R$ and $S$ (because $e$ must have either type $R$ or type $S$). Thus, the quantifier-free subformula $\phi_i(\ua[e_i])$ in a local formula involving only the variable $e_i$ must be of the kind $\phi_i(\underline{r}[e_i])$, for some artifact relation $R$  (here $\underline{r}$ are the artifact components of $R$). These observations will be often used in the sequel.}
 \end{remark}

 We compute the preimage $Pre(\ref{eq:del1},K)$

 \[
   \exists \underline{d} \,\exists e, \underline{e}\,\exists \ux_1^{\prime}, \ux_2^{\prime}\,\exists \underline{r}^{\prime}\begin{pmatrix}
        \pi(\ux_1,\ux_2)\; \wedge \; \psi(\underline{r}[e], \underline{d}) \; \wedge\; \psi^{\prime}(\ux_1^{\prime}, \ux_2^{\prime}) \; \wedge \\ \wedge\; \ux_1^{\prime}:=\underline{r}[e]\; \wedge \; \ux_2^{\prime}:=\underline{d}\; \wedge\;  \text{Diff}(\underline{e}) \; \wedge \; \bigwedge_{e_r \in \underline{e}} \phi_{e_r}( \underline{r}^{\prime}[e_r])\; \wedge \\
 \wedge \; \underline{r}^{\prime}:=\lambda j.(\mathtt{ if ~} j=e  \mathtt{ ~then ~ \nullv ~else ~} \underline{r}[j])\wedge\Theta\end{pmatrix}
 \]
 which can be rewritten as a disjunction of the following  formulae:

\begin{itemize}
\item
$\exists \underline{d} \,\exists e, \underline{e} \left(\begin{array}{@{}l@{}}
\text{Diff}(\underline{e},e) \; \wedge\; \pi(\ux_1,\ux_2)\; \wedge\;  \psi(\underline{r}[e], \underline{d}) \; \wedge\\ \wedge\; \psi^{\prime}(\underline{r}[e], \underline{d}) \; \wedge \; \bigwedge_{e_r \in \underline{e}} \phi_{e_r}(\underline{r}[e_r])\; \wedge \;  \Theta\end{array}\right)
$ \\
covering the case where $e$  is different from all $e_j\in\underline{e}$

 \item $
\exists \underline{d} \,\exists\underline{e} \left(\begin{array}{@{}l@{}}
\text{Diff}(\underline{e})\;  \wedge\pi(\ux_1,\ux_2)\; \wedge \; \psi(\underline{r}[e_j], \underline{d})\;  \wedge\; \psi^{\prime}(\underline{r}[e_j], \underline{d}) \; \wedge\\ \wedge\;  \bigwedge_{ e_r \in \underline{e}, e_r\neq e_j} \phi_{e_r}(\underline{r}[e_r])\wedge \phi_{e_j}(\nullv) \wedge\Theta\end{array}\right)
$ \\
covering the case where $e=e_j$, for some $e_j\in\underline{e}$
 \end{itemize}

  We can now move the existential quantifier $\exists \underline{d}$ in front of $\psi\wedge\psi^{\prime}$. We eliminate the quantifiers (applying the quantifier elimination procedure for $T^{\star}$) from the subformula $\exists \underline{d}\left(\psi(\underline{r}[e], \underline{d}) \wedge\psi^{\prime}(\underline{r}[e], \underline{d})\right)$ (or $\exists \underline{d}\left(\psi(\underline{r}[e], \underline{d}) \wedge\psi^{\prime}(\underline{r}[e], \underline{d})\right)$, resp.) obtaining a formula of the kind $\theta(\underline{r}[e])$ (or $\theta(\underline{r}[e_j]$).

 The final result is the disjunction of the formulae

 \begin{itemize}
 \item $\exists e, \underline{e} \left(\begin{array}{@{}l@{}}
\text{Diff}(\underline{e},e)\;  \wedge\; \pi(\ux_1,\ux_2)\; \wedge \; \theta(\underline{r}[e])\; \wedge \; \bigwedge_{e_r \in \underline{e}} \phi_{e_r}(\underline{r}[e_r])\; \wedge\;   \Theta
\end{array}\right)
$
\item $
\exists \underline{d}\, \exists\underline{e} \left(\begin{array}{@{}l@{}}
\text{Diff}(\underline{e})\;  \wedge\; \pi(\ux_1,\ux_2)\; \wedge  \; \theta(\underline{r}[e_j]) \; \wedge\\ \wedge  \bigwedge_{ e_r \in \underline{e}, e_r\neq e_j} \phi_{e_r}(\underline{r}[e_r])\; \wedge \; \phi_{e_j}(\nullv)\;  \wedge\; \Theta
\end{array}\right)$

\end{itemize}
 which is a strongly local formula.

 Analogous arguments show that:
 \begin{description}
 \item[(i)] transitions like Formula~\eqref{eq:del}, where the literals $ r_1[e]\neq \nullv\; \wedge... \wedge \; r_n[e]\neq \nullv$ are replaced with a generic constraint $\chi(\underline{r}[e])$;
 \item[(ii)] transitions that remove a tuple from an artifact relation (without transferring its values to the corresponding artifact variables);
 \item[(iii)] transitions that copy the the content of a tuple contained in an artifact relation to some artifact variables, non-deterministically reassigning the values of the other artifact variables;
 \item[(iv)] transitions that combine \textbf{(i)} and \textbf{(iii)}
\end{description}
  are also strongly local.

   \begin{remark}{\rm
 Notice that deletion updates with the propagation of some artifact variables $\ux_1$ (which are not allowed in~\cite{verifas} and in~\cite{DeLV16}) are \emph{not} strongly local, since the preimage of a strongly local formula can produce formulae of the form $\psi(\underline{r}[e], \ux_1)$. This preimage is \emph{still} local: however, the preimage of a local state formula through a deletion update can generate formulae of the form  $\psi(\underline{r}[e], \underline{r}[e'])$, with $e\neq e'$, destroying locality. Hence, the safety problem for a \ras equipped containing deletion updates with propagation in its transitions, is not guaranteed to terminate.
 }
 \end{remark}

\subsection{Insertion Updates}
\label{sec:insertion-updates}
We want to insert a tuple of values $\underline{t}:=(t_{1},...,t_{m})$ from the artifact variables $\ux_{1}:=(x_{i_{1}},...,x_{i_{m}})$ (let $\ux:=\ux_1,\ux_2$ as above) into an $m$-ary artifact relation $R$. This operation has to be applied only if the current artifact variables $\ux$ satisfy the pre-condition $\pi(\ux_1, \ux_2)$ and the updated artifact variables $\ux^{\prime}:=\ux_1^{\prime}, \ux_2^{\prime}$ satisfy the post-condition $\psi(\ux_1^{\prime}, \ux_2^{\prime})$. The variables $\ux$ are all not propagated, i.e. they are non deterministically reassigned. Let $\underline{r}:=r_1,...,r_m$ be the artifact components of $R$.
 Such an update can be formalized in a symbolic way as follows:
\begin{equation}\label{eq:ins}
\exists \underline{d}_1,\underline{d}_2\,\exists e \begin{pmatrix}
\pi(\ux_1, \ux_2)\; \wedge \; \psi(\ux_1^{\prime}, \ux_2^{\prime})
\; \wedge\: \underline{r}[e]=\nullv\\
\wedge \; (\ux_1^{\prime}:=\underline{d}_1\; \wedge\;  \ux_2^{\prime}:=\underline{d}_2\; \wedge\;  \underline{s}^{\prime}:=\underline{s}\; \wedge \\
\wedge\; \underline{r}^{\prime}:=\lambda j.(\mathtt{ if }\  j=e \mathtt{~ then ~} \ux_1\mathtt{~else ~} \underline{r}[j]))\end{pmatrix}
\end{equation}
where $\underline{s}$ are the artifact components of the artifact relations different from $R$. Notice that $\underline{d}_1, \underline{d}_2$ are non deterministically produced values for the updated $\ux^{\prime}_1,\ux^{\prime}_2$. In the terminology of \cite{verifas}, notice that no artifact variable is propagated in a insertion update. Notice that the following arguments remain the same even if $\underline{r}[e]=\nullv$ is replaced with a conjunction of \textit{some} literals of the form $r_j[e]=\nullv$, for some $j=1,...,m$, or even if $\underline{r}[e]=\nullv$ is replaced with a generic constraint $\chi(\underline{r}[e])$.

In this transition, the insertion of the same content in correspondence to different entries is allowed. If we want to avoid this kind of multiple insertions, the update $r^{\prime}$ must be modified as follows:

\[
  \underline{r}^{\prime}:=\lambda j.\left(\begin{array}{@{}l@{}}\mathtt{ if }\
      j=e \mathtt{~ then ~} \ux_1\mathtt{~else ~}\\ \mathtt{(if~}
      \underline{r}[j]=\ux_1 \mathtt{~ then ~ \nullv ~ else~ } \underline{r}[j]
      )\end{array}\right)
\]

The formula \eqref{eq:ins} is not in the format \eqref{eq:trans1} but can be easily converted into it as follows:

\begin{equation}\label{eq:ins1}
\exists \underline{d}_1,\underline{d}_2 \,\exists e \begin{pmatrix}
\pi(\ux_1, \ux_2)\; \wedge \; \psi(\underline{d}_1, \underline{d}_2)
\; \wedge \;\underline{r}[e]=\nullv\\
\wedge \; (\ux_1^{\prime}:=\underline{d}_1\; \wedge \; \ux_2^{\prime}:=\underline{d}_2\; \wedge \; \underline{s}^{\prime}:=\underline{s}\; \wedge \\
\wedge\; \underline{r}^{\prime}:=\lambda j.(\mathtt{ if }\  j=e \mathtt{~ then ~} \ux_1\mathtt{~else ~} \underline{r}[j]))\end{pmatrix}
\end{equation}

 We prove that the preimage along \eqref{eq:ins1} of a strongly local formula is strongly local.
 Consider a strongly local formula
 \[
   K:=\psi^{\prime}(\ux)\wedge\exists \underline{e} \left(
     \text{Diff}(\underline{e})  \wedge \bigwedge_{e_r \in \underline{e}}
     \phi_{e_r}( \underline{r}[e_r])\wedge\Theta\right)
 \]
 where  $\Theta$ is a formula involving the artifact relations $\underline{s}$ (which are not updated) such that no $e_r$ occurs in it.

 We compute the preimage $Pre(\ref{eq:ins1},K)$

\[
        \exists \underline{d}_1,\underline{d}_2 \,\exists e, \underline{e}\,\exists \ux_1^{\prime}, \ux_2^{\prime}\,\exists \underline{r}^{\prime} \begin{pmatrix}
        \pi(\ux_1,\ux_2)\; \wedge\;  \psi(\underline{d}_1, \underline{d}_2) \; \wedge\; \psi^{\prime}(\ux_1^{\prime}, \ux_2^{\prime}) \; \wedge \; \underline{r}[e]=\nullv \\  \wedge\; \ux_1^{\prime}:=\underline{d}_1\; \wedge \; \ux_2^{\prime}:=\underline{d}_2\; \wedge\;  \text{Diff}(\underline{e}) \; \wedge\;  \bigwedge_{e_r \in \underline{e}} \phi_{e_r}( \underline{r}^{\prime}[e_r])\; \wedge \\
 \wedge \; \underline{r}^{\prime}:=\lambda j.(\mathtt{ if ~} j=e_1  \mathtt{ ~then ~}\ux_1 \mathtt{ ~else ~} \underline{r}[j])\wedge\Theta\end{pmatrix}
 \]
 which can be rewritten as a disjunction of the following  formulae:
\begin{itemize}
\item
$\exists \underline{d}_1,\underline{d}_2\, \exists e, \underline{e} \left(\begin{array}{@{}l@{}}
\text{Diff}(\underline{e},e) \; \wedge\; \pi(\ux_1,\ux_2)\; \wedge \; \psi(\underline{d}_1, \underline{d}_2) \; \wedge \; \psi^{\prime}(\underline{d}_1, \underline{d}_2)\\  \wedge\; \underline{r}[e]=\nullv\; \wedge\; \bigwedge_{e_r \in \underline{e}} \phi_{e_r}(\underline{r}[e_r])\; \wedge \; \Theta\end{array}\right)
$ \\
covering the case where $e$  is different from all $e_j\in\underline{e}$

 \item $
\exists \underline{d}_1,\underline{d}_2\, \exists\underline{e} \left(\begin{array}{@{}l@{}}
\text{Diff}(\underline{e})\;  \wedge\; \pi(\ux_1,\ux_2)\; \wedge\;  \psi(\underline{d}_1, \underline{d}_2)\;  \wedge\; \psi^{\prime}(\underline{d}_1, \underline{d}_2)\;\wedge\\  \wedge\;\underline{r}[e]=\nullv\;  \wedge\;  \bigwedge_{ e_r \in \underline{e}, e_r\neq e_j} \phi_{e_r}(\underline{r}[e_r])\wedge \phi_{e_j}(\ux_1) \; \wedge\; \Theta\end{array}\right)
$ \\
covering the case where $e=e_j$, for some $e_j\in\underline{e}$.

 \end{itemize}

  We can move the existential quantifiers $\exists \underline{d}_1, \underline{d}_2$ in front of $\psi\wedge\psi^{\prime}$. We eliminate the quantifiers  (applying the quantifier elimination procedure for $T^{\star}$) from the subformula $\exists \underline{d}_1\underline{d}_2\left(\psi(\underline{d}_1, \underline{d}_2) \wedge\psi^{\prime}(\underline{d}_1, \underline{d}_2)\right)$  obtaining a ground formula $\theta$.

 The final result is a disjunction of formulae fo the kind

 \begin{itemize}
 \item $\exists e, \underline{e} \left(\begin{array}{@{}l@{}}
\text{Diff}(\underline{e},e)\;  \wedge\; \pi(\ux_1,\ux_2)\; \wedge \; \underline{r}[e]=\nullv\; \wedge\;\theta\;  \wedge \; \bigwedge_{e_r \in \underline{e}} \phi_{e_r}(\underline{r}[e_r])\; \wedge \;  \Theta
\end{array}\right)
$
\item $
\exists\underline{e} \left(\begin{array}{@{}l@{}}
\text{Diff}(\underline{e}) \; \wedge\;  \pi(\ux_1,\ux_2)\; \wedge \; \phi_{e_j}(\ux_1)\; \wedge  \; \underline{r}[e]=\nullv \;\wedge \;\theta \; \wedge\bigwedge_{ e_r \in \underline{e}, e_r\neq e_j} \phi_{e_r}(\underline{r}[e_r])\; \wedge \; \Theta
\end{array}\right)$

\end{itemize}
 which is a strongly local formula.

 Analogous arguments show that transitions that insert a tuple of values $\underline{t}:=(t_{1},...,t_{m})$ (where the values $t_j$ are taken  from the content of the artifact variables $\ux_{1}:=(x_{i_{1}},...,x_{i_{m}})$ or are \emph{constants}) into an $m$-ary artifact relation $R$ are also strongly local; in addition,  it is easy to see that ``propagation'' of variables $\ux_1$ (in the sense of the following subsection) is allowed in order to preserve strong locality of all those transitions. Notice that the transition introduced in Example~\ref{ex:hr-short}:

\[
\small
  \begin{array}{@{}l@{}}
    \exists \typedvar{i}{\appidx}\\
    \!\!\left(
      \begin{array}{@{}l@{}}
        \pstatev = \enab \land \recstate = \appreceived
        \\{}\land \appuser[i]=\nullv
        \\{}\land \pstatev' = \enab \land \recstate' = \nullv \land \cid'=\nullv
        \\{}\land
        \appjob' =
                  \smtlambda{j}{
                    \smtifinline{j=i}
                      {\jid}
                      {\appjob[j]}
                  }
        \\{}\land
        \appuser' =
                  \smtlambda{j}{
                    \smtifinline{j=i}
                      {\uid}
                      {\appuser[j]}
                  }
        \\{}\land
        \appresp' =
                  \smtlambda{j}{
                     \smtifinline{j=i}
                      {\eid}
                      {\appresp[j]}
                  }
        \\{}\land
        \appscore' =
                  \smtlambda{j}{
                    \smtifinline{j=i}
                      {\constant{-1}}
                      {\appscore[j]}
                  }
        \\{}\land
        \appres' =
                  \smtlambda{j}{
                    \smtifinline{j=i}
                      {\nullv}
                      {\appres[j]}
                  }
        \\{}\land \jid' = \nullv \land \uid' = \nullv \land \eid'=\nullv
      \end{array}
    \right)
  \end{array}
\]
 presents the described format.

We close this section with an important remark. In Appendix~\ref{sec:hiring-example}, we have seen that to forbid the insertion at different indexes of multiple identical tuples in an artifact relation, transitions break the strong locality requirement. A way to restore locality is to simply admit such repeated insertions. Notably, if one focuses on the fragment of strongly local \ras  that coincides with the model in \cite{DeLV16,verifas}, it can be shown, exactly reconstructing the same line of reasoning from \cite{DeLV16}, that \emph{verification problems (in the restricted common fragment) for artifact systems working over sets (i.e., insertions are performed over working memory without possible repetitions) and those working over multisets, are indeed equivalent}.


\subsection{Propagation Updates}

We want to propagate a tuple  $\underline{t}:=(t_{1},...,t_{m})$ of values contained in the artifact variables $\ux_{1}:=(x_{i_{1}},...,x_{i_{m}})$ (let $\ux:=\ux_1,\ux_2$) to the corresponding updated artifact variables $\ux_{1}^{\prime}$. This operation has to be applied only if the current artifact variables $\ux$ satisfy the pre-condition $\pi(\ux_1, \ux_2)$ and the updated artifact variables $\ux^{\prime}:=\ux_1^{\prime}, \ux_2^{\prime}$ satisfy the post-condition $\psi(\ux_1^{\prime}, \ux_2^{\prime})$.
Notice that in this transition no update of artifact component is involved.

\vspace{2mm}

 Such an update can be formalized in a symbolic way as follows:
\begin{equation}\label{eq:pro}
\exists \underline{d}\left(
\pi(\ux_1, \ux_2)\; \wedge \; \psi(\ux_1^{\prime}, \ux_2^{\prime})
\wedge (\ux_1^{\prime}:=\ux_1\; \wedge \; \ux_2^{\prime}:=\underline{d}\; \wedge\;  \underline{s}^{\prime}:=\underline{s})\right)
\end{equation}
where $\underline{s}$ stands for all the artifact components. Notice that the $\underline{d} $ are non deterministically produced values for the updated $\ux^{\prime}_2$. In the terminology of \cite{verifas}, notice that the artifact variables $\ux_1$ are propagated.

The formula \eqref{eq:ins} is not in the format \eqref{eq:trans1} but can be easily converted into it as follows:

\begin{equation}\label{eq:pro1}
\exists \underline{d}\left(
\pi(\ux_1, \ux_2)\; \wedge\;  \psi(\ux_1, \underline{d})
\wedge (\ux_1^{\prime}:=\ux_1\; \wedge \; \ux_2^{\prime}:=\underline{d}\; \wedge \; \underline{s}^{\prime}:=\underline{s})\right)
\end{equation}

 We prove that the preimage along \eqref{eq:pro1} of a strongly local formula is strongly local.
 Consider a strongly local formula
 \[
   K:=\psi^{\prime}(\ux)\wedge\exists \underline{e} \left(
     \text{Diff}(\underline{e}) \wedge\Theta\right)
 \]
 where  $\Theta$ is a formula involving the all artifact relations $\underline{s}$ (which are not modified in a propagation update), such that $K$ fits the format of~\eqref{eq:localstrong}.

 We compute the preimage $Pre(\ref{eq:pro},K)$

\[
        \exists \underline{d}\,\exists \ux_1^{\prime}, \ux_2^{\prime}\begin{pmatrix}
        \pi(\ux_1,\ux_2)\; \wedge \; \psi(\ux_1, \underline{d}) \; \wedge\; \psi^{\prime}(\ux_1, \ux_2^{\prime})\;  \wedge \\ \wedge \; \ux_1^{\prime}:=\ux_1\; \wedge \; \ux_2^{\prime}:=\underline{d}\; \wedge \; \text{Diff}(\underline{e}) \; \wedge\; \Theta\end{pmatrix}
 \]
 which can be rewritten as follows:

 \[
   \exists \underline{d}\,\exists \underline{e} \begin{pmatrix}
\text{Diff}(\underline{e}) \wedge\pi(\ux_1,\ux_2)\wedge \psi(\ux_1, \underline{d})\;  \wedge\\ \wedge\; \psi^{\prime}(\ux_1, \underline{d}) \; \wedge  \; \Theta\end{pmatrix}
\]

  We can move the existential quantifier $\exists \underline{d}$ in front of $\psi\wedge\psi^{\prime}$. We eliminate the quantifiers (applying the quantifier elimination procedure for $T^{\star}$) from the subformula $\exists \underline{d}( \psi(\ux_1, \underline{d}) \wedge\psi^{\prime}(\ux_1 \underline{d}))$  obtaining a formula of the kind $\theta(\ux_1)$.

 The final result is

 \[
   \exists\underline{e} \left(\begin{array}{@{}l@{}}
\text{Diff}(\underline{e})\;  \wedge\; \pi(\ux_1,\ux_2)\; \wedge \; \theta(\ux_1) \; \wedge \;  \Theta\end{array}\right)
\]
 which is a strongly local formula.

 Consider a transition that inserts constants or a non-deterministically generated new value $d^{\prime}$ (or a tuple of new values $\underline{d}^{\prime}$)
into an artifact component $r_i$ (or more than one) of an $m$-ary artifact relation $\underline{r}$, propagating all the other components and the artifact variables $\ux_1$ (with $\ux:=\ux_1,\ux_2$). Formally, this transition can be written in the following way:

\begin{equation}\label{eq:pro2}
\exists \underline{d}, d^{\prime}\, \exists e \begin{pmatrix}
\pi(\ux_1, \ux_2)\; \wedge\;  \psi(\ux_1^{\prime}, \ux_2^{\prime})\; \wedge\; \chi_1(d^{\prime})\;\wedge\;\chi_2(\underline{r}[e])\; \wedge \\ \wedge\; (\ux_1^{\prime}:=\ux_1\; \wedge \; \ux_2^{\prime}:=\underline{d}\; \wedge \; r^{\prime}_i= \lambda j.(\mathtt{ if}~ j=e~\mathtt{then}~ d^{\prime}~ \mathtt{else}~ r[j])\; \wedge\; \underline{s}^{\prime}:=\underline{s})\end{pmatrix}
\end{equation}
where $\underline{s}$ stands for all the artifact components different from $r_i$, and $\chi_1$ and $\chi_2$ are quantifier-free formulae. Notice that the $\underline{d} $ are non deterministically produced values for the updated $\ux^{\prime}_2$. In the terminology of \cite{verifas}, notice that the artifact variables $\ux_1$ are propagated.

The formula \eqref{eq:pro2} is not in the format \eqref{eq:trans1} but can be easily converted into it as follows:

\begin{equation}\label{eq:pro3}
\exists \underline{d}, d^{\prime}\, \exists e \begin{pmatrix}
\pi(\ux_1, \ux_2)\; \wedge \; \psi(\ux_1, \underline{d})\; \wedge\; \chi_1(d^{\prime})\;\wedge\;\chi_2(\underline{r}[e])
\; \wedge\\ \wedge \; (\ux_1^{\prime}:=\ux_1\; \wedge \; \ux_2^{\prime}:=\underline{d}\;  \wedge \;r^{\prime}_i= \lambda j.(\mathtt{if}~ j=e ~ \mathtt{then}~ d^{\prime} ~\mathtt{else}~ r[j])\; \wedge\; \underline{s}^{\prime}:=\underline{s})\end{pmatrix}
\end{equation}

Since $d^{\prime}$ does not occur in literals involving artifact variables, arguments analogous to the previous ones show that this transition is strongly local.

Notice that the transition (described in Example~\ref{ex:hr-short}):
\[
  \begin{array}{@{}l@{}}
  \exists \typedvar{i}{\joidx}, \typedvar{s}{\scoreval}\\
  \left(
    \begin{array}{@{}l@{}}
    \pstatev = \enab\\ {}\land \appuser[i] \neq \nullv \land \appscore[i] = \constant{-1}\\
       \recstate=\nullv \land  \recstate'=\nullv
        \land s \geq 0 \\
    {}\land \pstatev' = \enab \land \appscore'[i] = s
    \end{array}
  \right)
  \end{array}
\]
that assesses a Score to an applicant presents the structure of \eqref{eq:pro3}, so it is a strongly local transition. The same conclusion holds for the transition:
\[
\small
  \begin{array}{@{}l@{}}
  \exists \typedvar{u}{\userid}, \typedvar{j}{\jobcatid},\typedvar{e}{\empid},\typedvar{c}{\compinid}\\
  \left(
    \begin{array}{@{}l@{}}
    \pstatev = \enab \land \recstate = \nullv\\
    {}\land u \neq \nullv \land j \neq \nullv \land e \neq \nullv \land c \neq \nullv\\
    {}\land \compemp(c) = e \land \compjob(c) = j
    \\
    {}\land \pstatev' = \enab \land \recstate' = \appreceived\\
    {}\land \uid' = u \land \jid' = j \land \eid' = e \land \cid'=c
    \end{array}
  \right)
  \end{array}
\]
presented in Example~\ref{ex:hr-short}.

\subsection{Bulk Updates}


We want to unboundedly (bulk) update one (or more than one) artifact component(s) $r_i$ of one (or more than one) artifact relation(s) $\underline{r}$: if some conditions over the artifacts are satisfied for some entries, a global update that involves all those entries (inserting some constant $c_1$) is fired. In our symbolic formalism, we write:

\begin{equation}\label{eq:bulk}
\exists \underline{d}\, \begin{pmatrix}
\pi(\ux_1, \ux_2)\;\wedge\; \psi(\ux_1^{\prime}, \ux_2^{\prime})
\;
\wedge (\ux_1^{\prime}:=\ux_1\;\wedge \;\ux_2^{\prime}:=\underline{d}\; \wedge \; \underline{s}^{\prime}:=\underline{s}\;\wedge \\  \wedge\; r_1^{\prime}:=r_1\; \wedge...
\wedge\; r_i^{\prime}:=\lambda j.(\mathtt{ if }\  \kappa_1(\underline{r}[j]) \mathtt{~ then ~} c_1 \mathtt{~else ~} r_i[j]))\;\wedge...\wedge\; r_n^{\prime}:=r_n)\end{pmatrix}
\end{equation}
where $\underline{r}$ are the artifact components of an artifact relation $R$, $\underline{s}$ are the remaining artifact components, $\kappa_1$ is a quantifier-free formula\footnote{From the computations below, it is clear that strong locality holds also in case $\kappa_1$ depends also on the variables $\ux$, on the condition that $\kappa_1(\ux,\underline{r}[j])$ has the form $h_0(\ux)\wedge h_1(\underline{r}[j])$, with $h_0$ and $h_1$ quantifier-free formulae}, $c_1$ is a constant. The artifact component $r_i$ is updated in a global, unbounded way: we call this kind of update "bulk update". 

The formula \eqref{eq:bulk} is not in the format \eqref{eq:trans1} but can be easily converted into it as follows:

\begin{equation}\label{eq:bulk1}
\exists \underline{d}\, \begin{pmatrix}
\pi(\ux_1, \ux_2)\;\wedge\; \psi(\ux_1, \underline{d})
\;
\wedge (\ux_1^{\prime}:=\ux_1\;\wedge \;\ux_2^{\prime}:=\underline{d}\; \wedge \; \underline{s}^{\prime}:=\underline{s}\;\wedge \\
\wedge\; r_1^{\prime}:=r_1\; \wedge...\wedge\;r_i^{\prime}:=\lambda j.(\mathtt{ if }\  \kappa_1(\underline{r}[j]) \mathtt{~ then ~} c_1 \mathtt{~else ~} r_i[j]))\;\wedge...\wedge\; r_n^{\prime}:=r_n)\end{pmatrix}
\end{equation}

 We prove that the preimage along \eqref{eq:bulk1} of a strongly local formula is strongly local.
 Consider a strongly local formula
 \[
   K:=\psi^{\prime}(\ux)\wedge\exists \underline{e} \left(
     \text{Diff}(\underline{e})  \wedge \bigwedge_{e_r \in \underline{e}}
     \phi_{e_r}( \underline{r}[e_r])\wedge\Theta\right)
 \]
where  $\Theta$ is a formula involving the artifact relations $\underline{s}$ (which are not updated) such that no $e_r$ occurs in it.

We compute the preimage $Pre(\ref{eq:bulk1},K)$

\begin{equation}
\small
\exists \underline{d}\, \exists\underline{e} \begin{pmatrix}
\text{Diff}(\underline{e})\;\wedge	\;\pi(\ux_1, \ux_2)\;\wedge\; \psi(\ux_1, \underline{d})
        \;\wedge \psi^{\prime}(\ux_1, \underline{d})\;
        \wedge (\ux_1^{\prime}:=\ux_1\;\wedge \;\ux_2^{\prime}:=\underline{d}\; \wedge\;  \underline{s}^{\prime}:=\underline{s}\;\wedge \\ \bigwedge_{e_r \in \underline{e}} \phi_{e_r}( \underline{r}^{\prime}[e_r])
        \;  \wedge\;\Theta\;\wedge\; r_1^{\prime}:=r_1\; \wedge...\wedge\;r_i^{\prime}:=\lambda j.(\mathtt{ if }\  \kappa_1(\underline{r}[j]) \mathtt{~ then ~} c_1 \mathtt{~else ~} r_i[j]))\;\wedge... \wedge\; r_n^{\prime}:=r_n)\end{pmatrix}
\end{equation}
 which can be rewritten as a disjunction of the following formulae indexed by a function $f$ that associates to every $e_r$ a boolean value  in ${0,1}$:

 \begin{equation}
 \exists \underline{d},\,\exists\underline{e} \begin{pmatrix}
 \text{Diff}(\underline{e})\;\wedge	\;\pi(\ux_1, \ux_2)\;\wedge\; \psi(\ux_1, \underline{d})
 \;\wedge \psi^{\prime}(\ux_1, \underline{d})\;\wedge\;\\
 \bigwedge_{e_r \in \underline{e}} (\epsilon_f(e_r) \kappa_1(\underline{r}[e_r])\;  \wedge\; \phi(r_1[e_r],...\delta_f(e_r),...,r_n[e_r]))\;\wedge\;\Theta\;\end{pmatrix}
 \end{equation}
   where  $\epsilon_f(e_r):=\neg$ if $f(e_r)=0$, otherwise   $\epsilon_f(e_r):=\emptyset$, and $\delta_f(e_r):= c_1$ if $f(e_r)=0$, otherwise $\delta_f(e_r):=r_i[e_r]$.

  We can conclude as above (cf. propagation updates), by eliminating the existentially quantified variable $\underline{d}$, that this formula is strongly local.

  Notice that the previous arguments remain the same if $r_i^{\prime}:=\lambda j.(\mathtt{ if }\  \kappa_1(\underline{r}[j]) \mathtt{~ then ~} c_1 \mathtt{~else ~} r_i[j]))$ in Formula~\eqref{eq:bulk} is replaced by $r_i^{\prime}:=\lambda j.(\mathtt{ if }\  \kappa_1(\underline{r}[j]) \mathtt{~ then ~} c_1 \mathtt{~else ~} c_2)$, with $c_2$ a constant. Even in this case, the modified bulk transition is strongly local.

Analogous arguments show that transitions involving more than one artifact relations which are updated like $r_i$ are also strongly local.

The transition introduced in Example~\ref{ex:hr-short}

\[
\small
  \begin{array}{@{}l@{}}
    \pstatev = \final \land \pstatev'=\notified\\
    \recstate=\nullv \land  \recstate'=\nullv
    {}\land \appres' =    \smtlambda{j}{
                              \smtif{c}{\appscore[j] > \constant{80}}
                            {\win}{\los}
                          }
  \end{array}
\]
is a bulk update transition in the format described in this subsection, so it is a strongly local transition.

\section{Experiments}
\label{app_experiments}

\begin{table}
  \centering
  \caption{Summary of the experimental examples}
  \label{tab:benchmark}
        \begin{tabular}{llrrr}\toprule
       \multicolumn{2}{l}{\textbf{Example}} & \textbf{\#AC} & \textbf{\#AV} &\textbf{\#T} \\\midrule
            E1& JobHiring &9 & 18 & 15\\
            E2& Acquisition-following-RFQ &6 &13 &28\\
            E3& Book-Writing-and-Publishing & 4 & 14 & 13\\
            E4& Customer-Quotation-Request & 9 & 11 & 21\\
            E5& Patient-Treatment-Collaboration & 6 & 17 & 34\\
            E6& Property-and-Casualty-Insurance-Claim-Processing & 2 & 7 &15\\
            E7& Amazon-Fulfillment &  2 & 28 & 38\\
            E8& Incident-Management-as-Collaboration & 3 & 20 & 19
            \\\bottomrule
        \end{tabular}
    \end{table}

We base our experimental evaluation on the already existing benchmark provided in \cite{verifas}, that samples 32 real-world BPMN workflows published at the official BPM website (\url{http://www.bpmn.org/}).
Specifically, inspired by the specification approach adopted by the authors of \cite{verifas} in their experimental setup (\url{https://github.com/oi02lyl/has-verifier}), we select seven examples of varying complexity (see Table~\ref{tab:benchmark}) and provide their faithful encoding\footnote{Our encoding considers semantics of the framework studied in \cite{verifas}.} in the array-based specification using MCMT Version~2.8 (\url{http://users.mat.unimi.it/users/ghilardi/mcmt/}). Moreover, we enrich our experimental set with an extended version of the running example from Appendix~\ref{sec:hiring-example}. Each example has been checked against at least one safe and one unsafe conditions.
Experiments were performed on a machine with Ubuntu~16.04, 2.6\,GHz Intel
Core~i7 and 16\,GB RAM.


\begin{table}
  \centering
  \caption{Experimental results for safety properties}
  \label{tab:exp_extended}
        \begin{tabular}{cccrrrr}\toprule
            \textbf{Example} & \textbf{Property}&\textbf{Result}&\textbf{Time} &\textbf{\#N}  & \textbf{depth} &\textbf{\#SMT-calls}
            \\\midrule
            E1& E1P1 & \safe & 0.06 & 3 &  3 & 1238\\
             & E1P2& \unsafe & 0.36 & 46 & 10 & 2371\\
             & E1P3 & \unsafe &0.50 & 62 & 11 & 2867 \\
              & E1P4 & \unsafe & 0.35 & 42 & 10 & 2237 \\
            E2 & E2P1&  \safe & 0.72 & 50 & 9 & 3156\\
              & E2P2 &  \unsafe & 0.88 & 87 & 10 & 4238\\
              & E2P3& \unsafe &1.01 & 92 & 9 & 4811\\
             & E2P4 &  \unsafe & 0.83  & 80 & 9 & 4254\\
            E3 &  E3P1&  \safe & 0.05 & 1 & 1 & 700\\
             &E3P2 &  \unsafe & 0.06 & 14 & 3 & 899\\
            E4  & E4P1&  \safe &0.12 & 14 & 6 & 1460\\
              & E4P2&  \unsafe & 0.13 & 18 & 8 & 1525\\
           E5  & E5P1&  \safe &4.11 & 57 & 9 & 5618\\
              & E5P2&  \unsafe & 0.17& 13 & 3 & 2806\\
            E6   & E6P1& \safe & 0.04& 7 & 4 & 512\\
               & E6P2&  \unsafe &0.08 & 28 & 10 & 902\\
            E7  & E7P1 &  \safe & 1.00& 43 & 7 & 5281\\
             & E7P2&  \unsafe & 0.20 & 7 & 4 & 3412\\
            E8  & E8P1&  \safe & 0.70& 77 & 11 & 3720\\
              & E8P2 &  \unsafe & 0.15 & 25 & 7 & 1652
            \\\bottomrule
        \end{tabular}
\end{table}

%

Here \textbf{\#AV}, \textbf{\#AC} and \textbf{\#T} represent, respectively, the number of artifact variables, artifact components and transitions used in the example specification, while \textbf{Time} is the \mcmt execution time. The most critical measures are \textbf{\#N}, \textbf{depth} and \textbf{\#SMT-calls} that respectively define the number of nodes and the depth of the tree used for the backward reachability procedure adopted by \mcmt, and the number of the SMT-solver calls.
 Indeed, \mcmt  computes the iterated preimages of the formula describing the unsafe states along the various transitions. Such computation produces a tree, whose nodes are labelled by formulae describing sets of states that can reach an unsafe state and whose arcs are labelled by a transition.
In other words, an arc $t:\phi \rightarrow \psi$ means that $\phi$ is equal to $Pre(t,\psi)$. The tool applies forward and backward simplification strategies, so that whenever a node $\phi$ is deleted, this means that $\phi$ entails the disjunction of the remaining (non deleted) nodes. All nodes (both deleted and undeleted) can be visualized via the available online options (it is also possible to produce a Latex file containing their detailed description)

To stress test our encoding, we came up with a few formulae describing unsafe configurations (sets of ``bad'' states), that is, the configurations that the system should not incur throughout its execution.
\textbf{Property} references encodings of examples endowed with specific (un)safety properties done in \mcmt, whereas \textbf{Result} shows their verification outcome that can be of the two following types: \safe and \unsafe.
The \mcmt tool returns \safe, if the undesirable property it was asked to verify represents a configuration that the system cannot reach.
At the same time, the result is \unsafe  if there exists a path of the system execution that reaches ``bad'' states.
One can see, for example, that the job hiring \ras has been proved by \mcmt to be \safe w.r.t. the property defined in Example~\ref{ex:hr-prop}. The details about the successfully completed verification task can be seen in the first row of Table~\ref{tab:exp_extended}: the tool constructed a tree with 3 nodes and a depth of 3, and returned \safe in 0.06 seconds.
For the same job hiring \ras, if we slightly modify the safe condition discussed in Example~\ref{ex:hr-prop} by removing, for instance,
the check that a selected applicant is not a winning one, we obtain a description (see below) of a configuration in which it is still the case that an applicant could win:
\[
  \begin{array}{@{}l@{}}
   \exists  \typedvar{i}{\appidx}
    \left(  \begin{array}{@{}l@{}} \pstatev=\notified \land \appuser[i]\neq \nullv
    {}\land  \appres[i]\neq \los
    \end{array}
    \right)
  \end{array}
\]
In this case, the job hiring process analyzed against the devised property is evaluated as \unsafe by the tool (see E1P3 row in Table~\ref{tab:exp_extended}). When checking safety properties, \mcmt also allows to access an unsafe path of a given example in case
the verification result is \unsafe.

To conclude, we would like to point out that seemingly high number of SMT solver calls in \textbf{\#SMT-calls} against relatively
small execution time demonstrates that \mcmt could be considered as a promising tool supporting the presented line of research.
This is due to the following two reasons. On the one hand, the SMT technology underlying  solvers like \textsc{Yices} \cite{Dutertre} is quite mature and impressively well-performing. On the other hand, the backward reachability algorithm generates proof obligations which are relatively easy
to be analyzed as (un)satisfiable by the solver.


\end{document}